\documentclass{extarticle}
\usepackage{geometry}
\usepackage{titlesec}
\titleformat{\section}
  {\normalfont\fontsize{13}{13}\bfseries}{\thesection}{1em}{}
\titleformat{\subsection}
  {\normalfont\fontsize{11}{11}\bfseries}{\thesubsection}{1em}{}

\renewcommand\abstract{%
}

\usepackage{times}
\usepackage{soul}
\usepackage{url}
\usepackage[hidelinks]{hyperref}
\usepackage[utf8]{inputenc}
\usepackage[small]{caption}
\usepackage{graphicx}
\usepackage{amsmath}
\usepackage{amsthm}
\usepackage{booktabs}
\usepackage{algorithm}
\usepackage{algorithmic}
\usepackage[switch]{lineno}

\usepackage[T1]{fontenc}

\usepackage{amssymb}

\usepackage{caption}
\usepackage{subcaption}

\usepackage{xspace}

\usepackage{nicefrac}
\usepackage{bbm}
\usepackage{dsfont}
\newcommand{\1}{\mathds{1}}

\DeclareMathOperator{\poly}{poly}
\DeclareMathOperator{\ind}{\mathbbm{1}}

\usepackage{natbib}
\newcommand{\citespecial}[2]{\citeauthor{#1}~[\citeyear{#1},~#2]}

\newcommand{\np}{{{\mathrm{NP}}}}
\newcommand{\fpt}{{{\mathrm{FPT}}}}
\newcommand{\pref}{\succ}

\newcommand{\kmedian}{\textsc{$k$-Median}\xspace}
\newcommand{\kkemeny}{\textsc{$k$-Kemeny}\xspace}
\newcommand{\kkemenyamongvotes}{\textsc{$k$-Kemeny Among Votes}\xspace}
\newcommand{\maxkcover}{\textsc{Max $K$-Cover}\xspace}

\newcommand{\normphi}{{{\mathrm{norm}\hbox{-}\phi}}}

\newcommand{\swap}{\mathrm{swap}}
\newcommand{\ID}{ID\xspace}
\newcommand{\AN}{AN\xspace}
\newcommand{\UN}{UN\xspace}

\newcommand{\appUN}{\mathrm{UN}^*\xspace}
\newcommand{\appST}{\mathrm{ST}^*\xspace}

\newcommand{\calR}{{\mathcal{R}}}
\newcommand{\calS}{{\mathcal{S}}}
\newcommand{\calT}{{\mathcal{T}}}

\newcommand{\arxiv}[1]{#1}
\newcommand{\cameraready}[1]{}

\usepackage{ulem}
\newtheorem{theorem}{Theorem}
\newtheorem{definition}{Definition}
\newtheorem{corollary}{Corollary}
\newtheorem{proposition}{Proposition}

\title{Diversity, Agreement, and Polarization in Elections}
\author {
    Piotr Faliszewski,\textsuperscript{\rm 1}
    Andrzej Kaczmarczyk,\textsuperscript{\rm 1}
    Krzysztof Sornat,\textsuperscript{\rm 2}\\
    Stanisław Szufa,\textsuperscript{\rm 1}
    Tomasz Wąs\textsuperscript{\rm 3}\vspace{0.3cm}\\
    \textsuperscript{\rm 1} AGH University, Poland\\
    \textsuperscript{\rm 2} IDSIA, USI-SUPSI, Switzerland\\
    \textsuperscript{\rm 3} Pennsylvania State University, PA, USA\vspace{0.3cm}\\
    {\small faliszew@agh.edu.pl, andrzej.kaczmarczyk@agh.edu.pl, krzysztof.sornat@idsia.ch,}\\
    {\small szufa@agh.edu.pl, twas@psu.edu}
}
\date{}

\begin{document}
\maketitle

\begin{abstract}
  We consider the notions of agreement, diversity, and polarization
  in ordinal elections (that is, in elections where voters rank the candidates).
  While (computational) social choice offers good measures of agreement
  between the voters, such measures for the other two notions are lacking.
  We attempt to rectify this issue by designing appropriate measures, providing
  means of their (approximate) computation, and arguing that they, indeed, capture
  diversity and polarization well. 
  In particular, we present ``maps of preference orders'' that
  highlight relations between the votes in a given election and which help in making arguments 
  about their nature.
\end{abstract}

\section{Introduction}

The notions of \emph{agreement}, \emph{diversity}, and
\emph{polarization} of a society with respect to some issue are
intuitively quite clear. In case of agreement, most members of the
society have very similar views regarding the issue, in case of
diversity there is a whole spectrum of opinions, and in case of
polarization there are two opposing camps with conflicting views and
with few people taking middle-ground positions (more generally, if
there are several camps, with clearly separated views, then we speak
of \emph{fragmentation}; see, for example, the collection of
\citet{dyn-tie:b:social-sciences}). We study these three notions for
the case of ordinal elections---that is, for elections where each voter
has a preference order (his or her vote) ranking the candidates from
the most to the least appealing one---and analyze ways of quantifying
them.%
\cameraready{\footnote{
More extensive analysis is available in the full version of our paper~\citep{arxiv_version}
and the code of our experiments is available at
\url{https://github.com/Project-PRAGMA/diversity-agreement-polarization-IJCAI23}.}}%
\arxiv{\footnote{
A conference version of this work appears in IJCAI 2023~\citep{ijcai-version}
and the code of our experiments is available at
\url{https://github.com/Project-PRAGMA/diversity-agreement-polarization-IJCAI23}.
}}

Interestingly, even though agreement, diversity, and polarization seem
rather fundamental concepts for understanding the state of a given society
(see, for example, the papers in a special issue edited
by~\citet{lev-mil-perr:dynamics-of-polit-polar}), so far
(computational) social choice mostly focused on
the agreement-disagreement spectrum.
Let us consider the following notion:
\begin{itemize}
\item[]%
  Given an election, the voters' agreement index for candidates $a$
  and $b$ is the absolute value of the difference between the fraction
  of the voters who prefer $a$ to $b$ and the fraction of those with
  the opposite view. Hence, if all voters rank $a$ over $b$ (or, all
  voters rank $b$ over $a$) then the agreement index for these
  candidates is equal to~$1$. On the other hand, if half of the voters
  report $a \pref b$ and half of them report $b \pref a$, then the
  index is equal to $0$. The agreement index of the whole election is
  the average over the agreement indices of all the candidate pairs.
\end{itemize}
For an election $E$, we denote its agreement index as $A(E)$.
\citet{alc:vor:j:cohesiveness} viewed this index as measuring voter
cohesiveness---which is simply a different term for voter
agreement---and provided its axiomatic characterization.
\citet{has-end:c:diversity-indices} focused on measuring diversity and
provided axiomatic and experimental analyses of a number of election
indices, including $1-A(E)$.
$1-A(E)$ was also characterized axiomatically by
\citet{can-ozk-sto:polarization}, who saw it as measuring
polarization; their point of view was that for each pair of candidates
one can measure polarization independently. (In
Section~\ref{sec:prelim} we briefly discuss other election indices
from the literature; generally, they are strongly interrelated with the
agreement one).

Our view is that $1-A(E)$ is neither a measure of diversity nor of
polarization, but of disagreement. Indeed, it has the same, highest
possible, value on both the antagonism election (AN), where half of
the voters report one preference order and the other half reports 
the
opposite one, and on the uniformity election (UN), where each possible
preference order occurs the same number of times.
Indeed, both these elections arguably represent extreme cases of
disagreement.
Yet, the nature of this disagreement is very different. In the former,
we see strong polarization, with the voters taking one of the two
opposing positions, and in the latter we see perfect diversity of
opinion.
The fundamental difference between these notions becomes clear in
the text of~\citet{lev-mil-perr:dynamics-of-polit-polar} which highlights
``the loss of diversity that extreme polarization creates'' as a central
theme of the related special issue.  Our main goal is to design
election indices that distinguish these notions.

Our new indices are based on what we call the $k$-\textsc{Kemeny}
problem. In the classic \textsc{Kemeny Ranking} problem (equivalent to
$1$-\textsc{Kemeny}), given an election we ask for a ranking whose sum
of swap distances to the votes is the smallest (a swap distance between
two rankings is the number of swaps of adjacent candidates needed to
transform one ranking into the other). The $k$-\textsc{Kemeny} problem
is defined analogously, but we ask for $k$ rankings that minimize the
sum of each vote's distance to the closest one (readers familiar with
multiwinner elections~\citep{fal-sko-sli-tal:b:multiwinner-voting} may
think of it as the Chamberlin--Courant rule~\citep{cha-cou:j:cc} for
committees of rankings rather than candidates). We refer to this value
as the $k$-Kemeny distance.
Unfortunately, the $k$-\textsc{Kemeny} problem is intractable---just like
\textsc{Kemeny
  Ranking}~\citep{bar-tov-tri:j:who-won,hem-spa-vog:j:kemeny}---so we
develop multiple ways (such as fast approximation algorithms) to
circumvent this issue.

Our polarization index is a normalized difference between the $1$-Kemeny
and $2$-Kemeny distances of an election, and our diversity index is a
weighted sum of the $k$-Kemeny distances for $k = 1, 2, 3,
\ldots$. The intuition for the former is that if a society is
completely polarized (that is, partitioned into two equal-sized groups
with opposing preference orders), then $1$-Kemeny distance is the
largest possible, but $2$-Kemeny distance is zero. The intuition for
the latter is that if a society is fully diverse (consists of all
possible votes) then each $k$-Kemeny distance is non-negligible (we use
weights for technical reasons). Since our agreement index can also be
seen as a variant of the \textsc{Kemeny Ranking} problem, where we
measure the distance to the majority relation, all these indices are
based on similar principles.

To evaluate our indices, we use the ``map of elections'' framework of
\citet{szu-fal-sko-sli-tal:c:map},
\citet{boe-bre-fal-nie-szu:c:compass}, and~\citet{boe-fal-nie-szu-was:c:map-measures},
applied to a dataset of randomly generated elections. In particular, we find
that our indices are correlated with the distances from several
characteristic points on the map and, hence, provide the map with a
semantic meaning.
Additionally, we develop a new form of a map that visualizes the
relations between the votes of a single election (the original maps
visualized relations between several elections from a given
dataset). We use this approach to get an insight regarding the
statistical cultures used to generate our dataset and to validate
intuitions regarding the agreement, diversity, and polarization of its
elections.
In our experiments, we focused on elections
with a relatively small number of candidates
(8 candidates and 96 voters).
While we believe that our main conclusions
extend to all sizes of elections,
it would be valuable to check this
(however, this would require
quite extensive computation that,
currently, is beyond our reach).

\section{Preliminaries}\label{sec:prelim}

For every number $k \in \mathbb{N}$, by $[k]$ we understand the set $\{1,\dots,k\}$.
For two sets $A$ and $B$ such that $|A|=|B|$,
by $\Pi(A,B)$ we mean the set of all bijections from $A$ to $B$.

\subsubsection*{Elections}  
An \emph{election} $E = (C,V)$ is a pair, 
where $C$ is a set of \emph{candidates} and $V$ is a collection of
\emph{voters} whose preferences (or, \emph{votes}) are represented as
linear orders over $C$ (we use the terms \emph{vote} and \emph{voter}
interchangeably, depending on the context).
For a vote~$v$, we write
$a \pref_v b$ (or, equivalently, $v \colon a \pref b$) to indicate
that $v$ \emph{prefers} candidate~$a$ over candidate~$b$. We also
extend this notation to more candidates. For example, for candidate
set $C = \{a,b,c\}$ by $v \colon a \pref b \pref c$ we mean that $v$
ranks $a$ first, $b$ second, and $c$~third.  For two candidates $a$
and $b$ from election $E$, by $p_E(a,b)$ we denote the fraction of
voters in $E$ that prefer $a$ over $b$.

We will often speak of the following
three characteristic elections, introduced by~\citet{boe-bre-fal-nie-szu:c:compass}
as ``compass elections'' 
(we assume candidate set $C = \{c_1, \ldots, c_m\}$ here;
\citet{boe-bre-fal-nie-szu:c:compass} also considered the fourth
election, i.e., \emph{stratification}, but it will not play an important role for
us):

\begin{description}
\item[Identity (\ID).] In an identity election all votes are
  identical. We view this election as being in perfect agreement.

\item[Antagonism (\AN).] In an antagonism election, exactly half of the voters
  have one preference order (for example,
  $c_1 \pref c_2 \pref \cdots \pref c_m)$ and the other half has the
  reversed one ($c_m \pref c_{m-1} \pref \cdots \pref c_1$).  We view
  this election as being perfectly polarized.

\item[Uniformity (\UN).] A uniformity election contains the same
  number of copies of every possible preference order. We view this
  election as being perfectly diverse.
\end{description}

\subsubsection*{Kemeny Rankings and Swap Distance}
For two votes $u$ and $v$ over a candidate set $C$, by $\swap(u,v)$ we
mean their \emph{swap distance}, that is, the minimal number of swaps of
consecutive candidates required to transform $u$ into~$v$. This value
is also known as \emph{Kendall's $\tau$} distance and is equal to the
number of candidate pairs $a,b \in C$ such that $a \pref_u b$ but
$b \pref_v a$.
A \emph{Kemeny ranking} of an election $E=(C,V)$ is a linear order
over $C$ that minimizes the sum of its swap distances to the votes
from $V$~\citep{kem:j:no-numbers}.  It is well known that computing a
Kemeny ranking is $\np$-hard \citep{bar-tov-tri:j:who-won} and, more
precisely, $\Theta_2^p$-complete~\citep{hem-spa-vog:j:kemeny}.

For two elections, $E = (C,V)$ and $F = (D,U)$, such that $|C|=|D|$,
$V = (v_1,\dots,v_n)$, and $U = (u_1,\dots,u_n)$, 
by $d_\swap(E,F)$ we denote their \emph{isomorphic swap
distance}~\citep{fal-sko-sli-szu-tal:c:isomorphism}, that is, the
(minimal) sum of swap distances between the votes in both elections,
given by optimal correspondences between their candidates and their
voters. Formally:
\[
  d_{\swap}(E, F)\hspace{-2pt} = \hspace{-12pt}\min_{\sigma \in \Pi([n],[n])}\min_{\pi \in \Pi(C, D)}
    \textstyle\sum_{i = 1}^n \swap( \pi(v_i), u_{\sigma(i)} ),
\]
where by $\pi(v_i)$ we denote vote $v_i$ with every candidate
$c \in C$ replaced by candidate $\pi(c)$.

\subsubsection*{Maps of Elections}

A \emph{map of elections} is a collection of elections
represented
on a 2D plane as points, so that the Euclidean distances between the
points reflect the similarity between the elections (the closer two
points are, the more similar should their elections be).  Maps of
elections were introduced by \citet{szu-fal-sko-sli-tal:c:map} (together with an
open-source Python library \emph{mapel}, which we use and build on) and
\citet{boe-bre-fal-nie-szu:c:compass}, who used the
distance based on position matrices of elections
as a measure of similarity.
We use the isomorphic swap distance instead. Indeed,
\citet{szu-fal-sko-sli-tal:c:map} and
\citet{boe-bre-fal-nie-szu:c:compass} admitted that isomorphic swap
distance would be more accurate but avoided it because it is hard to
compute (\citet{boe-fal-nie-szu-was:c:map-measures} analyzed the
consequences of using various distances). We are able to use the swap
distance because we focus on small candidate sets.
To present a set of elections as a map, we compute the distance
between each two elections and then run the multidimensional scaling
algorithm (MDS)\arxiv{\footnote{We use Python implementation from \emph{sklearn.manifold.MDS}.}} to find an embedding of points on a plane that
reflects the computed distances.  For an example of a map, see
Fig.~\ref{fig:swap-map:standard} at the end of the paper;
we describe its elections in Section~\ref{sec:cultures}.

\subsubsection*{Agreement and Other Election Indices}

\emph{Election index} is a function that given an election outputs a
real number.  The next index is among the most studied ones and
captures voter agreement. 

\begin{definition}
The \emph{agreement index} of an election $E=(C,V)$ is:
\begin{equation*}
\textstyle    A(E) = \left( \sum_{\{a,b\} \subseteq C} | p_E(a,b) - p_E(b,a) | \right) \Big/ \textstyle \binom{|C|}{2}.
\end{equation*}
\end{definition}
\noindent The agreement index takes values between $0$ and $1$,
where~$0$ means perfect disagreement and $1$ means perfect agreement.
Indeed, we have $A(\mathrm{ID}) = 1$ and
$A(\mathrm{UN}) = A(\mathrm{AN}) = 0$.

There is also a number of other election indices in the literature.
Somewhat disappointingly, they mostly fall into one or more of the
following categories: (1)~They are generalizations of the agreement
index (or its linear
transformation)~\citep{alc:vor:j:cohesiveness2,can-ozk-sto:polarization2};
(2)~They are highly correlated with the agreement index (at least on
  our datasets)~\citep{has-end:c:diversity-indices,kar:j:preference-diversity,alc-and-cas:consensus};
(3)~Their values come from a small set,
limiting their
  expressiveness and robustness~\citep{bos:consensus,has-end:c:diversity-indices}.

\section{Diversity and Polarization Indices}
In this section, we introduce our two new election indices, designed
to measure the levels of diversity and polarization in elections. Both of them
are defined on top of a generalization of the \textsc{Kemeny ranking}
problem (note that this generalization is quite different from that
studied by \citet{arr-fer-lok-oli-wol:c:diverse-kemeny} under a
related name).

\begin{definition}\label{def:k-kemeny-rankings}
  \emph{$k$-Kemeny rankings} of election $E=(C,V)$ are the elements of a set
  $\Lambda=\{\lambda_1,\dots,\lambda_k\}$ of $k$ linear orders over
  $C$ that minimize:
\[
    \textstyle \sum_{v \in V} \min_{i \in [k]} \swap(v, \lambda_i).
\]
The \emph{$k$-Kemeny distance},
$\kappa_k(E)$, is equal to
this minimum.
\end{definition}

We can think of finding $k$-Kemeny rankings as finding an optimal
split of votes into $k$ groups and minimizing the sum of each group's
distance to its Kemeny ranking.  Hence, $1$-Kemeny distance is simply
the distance of the voters from the (standard) Kemeny ranking. 
We will later argue that $\kappa_1(E)$ is closely related to the
agreement index. 

We want our diversity index to be high for \UN, but small for \AN and
\ID. For identity, $1$-Kemeny distance is equal to zero, but for
both \UN and \AN, $1$-Kemeny distance is equal to
$|V| \cdot \binom{|C|}{2} / 2$, which is the maximal possible value (as shown, for example,
by~\citet{boe-fal-nie-szu-was:c:map-measures}).
However, for $k \ge 2$ we observe a sharp difference between
$k$-Kemeny distances in these two elections. For \AN, we get distance
zero (it suffices to use the two opposing votes as the $k$-Kemeny
rankings),
and for \UN we get non-negligible positive distances (as long as $k$ is
smaller than the number of possible votes).
Motivated by this, we define the \emph{diversity} index as a
normalized sum of all $k$-Kemeny distances.
\begin{definition}
The \emph{diversity index} of an election $E=(C,V)$ is:
\begin{equation*}
\textstyle D(E) = \left( \sum_{k \in [|V|]} \kappa_k(E) / k \right) \Big/ \left( |V| \cdot \textstyle \binom{|C|}{2} \right).
\end{equation*}
\end{definition}
The sum in the definition is divided by the number of voters and the
maximal possible distance $\binom{|C|}{2}$ between two votes.
As a result, the values of the index are more consistent across
elections with different number of voters and candidates (for example,
diversity of \AN is always equal to $\nicefrac{1}{2}$).  Apart from
that, in the sum, each $k$-Kemeny distance is divided by $k$.  This
way, the values for large $k$ have lesser impact on the total value, and
it also improves scalability.  However, we note that even with
this division, diversity of \UN seems to grow slightly faster than
linearly with the growing number of candidates and there is a
significant gap between the value for \UN with all $m!$ possible votes and
even the most diverse election with significantly smaller number of voters.
The currently
defined diversity index works well on our datasets (see
Section~\ref{sec:results}), but finding a more robust normalization is
desirable (the obvious idea of dividing by the highest possible value of the sum
is challenging to implement
and does not prevent the vulnerability to changes in the voters count).

To construct the polarization index, we look at \AN and take advantage
of
the sudden drop from the maximal possible value of the $1$-Kemeny distance to
zero for the $2$-Kemeny distance. We view this drop as characteristic for polarized
elections because they include two opposing, but coherent, factions.
Consequently, we have the following definition (we divide by
$|V| \cdot \binom{C}{2}/2$ for normalization; the index takes
values between $0$, for the lowest polarization, and $1$, for the highest).

\begin{definition}
The \emph{polarization index} of an election $E = (C, V )$ is:
\begin{equation*}
\textstyle    P(E) = 2 \left( \kappa_1(E) - \kappa_2(E) \right) \Big/ \left( |V| \cdot \textstyle \binom{|C|}{2} \right).
\end{equation*}
\end{definition}

For \AN polarization is
one, while for \ID
it is
zero. For
\UN with 8~candidates, it is $0.232$. This is intuitive as in \UN
every vote
also
has its reverse. However, we have experimentally
checked that with a growing number of candidates the polarization of
\UN seems to
approach zero (e.g.,\ it is
$0.13$, $0.054$, and $0.024$ for, respectively, 20, 100, and 500~candidates).

We note that there is extensive literature on polarization measures in different settings, such us regarding distributions over continuous intervals~\citep{EstRay-1994-Polarization} or regarding opinions on networks~\citep{MusMusTso-2018-Polarization,HurOzk-2022-Polarization,TuNeu-2022-Polarization,ZhuZha-2022-Polarization} (we only mention a few example references).
On the technical level, this literature is quite different from our setting,
but finding meta connections could be very inspiring.

Concluding our discussion of the election indices, 
we note a connection between the agreement index and the $1$-Kemeny
distance.  Let $\mu$ be the \emph{majority relation} of an election
$E = (C,V)$, that is, a relation such that for candidates $a, b \in C$,
$a \succeq_\mu b$ if and only if $p_E(a,b) \ge p_E(b,a)$.  If $E$ does
not have a \emph{Condorcet cycle}, that is, there is no
cycle within $\mu$, 
then $\mu$ is identical to the Kemeny ranking.
As noted by~\citet{can-ozk-sto:polarization},
the agreement index can be expressed as a linear transformation of the
sum of the swap distances from all the votes to $\mu$
(we also formally prove it in Appendix~\ref{app:agr-1kemeny}).
Hence, if there is no Condorcet cycle,
the agreement index is strictly linked to
$\kappa_1(E)$ and all three of our indices are related.%

\section[Computation of k-Kemeny Distance]{Computation of $\boldsymbol{k}$-Kemeny Distance}
\label{sec:multikemeny}
We define an optimization problem \kkemeny
in which the goal is to find the $k$-Kemeny distance of a given election (see Definition~\ref{def:k-kemeny-rankings}).
In a decision variant of \kkemeny, we check if the $k$-Kemeny distance is at most a given value.
We note that \kkemeny is NP-hard~\citep{bar-tov-tri:j:who-won},
even for $k=1$ and $n=4$~\citep{dwo-kum-nao-siv:c:rank-aggregation}.
Hence, we seek polynomial-time approximation algorithms.

\subsection{Approximation Algorithms}\label{sec:aprox-algs}

While there is a \arxiv{polynomial-time approximation scheme (PTAS)}\cameraready{PTAS} for {\sc $1$-Kemeny}~\citep{KenyonMathieuS07}, it
is not obvious how to approximate even {\sc $2$-Kemeny}.  Yet, we
observe that \kkemeny is related to the classic facility location
problem \kmedian~\citep{WilliamsonShmoys11}.
In this problem we are given a set of clients~$X$, a set of potential facility
locations $F$, a natural number $k$, and a metric $d$ defined over
$X \cup F$.  The goal is to find a subset
$W = \{f_1, f_2, \dots, f_k\}$ of facilities which minimizes \emph{the
total connection cost} of the clients, that is,
$\sum_{x \in X} \min_{f \in W} d(x,f)$. We see that \kkemeny is equivalent to
\kmedian in which the set of clients are the votes from the input
election, the set of facilities is the set of all possible votes, and
the metric is the swap distance.
Hence, to approximate \kkemeny we can use approximation algorithms
designed for \kmedian.  The issue is that there are $m!$ possible
Kemeny rankings and the algorithms for \kmedian run in polynomial
time with respect to the number of facilities so they would need
exponential time.

We tackle the above issue by reducing the search space from all
possible rankings to those appearing in the input.  We call this
problem \kkemenyamongvotes and provide the following result.
\footnote{We note
that the special case of Theorem~\ref{thm:2-apx-by-among-votes}
for $k=1$ and $\alpha=1$ was proved by \citet{EndGra-2014-VoteKemeny}.}
\begin{theorem}\label{thm:2-apx-by-among-votes}
  An $\alpha$-approximate solution for \kkemenyamongvotes is a $2\alpha$-approximate solution for \kkemeny.
\end{theorem}

This allows us to use the rich literature on approximation algorithms
for \kmedian~\citep{WilliamsonShmoys11}.  For example,
using the (currently best) $2.7$-approximation algorithms for
\kmedian~\citep{ByrkaPRST17,CohenAddadGLS23,GowdaPST2023} we get the following.
\begin{corollary}
  There is a polynomial-time $5.4$-approximation algorithm for \kkemeny.
\end{corollary}
The algorithms of \citet{ByrkaPRST17}, \citet{CohenAddadGLS23} and \citet{GowdaPST2023} are
based on a complex procedure for rounding a solution of a linear
program, which is difficult to implement.  Moreover, there are large
constants hidden in the running time.  Fortunately, there is a simple
local search algorithm for \kmedian which achieves
$(3+\frac{2}{p})$-approximation in time $|F|^p \cdot \poly(|F|,|X|)$,
where $p$ is the \emph{swap size} (as a basic building block, the
algorithm uses a swap operation which replaces $p$ centers with $p$
other ones, to locally minimize the connection
cost)~\citep{AryaGKMP01}.
\begin{corollary}\label{cor:kkem-local-search}
  There is a local search $(6+4/p)$-approximation algorithm for \kkemeny,
  where $p$ is the swap size.
\end{corollary}
We implemented the local search algorithm for $p=1$ and used it in
our experiments (see Section~\ref{sec:results}).  We note that there
is a recent result~\citep{CohenAddadGHOS22} which shows that the same
local search algorithm actually has an approximation ratio
$2.83+\epsilon$, but at the cost of an enormous swap size (hence also the
running time)---for example, for approximation ratio below $3$ one needs
swap size larger than $10^{10000}$.

In our experiments in Section~\ref{sec:results}, we also use a greedy
algorithm, which constructs a solution for \kkemenyamongvotes
iteratively: It starts with an empty set of rankings and then, in each
iteration, it adds a ranking (from those appearing among the votes)
that decreases the $k$-Kemeny distance most.  It is an open question
if this algorithm achieves a bounded approximation ratio.

We also point out that using the PTAS for {\sc $1$-Kemeny}, we can
obtain an approximation scheme in parameterized time for \kkemeny
(parameterized by the number of voters;
note that an exact parameterized algorithm is unlikely as
{\sc $1$-Kemeny} is already $\np$-hard for four
voters~\citep{dwo-kum-nao-siv:c:rank-aggregation}). The idea is to
guess the partition of the voters and solve {\sc $1$-Kemeny} for
each group.

\begin{theorem}\label{thm:fpt-as-n}
  For every $\epsilon > 0$, there is a $(1+\epsilon)$-approximation
  algorithm for \kkemeny which runs in time $\fpt$ w.r.t.~$n$.
\end{theorem}

All algorithms in this section, besides solving the decision
problem, also output the sought $k$-Kemeny rankings.

\subsection[Hardness of k-Kemeny Among Votes]{Hardness of $\boldsymbol{k}$-Kemeny Among Votes}\label{sec:hardness}
The reader may wonder why we use \kmedian algorithms
instead of solving \kkemenyamongvotes directly.
Unfortunately, even this restricted variant is intractable.
\begin{theorem}\label{thm:w2-kkemeny}
  \kkemenyamongvotes is $\np$-complete and $\mathrm{W}$[2]-hard when parameterized by~$k$.
\end{theorem}
\begin{proof}
  We give a 
  reduction from the \maxkcover problem (which is equivalent to the
  well-known Approval Chamberlin-Courant voting
  rule~\citep{ProcacciaRZ08}).  In \maxkcover we are given a set of
  elements $X = \{x_1, x_2, \dots, x_N\}$, a family
  $\calS = \{ S_1, S_2, \dots, S_M\}$ of nonempty, distinct subsets of
  $X$, and positive integers $K \leq M$ and~$T$.  The goal is to find
  $K$ subsets from $\calS$ which together cover at least $T$ elements
  from $X$.

  We take an instance $(X, \calS, K, T)$ of \maxkcover and construct
  an instance of \kkemenyamongvotes as follows.
  We create three \emph{pivot-candidates} $p_1$, $p_2$, and $p_3$.
  For every set $S \in \calS$, we create two \emph{set-candidates}
  $c_S$ and $d_S$ obtaining, in total, $m = 2M+3$~candidates.
  Next, we create the votes, each with the following \emph{vote
    structure}:
  \newcommand{\hso}{\hspace{-1pt}} 
  $$\{p_1, p_2, p_3\} \hso\pref\hso \{c_{S_1}, d_{S_1}\} \hso\pref\hso \{c_{S_2}, d_{S_2}\} \hso\pref\hso \dots \hso\pref\hso \{c_{S_M}, d_{S_M}\},
  $$
  where $\{c,d\}$ means that the order of candidates $c$ and
  $d$ is not specified. Hence, when defining a vote we will only
  specify the voter's preference on the unspecified pairs of
  candidates.

  For every set $S_j \in \calS$, we create
  $L=N(M+4)$~\emph{set-voters}~$v_j$ (we do not need to distinguish
  between these copies, hence we call any of them $v_j$) with the
  following specification over the vote structure:
\begin{align*}
  p_1 \pref_{v_j} p_2 \pref_{v_j} p_3;\quad
  d_{S_j} \pref_{v_j} c_{S_j};\quad c_S \pref_{v_j} d_S, \text{ for } S \neq S_j.
\end{align*}
For each two set-voters $u$ and $v$, $\swap(u,v) \in \{0,2\}$ and it
equals $0$ if and only if $u$ and $v$ come from the same set (our sets
are nonempty).

For every element $x_i \in X$, we create an \emph{element-voter} $e_i$ with the following specification over the vote structure:
\begin{align*}
  p_3 \pref_{e_i} p_2 \pref_{e_i} p_1;\quad
  d_S \pref_{e_i} c_S, \text{ for } e_i \in S;\\
  c_S \pref_{e_i} d_S, \text{ for } e_i \notin S.
\end{align*}
Note that for each element-voter $e_i$ and set voter $v_j$,
$\swap(e_i,v_j) \geq 3$.
In total we have $n = N(M^2+4M+1)$~voters.
We define $k=K$ and we ask if the $k$-Kemeny distance in
\kkemenyamongvotes is at most
$D = 2L(M-K) + \sum_{j \in [M]} |S_j| + 4N - 2T$.

The formal proof of correctness of the reduction is included in
Appendix~\ref{app:w2-kkemeny}.  We just notice that one direction
follows by taking $k$ set-voters corresponding to a solution for
\maxkcover.  The other one follows by observing that a solution to
\kkemenyamongvotes may contain only set-voters (because there are
$N(M+4)$ copies of each) and, hence, we can derive a corresponding
solution for \maxkcover.

In order to achieve the theorem statement we notice that \maxkcover is
W[2]-hard w.r.t.~$K$~\citep{CyganFKLMPPS15},\footnote{Actually, the
  result comes from W[2]-hardness of the {\sc Set Cover} problem and a
  folklore reduction to \maxkcover by setting $T=N$.} $k = K$, and the
reduction runs in polynomial time.
\end{proof}

Using the same reduction as in the proof of
Theorem~\ref{thm:w2-kkemeny}, we can provide more fine-grained hardness
results; they are presented in Appendix~\ref{app:propositions-thm3}.

\section{Statistical Cultures of Our Dataset}\label{sec:cultures}

Before we move on to our main experiments, we describe and analyze our
dataset.  It consists of 292 elections with 8 candidates and 96 voters
each, generated from several \emph{statistical cultures}, that is, models
of generating random elections (we describe its exact composition in
Appendix~\ref{app:assorted-dataset}).
For example, under \emph{impartial culture (IC)} each vote is drawn
uniformly at random from all possible votes
(thus, it closely resembles \UN).
We present our dataset as a map of elections on
Fig.~\ref{fig:swap-map:standard}.  In the appendix we consider also
two more datasets: extended dataset in which we include also elections
from additional statistical cultures not mentioned in this section
(Appendix~\ref{app:extended_dataset}); and Mallows dataset in which
the elections come from mixtures of two Mallows models
(Appendix~\ref{app:mallows_dataset}).

Below, we discuss each statistical culture used in our dataset and
build an intuition on how our indices should evaluate elections
generated from them.  To this end, we form a new type of a map, which
we call a \emph{map of preferences}, where we look at relations
between votes within a single election.  In other words, a map of
elections gives a bird's eye view of the space of elections, and a map
of preferences is a microscope view of a single election.

\subsection{Maps of Preferences}
To generate a map of preferences for a given election, we first
compute the (standard) swap distance between each pair
of its votes. Then, based on these distances, we create a map in the
same way as for maps of elections (that is, we use the multidimensional
scaling algorithm). We obtain a collection of points in 2D, where each
point corresponds to a vote in the election, and Euclidean distances
between the points resemble the swap distances between the votes they
represent.

For each model, we generated a \emph{single} election
and created its map of preferences.
The results are shown in Fig.~\ref{fig:microscope}.
The elections have $1000$ voters instead of $96$,
so that the pictures look similar
each time we draw an election from the model.
In Appendix~\ref{app:microscope_96},
we include the version with $96$ votes.

\begin{figure}
    \centering
    \includegraphics[width = \linewidth]{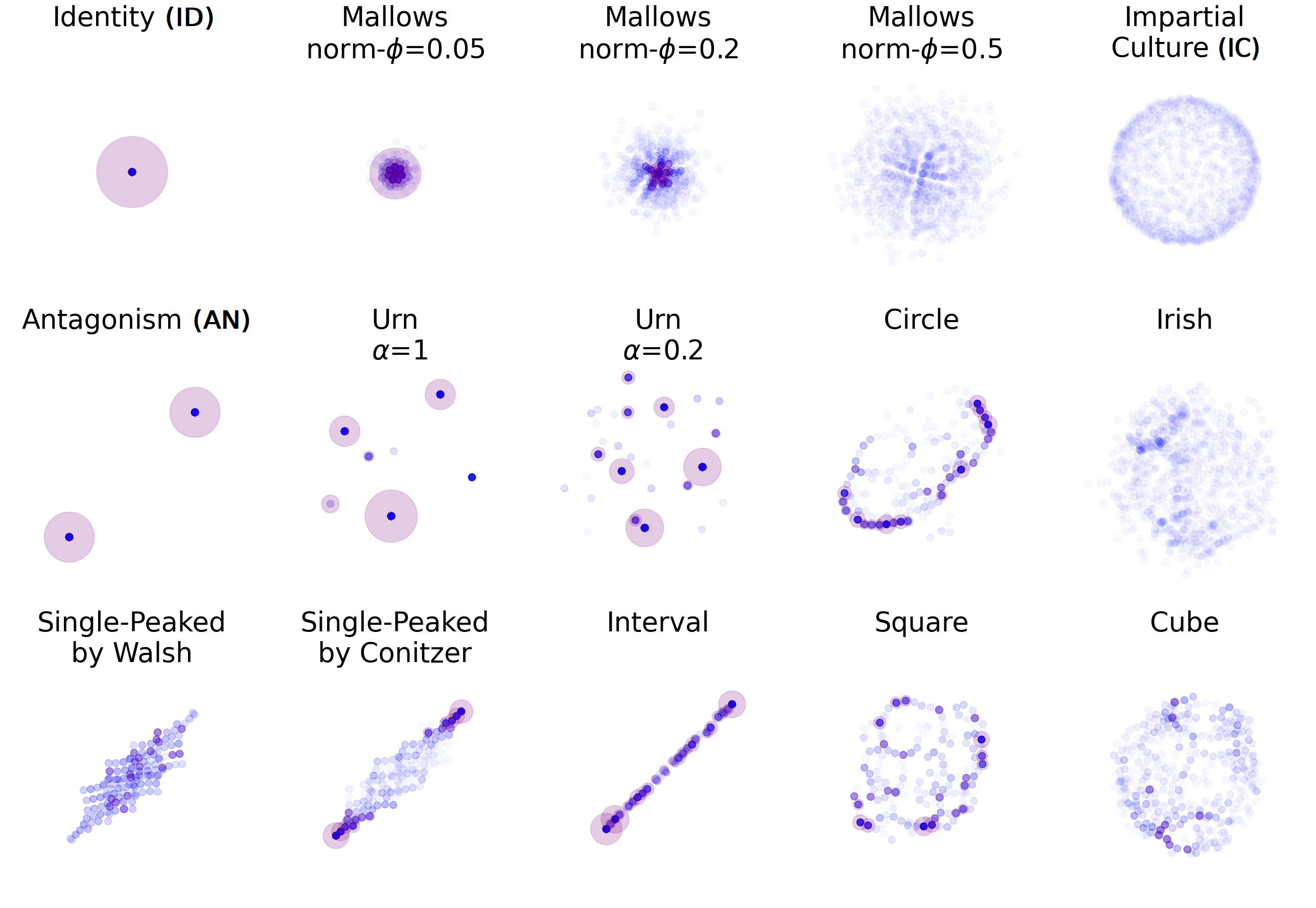}
    \caption{Maps of Preferences (8 candidates, 1000 voters).
    If there
      are more than 10 copies of the same vote, we add a purple disc
      with a radius proportional to the number of voters.
      }

    \label{fig:microscope}
\end{figure}

\subsection{Model Definitions and Analysis}

\subsubsection*{ID, AN, and IC}
We first consider \ID, \AN, and IC elections (which, for the time
being, covers for \UN).  \ID and \AN are shown as the first entries of
the first two rows in Fig.~\ref{fig:microscope}.
The former, with 1000 copies of the same vote, presented as a single
point with a large purple disc, embodies perfect agreement. The
latter, with 500 votes of one type and 500 its reverses, represents a
very polarized society, which is well captured by the two faraway
points with large discs on its map.
Under IC, whose map is the last one in the first row, we see no clear
structure except that, of course, there are many pairs of votes at
high swap distance (they form the higher-density rim). Yet, for each
such pair there are also many votes in between.  Hence, it is close to
being perfectly diverse.

We do not present \UN in our maps because it requires at least $m!$
votes. Indeed, from now on instead of considering \UN, we will talk
about its approximate variant, $\appUN$, 
which we generate by sampling votes from
its scaled position matrix
(see Appendix~\ref{app:assorted-dataset}
for details).

\subsubsection*{Mallows Model} 
The Mallows model is parameterized by the central vote $u$ and the
dispersion parameter~$\phi \in [0,1]$. Votes are generated
independently and the probability of generating a vote~$v$ is
proportional to $\phi^{\swap(u,v)}$. Instead of using the parameter
$\phi$ directly, we follow \citet{boe-bre-fal-nie-szu:c:compass} and
use its normalized variant, 
$\normphi \in [0,1]$, which is internally converted to~$\phi$ (see
their work for details; with $8$ candidates the conversion is nearly linear).
For $\normphi = 1$, the Mallows model is equivalent to IC, for
$\normphi = 0$ it is equivalent to \ID%
, and for values in between we get a smooth transition between these
extremes (or, between agreement and diversity, to use our high-level
notions).
We see this 
in the first row of Fig.~\ref{fig:microscope}.

\subsubsection*{Urn Model} In the Pólya-Eggenberger urn
model~\citep{ber:j:urn-model,mcc-sli:j:similarity-rules}, we have a
parameter of contagion $\alpha \in [0,\infty)$. We start with an urn
containing one copy of each possible vote and we repeat the following
process~$n$ times: We draw a vote from the urn, its copy is included
in the election, and the vote, together with $\alpha \cdot m!$ copies,
is returned to the urn. For $\alpha = 0$ the model is equivalent to
IC. The larger is the $\alpha$ value, the stronger is the correlation
between the votes.

In Fig.~\ref{fig:microscope}, urn elections (shown in the middle of the
second row) consist of very few distinct votes. For example, for
$\alpha = 1$ we only have seven votes, thus this election's map looks
similarly to that for \AN---few points with discs. Such elections,
with several popular views but without a spectrum of opinions in
between, are known as
\emph{fragmented}~\citep{dyn-tie:b:social-sciences}.  Hence, we expect
their diversity to be small.
As $\alpha$ decreases, urn elections become less fragmented.

We upper-bound the expected number of different votes in an urn
election with~$m$ candidates, $n$ voters (where $n$ is significantly
smaller than $m!$), and parameter $\alpha$ by
$\sum_{i=1}^{n} \nicefrac{1}{(1+(i-1)\alpha)}$ (the first vote is
always unique, the second one is drawn from the original $m!$ votes
from the urn with probability $\nicefrac{1}{(1+\alpha)}$, and so on;
if we draw one of the original votes from the urn it still might be
the same as one of the previous ones, but this happens with a small
probability when $n$ is significantly smaller than $m!$). For example,
for $n=1000$ and $\alpha$ equal to $1$, 
our formula gives $7.48$.
In the literature, authors often use
$\alpha =
1$~\citep{erd-fel-rot-sch:j:bucklin-fallback-exp,kel-has-haz:j:approx-manipulation,wal:j:hard-manipulation},
sometimes explicitly noting the strong correlations and modifying the
model~\citep{erd-fel-rot-sch:j:bucklin-fallback-exp}.  However, smaller
values of $\alpha$ also are
used~\citep{sko-fal-sli:j:multiwinner,mcc-sli:j:similarity-rules}.
Since $\alpha=1$ gives very particular elections, it should be used
consciously.

\subsubsection*{Single-Peaked Elections}
Single-peaked elections~\citep{bla:b:polsci:committees-elections}
capture scenarios where voters have a spectrum of opinions between two
extremes (like choosing a preferred temperature in a room).

\begin{definition}[\citet{bla:b:polsci:committees-elections}]
  Let $C$
  be a set of candidates and let $\rhd$ be an order over
  $C$, called the societal axis.
  A vote is single-peaked with respect to $\rhd$ if for each $t
  \in [|C|]$, its top $t$ candidates form an interval w.r.t.~$\rhd$.
  An election is single-peaked (w.r.t.~$\rhd$) if its votes are.
\end{definition}

We use the \emph{Walsh}~\citep{wal:t:generate-sp} and the
\emph{Conitzer} \emph{(random peak)}
models~\citep{con:j:eliciting-singlepeaked} of generating
single-peaked
elections. 
In the former, we fix the societal axis and choose votes single-peaked
with respect to it uniformly at random (so we can look at it as
IC over the single-peaked domain).
In the Conitzer model we also first fix the
axis, and then generate each vote as follows: We choose the top-ranked
candidate uniformly at random and fill-in the following positions by
choosing either the candidate directly to the left or directly to the
right of the already selected ones on the axis, with probability
$\nicefrac{1}{2}$
(at some point we run out of the candidates on one side and then
only use the other one).

In Fig.~\ref{fig:microscope}, Conitzer and Walsh elections are
similar, but the former one has more votes at large swap distance.
Indeed, under the Conitzer model, we generate a vote equal to the axis
(or its reverse) with probability $\nicefrac{2}{m}$, which for $m=8$
is $25\%$. Under the Walsh model, this happens with probability
$1.5\%$ (it is known there are $2^{m-1}$ different single-peaked
votes and Walsh model chooses each of them with equal
probability). Hence, our Conitzer elections are more polarized (see the purple
discs at the farthest points) than the Walsh ones, and Walsh ones appear to be
more in agreement (in other words, the map for the Conitzer election is more
similar to that for \AN, and the map for Walsh election is more similar to
\ID).

\begin{figure}
    \centering
    \begin{subfigure}{0.48\textwidth}
        \includegraphics[width=\textwidth]{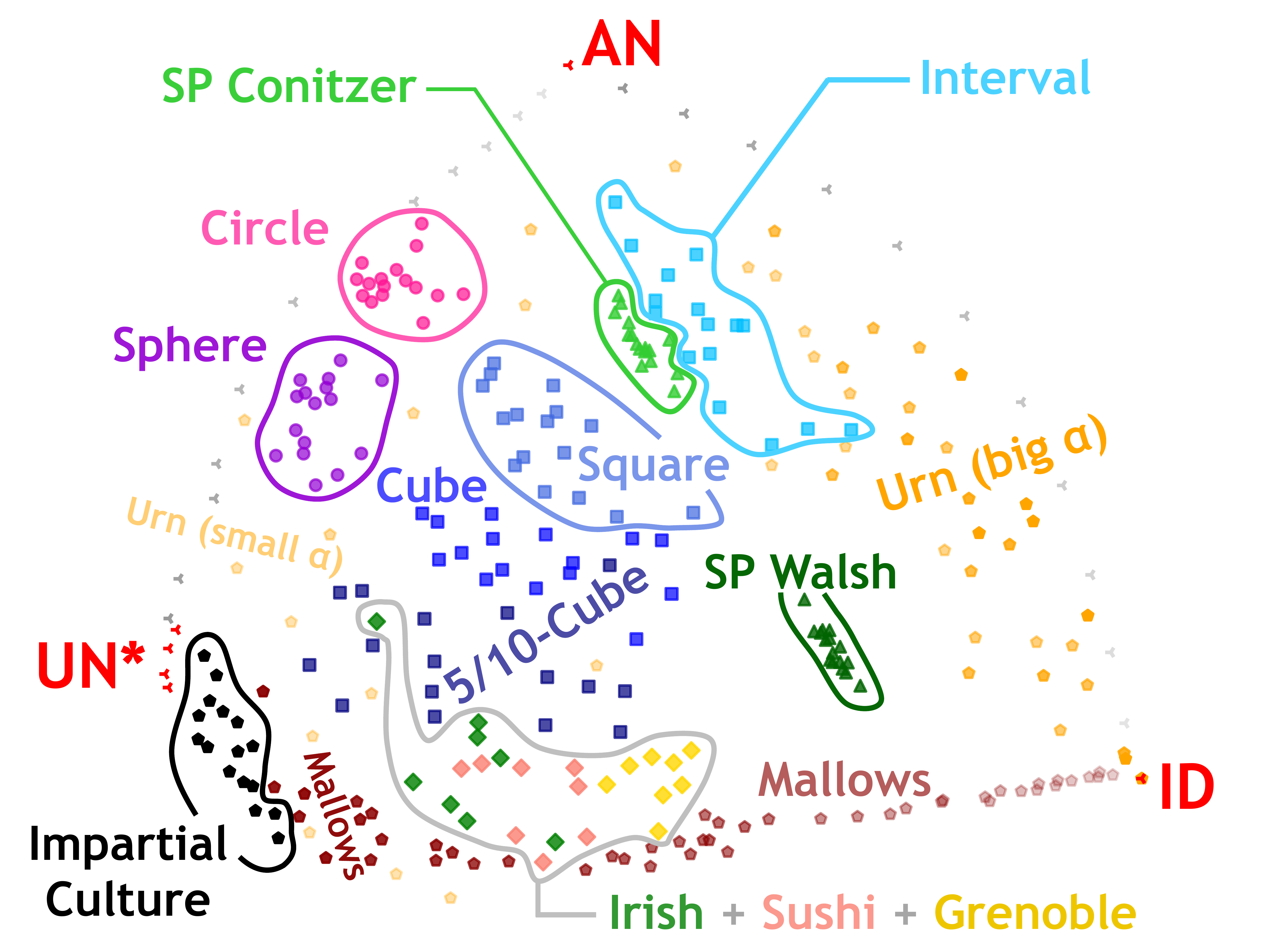}
        \caption{Standard map of elections.}
        \label{fig:swap-map:standard}
    \end{subfigure}
    \hfill
    \begin{subfigure}{0.48\textwidth}
        \includegraphics[width=\textwidth]{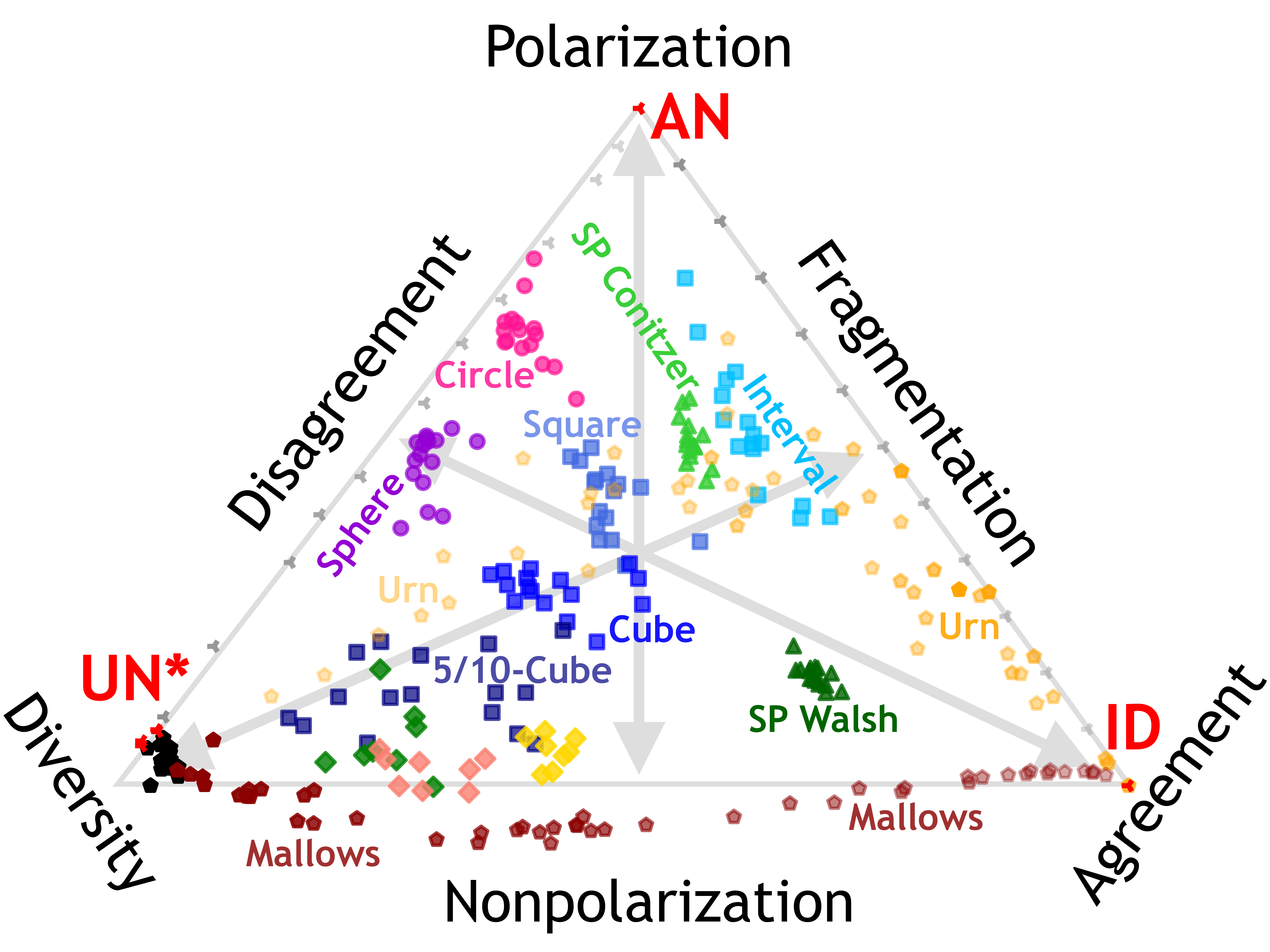}
        \caption{Agreement-diversity map.}
        \label{fig:triangle-map:standard}
    \end{subfigure}
    \caption{%
    Plot (a) shows a
    map of elections in our dataset obtained using isomorphic swap distance and MDS.
    Plot (b) presents an
    affine transformation
    of a plot where x/y coordinates of the elections are their agreement/diversity indices.}
\end{figure}

\subsubsection*{Euclidean Models}
In $d$-dimensional Euclidean elections ($d$-Euclidean elections) every
candidate and every voter is a point in $\mathbb{R}^d$, and a voter
prefers candidate $a$ to candidate $b$ if his or her point is closer
to that of $a$ than to that of $b$.  To generate such elections, we
sample the candidate and voter points
as follows: (a) In the $d$-Cube model, we sample the points uniformly
at random from a $d$-dimensional hypercube $[0,1]^d$, and (b) in the
Circle and Sphere models we sample them uniformly at random from a
circle (embedded in 2D space) and a sphere (embedded in 3D space).  We
refer to the 1-Cube, 2-Cube, and 3-Cube models as, respectively,
the Interval, Square, and Cube models.
In Fig.~\ref{fig:microscope}, we see that as the dimension increases,
the elections become more similar to the IC one (see the transition
from the Interval to the Cube one). The Interval election is very
similar to those of Conitzer and Walsh, because 1-Euclidean elections
are single-peaked. It is also worth noting that the Circle election is
quite polarized (we see an increased density of votes on two opposite
sides of its map).

\subsubsection*{Irish and Other Elections Based on Real-Life Data}

We also consider elections generated based on real-life data from a
2002 political election in Dublin~\citep{mat-wal:c:preflib}. We treat
the full Irish data as a distribution and sample votes from it as
from a statistical culture (technical details in~Appendix~\ref{app:preprocessing}).
The Irish election in Fig.~\ref{fig:microscope} is, in some sense,
between the Cube and Mallows ones for $\normphi =
0.5$. Intuitively, we would say that it is quite diverse.
In the dataset, we also include Sushi and Grenoble elections, similarly
generated using different real-life data~\citep{mat-wal:c:preflib}.

\section{Final Experiments and Conclusion}
\label{sec:results}

In this section we present the results of computing the agreement,
diversity, and polarization indices on 
our dataset.

\subsection{Computing the Indices in Practice}

First, we compared three ways of computing $k$-Kemeny distances: the
greedy approach, the local search with swap size equal to $1$, and a combined
heuristic where we first calculate the greedy solution and then try to
improve it using the local search. We ran all three algorithms for all
$k \in [96]$ and for every election in 
our dataset. The complete results are in
Appendix~\ref{app:kkemeny-computation}.  The conclusion is that the local search and the
combined heuristic gave very similar outcomes and both outperformed
the greedy approach.
Hence, in further computations, we used the former two algorithm and
took the smaller of their outputs.

\begin{figure}[t]
     \centering
     \begin{subfigure}[t]{0.32\textwidth}
         \centering
         \includegraphics[width=\textwidth]{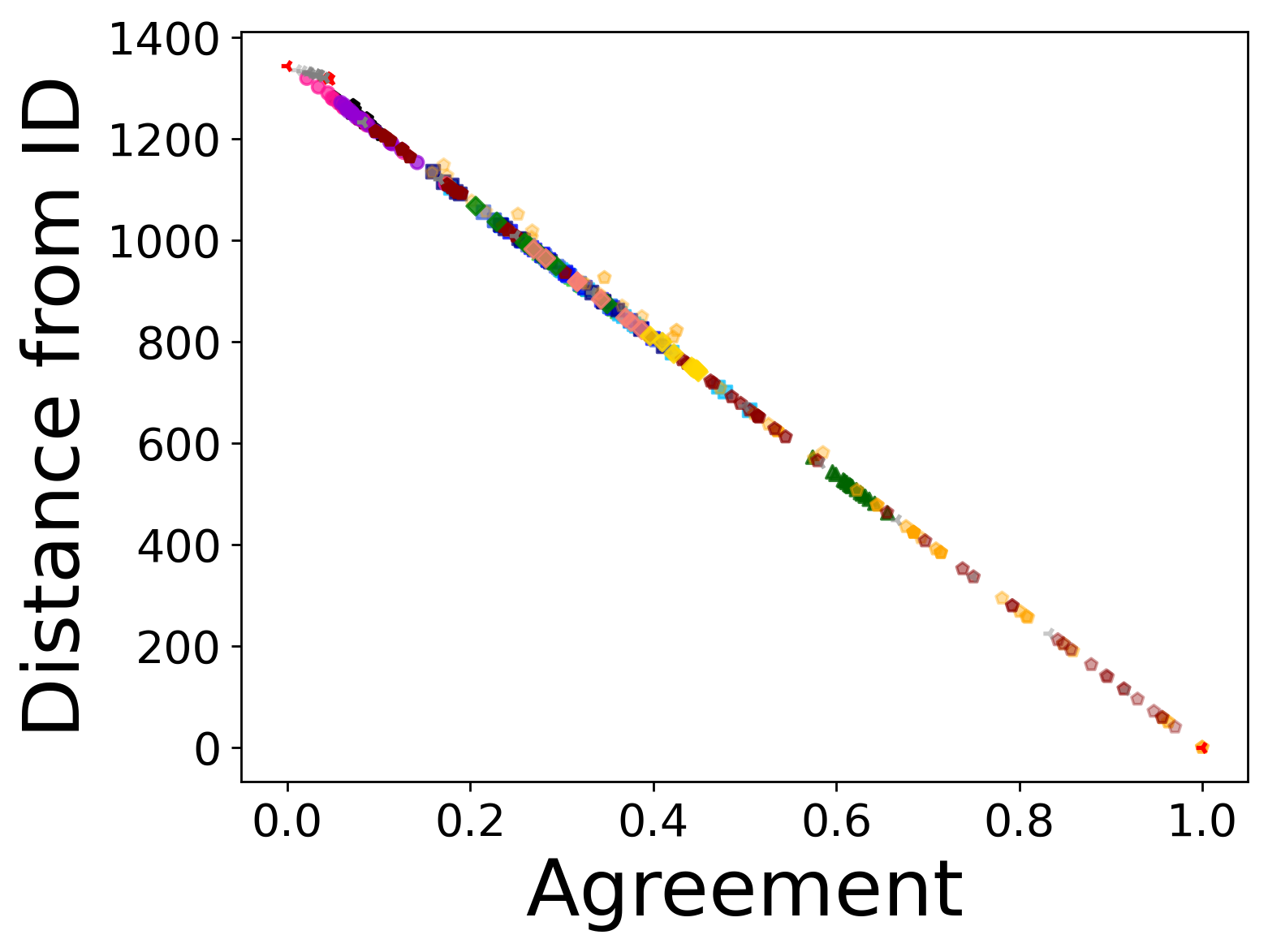}
     \end{subfigure}
     \hfill
     \begin{subfigure}[t]{0.32\textwidth}
         \centering
         \includegraphics[width=\textwidth]{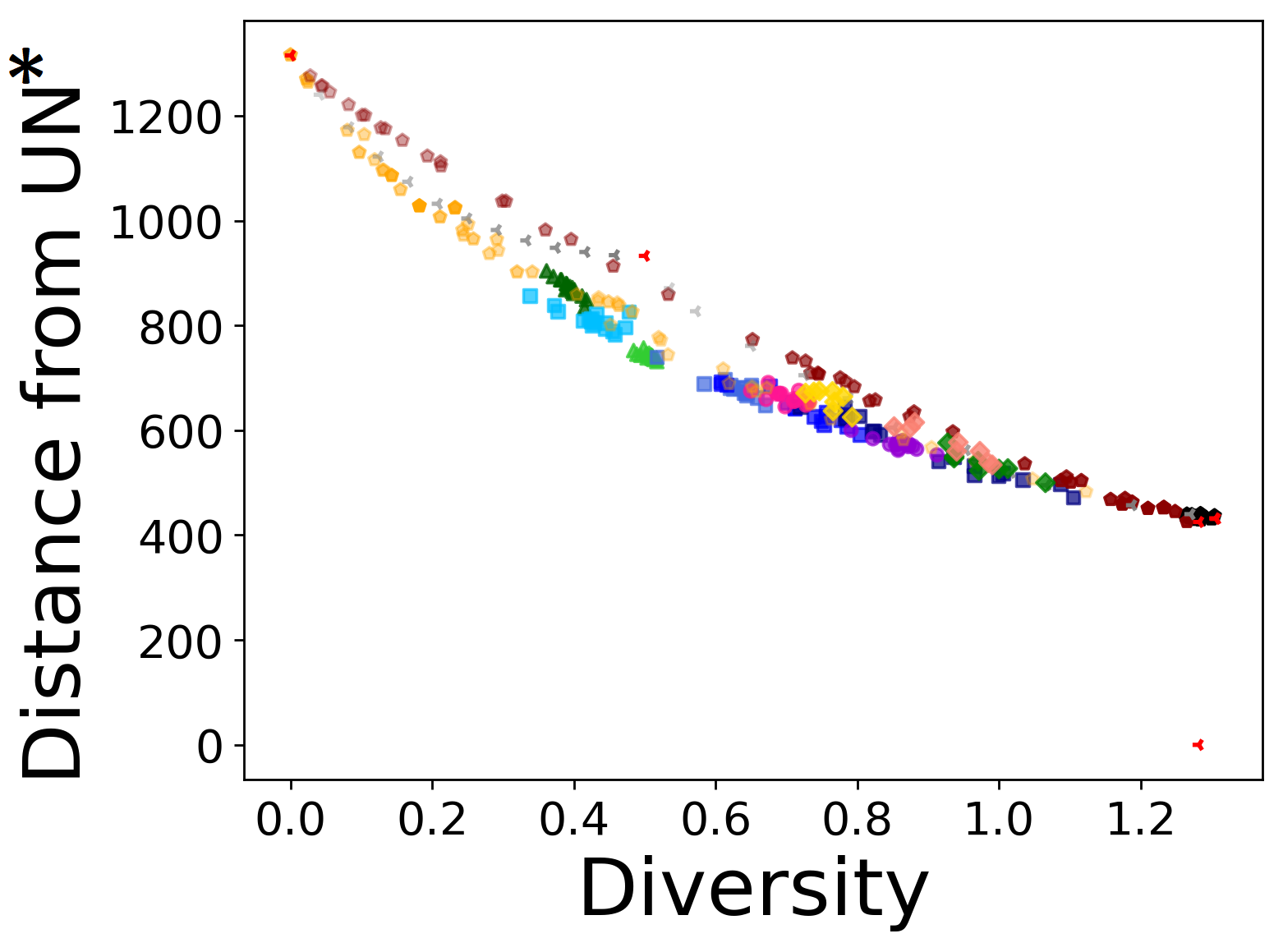}
     \end{subfigure}
     \hfill
     \begin{subfigure}[t]{0.32\textwidth}
         \centering
         \includegraphics[width=\textwidth]{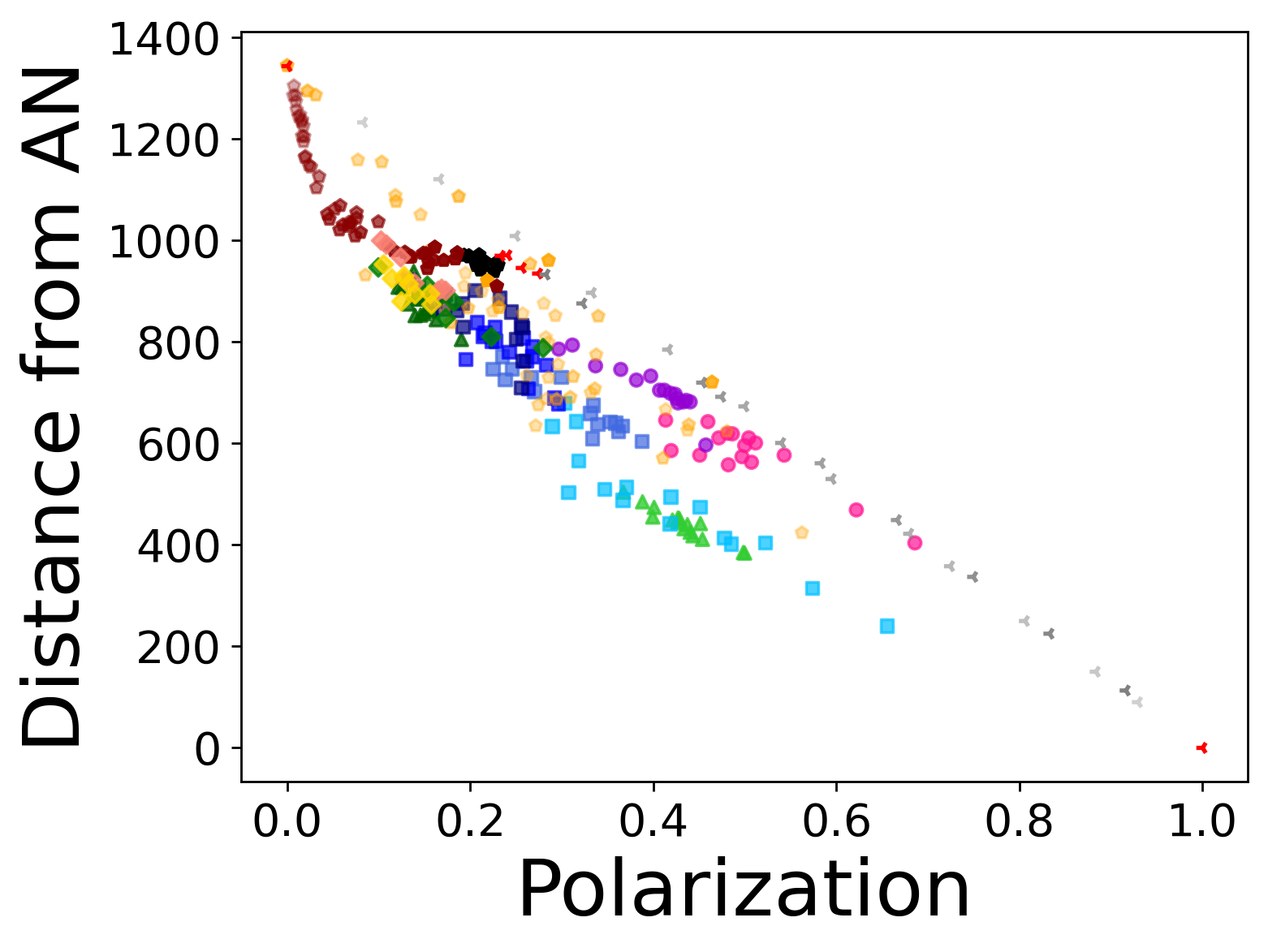}
    \end{subfigure}
    \caption{\label{fig:corr}Correlation between our indices and the distance from the respective compass election.}
\end{figure}

\subsection{Understanding the Map via Agreement, Diversity, and Polarization}

Using the $\kappa_k(E)$ values computed in the preceding experiment,
we calculated diversity and polarization indices of all the elections
from our datasets, along with their agreement indices (which are
straightforward to compute). We illustrate the results in several
ways.

\arxiv{%
First, we consider Fig.~\ref{fig:corr}.  In the leftmost plot, each
election from our dataset is represented as a dot whose x/y
coordinates are the values of the agreement index and the distance
from~\ID, and whose color corresponds to the
statistical culture from which it comes (it is the same as in
Fig.~\ref{fig:swap-map:standard}, though due to large density of the
dots, this only gives a rough idea of the nature of the
elections). The next two plots on the right are analogous, except that it regards
diversity or polarization and the distance from $\appUN$ or \AN, respectively.}
\cameraready{%
First, we consider Fig.~\ref{fig:corr}.  In the left plot, each
election from our dataset is represented as a dot whose x/y
coordinates are the values of the diversity index and the distance
from~$\appUN$, and whose color corresponds to the
statistical culture from which it comes (it is the same as in
Fig.~\ref{fig:swap-map:standard}, though due to large density of the
dots, this only gives a rough idea of the nature of the
elections). The plot on the right is analogous, except that it regards
polarization and distance from \AN. An analogous plot for agreement
and distance from ID is almost perfectly linear (we present it in
Appendix~\ref{app:plots}).}
The Pearson correlation coefficient between
each of the three indices and the distance from the respective compass
election is below $-0.9$, which means that the correlation is very
strong. This is our first indication that the locations on the map of
elections, in particular, the one from
Fig.~\ref{fig:swap-map:standard}, can be understood in terms of
agreement, diversity, and polarization.

Next, for all three pairs of our indices 
we plotted our dataset in such a way that each election's x/y
coordinates are the values of the respective indices
(these plots can be found in Appendix~\ref{app:plots}). We observed
that each of these plots 
resembles the original map from Fig.~\ref{fig:swap-map:standard}. Hence,
for the sake of clearer comparison, we took the plot for agreement and
diversity indices and, by an affine transformation, converted it to a
roughly equilateral triangle spanned between \ID, \AN, and~$\appUN$.
Fig.~\ref{fig:triangle-map:standard} presents the result of
this operation.

The similarity between Figs.~\ref{fig:swap-map:standard}
and~\ref{fig:triangle-map:standard} is striking as most elections can
be found in analogous locations.  Even the positions of the outliers
in the groups are, at least approximately, preserved. Yet, there are
also differences. For example, in Fig.~\ref{fig:triangle-map:standard}
elections from most of the statistical cultures are closer to each
other, whereas on Fig.~\ref{fig:swap-map:standard} they are more
scattered.
Nonetheless, the similarity between these two figures is our second
argument for understanding the map in terms of agreement, diversity,
and polarization. Specifically, the closer an election is to \ID, \AN,
or $\appUN$, the more agreement, polarization, or diversity it exhibits.

\subsection{Validation Against Intuition}

Finally, let us check our intuitions from Section~\ref{sec:cultures}
against the actually computed values of the indices, as presented on
the plot from Fig.~\ref{fig:triangle-map:standard}. We make the
following observations:
\begin{enumerate}
\item We see that Mallows elections indeed progress from ID (for which
  we use $\normphi = 0$) to IC (for which we use $\normphi = 1$), with
  intermediate values of $\normphi$ in between. The model indeed
  generates elections on the agreement-diversity spectrum.
\item Elections generated using the urn model with large value of
  $\alpha$ appear on the agreement-polarization line. Indeed, for very
  large values of $\alpha$ nearly all the votes are identical, but for
  smaller values we see polarization effects. Finally, as the values
  of $\alpha$ go toward $0$, the votes become more and more diverse.
\item Walsh elections are closer to agreement (\ID) and Conitzer
  elections are closer to polarization (\AN).
\item High-dimensional Cube elections have fairly high
  diversity. Circle and Sphere elections are between diversity
  and polarization.
\item Irish elections are between Mallows and high-dimensional Cube elections.
\end{enumerate}
All in all, this confirms our intuitions and expectations.

\section{Summary}
The starting point of our work was an observation that the measures of
diversity and polarization used in computational social choice
literature should, rather, be seen as measures of disagreement. We
have proposed two new measures and we have argued that they do capture
diversity and polarization. On the negative side, our measures are
computationally intractable. Hence, finding a measure that would be easy
to compute but that would maintain the intuitive appeal of our ones is
an interesting research topic.

\section*{Acknowledgements}
Krzysztof Sornat was supported by the SNSF Grant 200021\_200731/1.
This project has received funding from the European Research Council
(ERC) under the European Union’s Horizon 2020 research and innovation
programme (grant agreement No 101002854).
    
\begin{center}
  \includegraphics[width=3cm]{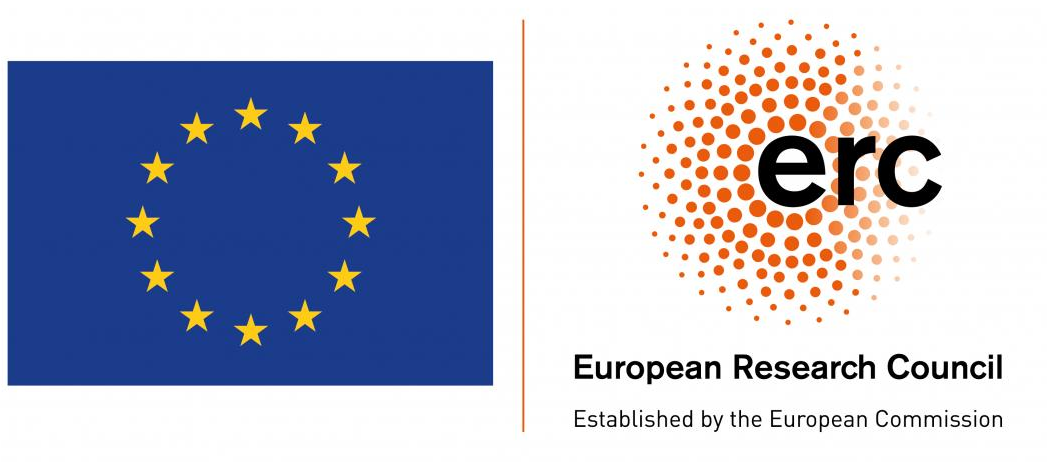}
\end{center}

\bibliographystyle{plainnat}
\bibliography{bib}

\appendix

\section*{Appendix}
\section[Agreement and 1-Kemeny Distance]{Agreement and $\boldsymbol{1}$-Kemeny Distance}
\label{app:agr-1kemeny}
In this section,
we show the relation between the agreement index and $1$-Kemeny distance.
By $\mu$ let us denote the \emph{majority relation}, which
is a (possibly intransitive and not asymmetric) relation on the set of candidates such
that for each $a, b \in C$, $a \succeq_\mu b$ if and only if
$p_E(a,b) \ge p_E(b,a)$.
Let say that the linear order $\lambda$ of candidates
is \emph{consistent} with $\mu$ if
$a \pref_\lambda b$ implies $a \succeq_\mu b$,
for every $a, b \in C$.
Observe that if an election $E = (C,V)$
does not have a \emph{Condorcet cycle},
i.e., there is no sequence of
candidates $c_1,\dots,c_k \in C$ such that
$p_E(c_i,c_{i+1}) \ge p_E(c_{i+1},c_i)$, for every $i \in [k-1]$, and
$p_E(c_k,c_1) > p_E(c_1,c_k)$,
then the set of preference orders
consistent with $\mu$ is the set of all Kemeny rankings.

The Kendall's $\tau$ distance can be generalized for any relations.
For every linear order $\lambda$ over candidates $C$, we have
\[
    \tau(\lambda, \mu) = \frac{1}{2}
        \sum_{a,b \in C}
        \left(
        \1_{a \pref_\lambda b} \cdot \1_{a \not\succeq_\mu b} +
        \1_{b \succeq_\mu a} \cdot \1_{b \not\pref_\lambda a}
        \right).
\]
In other words, for every pair of candidates $a, b$
for which $a \pref_\lambda b$,
we count 1, if $b \succeq_\mu a$ and $a \not\succeq_\mu b$,
and $\nicefrac{1}{2}$, if $b \succeq_\mu a$ but also $a \succeq_\mu b$.
As we show in the following proposition,
there is a strict relation between the agreement index and
the average Kendall's $\tau$ distance from all votes to the
majority relation.

\begin{proposition}
For every election $E=(C,V)$, it holds that
\[
    A(E) = 1 - \frac{2 \cdot \sum_{v \in V} \tau(v, \mu)}{|V| \cdot \textstyle \binom{|C|}{2}}.
\]
\end{proposition}
\begin{proof}
We split the set of pairs of candidates into two subsets:
$A$ containing the pairs with perfect disagreement,
and $B$ with the pairs for which some opinion is stronger than the other.
Formally, let
$A = \{ \{a,b\} \subseteq C : p_E(a,b) = p_E(b,a)\}$
and
$B = \{ \{a,b\} \subseteq C : p_E(a,b) \neq p_E(b,a) \}$.
Without loss of generality,
throughout the proof we assume that for pair $\{a,b\} \in B$
we have $a \succeq_\mu b$.
Then, by the definition of $\mu$ we get $p_E(a,b) > p_E(b,a)$ and thus
\begin{align*}
    |p_E(a,b) - p_E(b,a)| &= p_E(a,b) - p_E(b,a) \\
        &= (p_E(a,b) + p_E(b,a)) - 2 p_E(b,a) \\
        &= 1 - 2 p_E(b,a).
\end{align*}
Since for $\{a,b\} \in A$ we have $|p_E(a,b) - p_E(b,a)|=0$,
by the definition of the agreement index, we get that
\begin{align*}
    A(E) &= \frac{\sum_{\{a,b\} \in B} 
                \left( 1 - 2 p_E(b,a) \right)
            }{\binom{|C|}{2}} \\
        &= \frac{|B| - 2 \cdot \sum_{\{a,b\} \in B} \sum_{v \in V} 
                \1_{b \pref_v a} / |V|
            }{\binom{|C|}{2}} \\
        &= \frac{|B|}{\binom{|C|}{2}} - 
            2 \frac{\sum_{v \in V} \sum_{\{a,b\} \in B} 
                \1_{b \pref_v a}
            }{|V| \cdot \binom{|C|}{2}} \\
        &= \frac{|B|}{\binom{|C|}{2}} - 
            2 \frac{\sum_{v \in V} \sum_{\{a,b\} \in B}
                \nicefrac{1}{2} \1_{b \pref_v a} + \nicefrac{1}{2} \1_{a \not\pref_v b}
            }{|V| \cdot \binom{|C|}{2}},
\end{align*}
where the last equation comes from the fact that $\pref_v$ is asymmetric, so $\1_{b \pref_v a} = \1_{a \not\pref_v b}$.
Since for $\{a,b\} \in B$ we have $p_E(a,b) > p_E(b,a)$,
then we know that $\1_{a \succeq_\mu b} = \1_{b \not\succeq_\mu a} = 1$ and, conversely,
$\1_{a \not\succeq_\mu b} = \1_{b \succeq_\mu a} = 0$.
In particular, this means that
\begin{align*}
    \nicefrac{1}{2}\1_{b \pref_v a} + \nicefrac{1}{2}\1_{a \not\pref_v b} &= 
    \nicefrac{1}{2}\1_{b \pref_v a}\1_{b \not\succeq_\mu a} + 
    \nicefrac{1}{2}\1_{a \not\pref_v b}\1_{a \succeq_\mu b} \\ &+
    \nicefrac{1}{2}\1_{a \pref_v b}\1_{a \not\succeq_\mu b} + 
    \nicefrac{1}{2}\1_{b \not\pref_v a}\1_{b  \succeq_\mu a},
\end{align*}
which we denote as $\tau_{a,b}(v.\mu)$.
Then, we have that
\begin{equation}
\label{eq:agreement:majority_relation}
A(E) = \frac{|B|}{\binom{|C|}{2}} - 
        2 \frac{\sum_{v \in V} \sum_{\{a,b\} \in B}
            \tau_{a,b}(v,\mu)
        }{|V| \cdot \binom{|C|}{2}}.
\end{equation}

Now, let us consider a pair of candidates $\{a,b\} \in A$.
Observe that independently whether $a \pref_v b$ or $b \pref_v a$
we have that 
\begin{align*}
    \tau_{a,b}(v,\mu) = \hspace{3pt} &\nicefrac{1}{2}\1_{b \pref_v a}\cdot 0 + 
    \nicefrac{1}{2}\1_{a \not\pref_v b}\cdot 1\\
    + \hspace{3pt}&\nicefrac{1}{2}\1_{a \pref_v b}\cdot 0
    + \nicefrac{1}{2}\1_{b \not\pref_v a}\cdot 1
    = \nicefrac{1}{2}.
\end{align*}
Therefore, summing for all voters and pairs of candidates in set $A$, we obtain
\[
    \textstyle\sum_{v \in V}  \sum_{\{a,b\} \in A} \tau_{a,b}(v,\mu) = \nicefrac{1}{2} |A| \cdot |V|.
\]
We can rearrange this equation and divide by $\nicefrac{1}{2}|V|\binom{|C|}{2}$,
to get
\[
    0 = \frac{|A|}{\binom{|C|}{2}} - 
        2\frac{\sum_{v \in V}  \sum_{\{a,b\} \in A} 
            \tau_{a,b}(v,\mu)
        }{|V| \cdot \binom{|C|}{2}}.
\]
Combining this we equation~\eqref{eq:agreement:majority_relation}
we obtain
\begin{align*}
A(E) &= 1 - 
        2 \frac{\textstyle\sum_{v \in V} \sum_{\{a,b\} \subseteq C} 
            \tau_{a,b}(v,\mu)
        }{|V| \cdot \textstyle \binom{|C|}{2}} \\
    &= 1 - \frac{2 \cdot \textstyle\sum_{v \in V} \tau(v,\mu)}
            {|V| \cdot \textstyle \binom{|C|}{2} }.
\end{align*}
\end{proof}
Since in elections without a Condorcet cycles
every Kemeny ranking is consistent with $\mu$,
we get that in such elections
there is a strict relation between the agreement index and $1$-Kemeny distance.
\begin{corollary}
 For every election $E=(C,V)$ without a Condorcet cycle, it holds that
\[
    A(E) = 1 - 2 \cdot \kappa_1(E) / \left(|V| \cdot \textstyle \binom{|C|}{2}\right).
\]
\end{corollary}

\section[Proof of Theorem 1]{Proof of Theorem~\ref{thm:2-apx-by-among-votes}}
For a given instance $I = (E=(C,V),k)$ a feasible solution for \kkemenyamongvotes is also a feasible solution to \kkemeny.
Let $\mu_k(E)$ be the optimum value of \kkemenyamongvotes on $I$ and $\kappa_k(E)$ be the optimum value of \kkemeny on $I$.
In order to show the theorem statement it is enough to show that $\mu_k(E) \leq 2\kappa_k(E)$.

Let $\Lambda = \{\lambda_1, \dots, \lambda_k\}$ be an optimum solution for \kkemenyamongvotes
and $\Gamma = \{\gamma_1, \dots, \gamma_k\}$ be an optimum solution for \kkemeny.
Let $v(x) \in V$ be a voter that is closest to some ranking $x$ and $v(\Gamma) = \{v(\gamma): \gamma \in \Gamma \}$.
Let $\gamma(x) \in \Gamma$ be a ranking from $\Gamma$ that is closest to some
ranking~$x$.
We define $\swap(v,X) = \min_{x \in X} \swap(v,x)$.
We have
\begin{align*}
  \mu_k(E)
  &=
    \sum_{v \in V} \swap(v,\Lambda)\\
  &\leq
    \sum_{v \in V} \swap(v,v(\Gamma))\\
  &\leq \sum_{v \in V} \Big(\swap(v,\gamma(v)) + \swap(\gamma(v),v(\gamma(v)))\Big)\\
  &\leq 2 \cdot \sum_{v \in V} \swap(v,\gamma(v)) = 2\kappa_k(E).
\end{align*}
where the first inequality holds because of optimality of $\Lambda$ restricted
to votes and the second inequality is due to the triangle inequality. The third
inequality follows from $\swap(\gamma(v),v(\gamma(v))) \leq \swap(\gamma(v),v)$,
which expresses that for some vote~$v \in V$, its distance to the closest
ranking~$\gamma(v)$ from~$\Gamma$ is at least as large as the distance
between~$\gamma(v)$ and a vote closest to it.
This finishes the proof.

\section[Proof of Theorem 2]{Proof of Theorem~\ref{thm:fpt-as-n}}
  Let us fix some $\epsilon > 0$.

  We consider every possible subset of votes as a cluster;
  there are $2^n$ of them. First, our algorithm runs a PTAS designed for
  {\sc $1$-Kemeny}~\citep{KenyonMathieuS07} for every possible cluster
  and store the result. This gives us an $(1+\epsilon)$-approximate
  solution for every cluster separately.
  
  Second, our algorithm guesses a $k$-clustering $\{V_1, V_2,
  \ldots, V_k\}$~of votes. Then, for each cluster in the clustering, we take an
  $(1+\epsilon)$-approximate solution to {\sc $1$-Kemeny} (which was computed in
  the first step) and store it. The algorithm repeats this procedure for each of
  $k^n$~possible clusterings and outputs the smallest computed distance.

  It is clear that an optimum solution corresponds to one of the
  $k$-clusterings, say $K$, analyzed by the algorithm in the second step.
  Moreover, in each cluster of $K$ the solution returned by the algorithm is a
  $(1+\epsilon)$-approximation of the optimum solution of the cluster under
  consideration. Hence, eventually, the algorithm returns a \kkemeny solution
  that costs at most a multiplicative factor $1+\epsilon$ more than the optimum
  one, as claimed.
  
  Regarding the running time, note that $k<n$; otherwise, the set of votes
  gives a solution of cost $0$. The algorithm computes a solution for $2^n$
  many clusters (each in polynomial time) and considers $k^n \leq n^n$ many
  clusterings (each in polynomial time), so the running time is FPT
  w.r.t.~$n$, namely $n^n \cdot \poly(n,m)$.

\section[Proof of Theorem 3]{Proof of Theorem~\ref{thm:w2-kkemeny}}\label{app:w2-kkemeny}
In the main text we provided the construction of the reduction.
Here we prove its correctness.

First, let us assume that there is some (partial) cover $\calR \subseteq \calS, |\calR| = K$ such that $|\bigcup_{S \in \calR} S| \geq T$.
We claim that the set~$\Lambda = \{\pref_{v_j}: S_j \in \calR \}$ of $k$~rankings has the $k$-Kemeny distance at most $D$.

For every (copy of) set-voter $v_j$ such that $S_j \in \calR$, we have $\swap(v_j,\Lambda) \leq \swap(v_j,v_j) = 0$ and for the remaining $L(M-K)$ set-voters the distance to $\Lambda$ equals $2$.
Hence, set-voters realize the distance equal to the first term in the definition of $D$.

Now, we calculate the distance realized by element-voters.
For each element-voter $e_i$, representing element~$x_i$ that is not covered
by~$\calR$, its swap distance~$\swap(e_i,\Lambda)$ can be computed
as follows. Starting from the distance being~$0$, we add one for each set in
which $x_i$ is included and we add $3$ because of the pivot-candidates.
Furthermore, we increase the distance by one once more, due to the following.
For every vote~$v_j \in \Lambda$ (recall that in~$v_j$ candidate~$d_{S_j}$ is
preferred to~$c_{S_j}$), we have that in vote~$e_i$ candidate~$c_{S_j}$ is
preferred to~$d_{S_j}$, since~$x_i$ is not covered. 
So, formally, for an
element-voter~$e_i$ that represents an element~$x_i$ not covered by~$\calR$, we
obtain the following formula:
$$
\swap(e_i,\Lambda) = |\{S \in \calS: x_i \in S\}| + 3 + 1.
$$
If, however, element~$x_i$ is covered by some set, say~$S_j$, in~$\calR$, then
candidates~$c_{S_j}$ and~$d_{S_j}$ are in the same order in~$e_i$ and~$v_j$
and~$v_j \in \Lambda$. Hence, we should decrease the computed distance by two.
By one, due to the fact that, we added one for each set in~$x_i$ is included;
hence we also assumed that the order of~$c_{S_j}$ and~$d_{S_j}$ is reversed
in~$e_i$ and~$v_i$. By another one because also the last summand of the
aforementioned formula came from the (now false) assumption there is no vote
in~$\Lambda$ for which ~$c_{S_j}$ and~$d_{S_j}$ are in the same order in~$e_i$
and~$v_i$.  Since we computed their inversion in the first part of the formula.
Eventually, introducing the indicator function~$\ind[\Phi]$ such that
$\ind[\Phi] = 1$ if $\Phi$ is true, and $\ind[\Phi] = 0$ otherwise, formally the
sought~$\swap(e_i,\Lambda)$ is
$$\swap(e_i,\Lambda) = |\{S \in \calS: x_i \in S\}| + 4 - 2 \cdot \ind[x_i \in
\bigcup_{S \in \calR} S].$$
It means that the distance realized by element-voters is equal to
\begin{align*}
  \sum_{i \in [N]} \hspace{-2pt}\swap(e_i,\Lambda)
  &= \sum_{j \in [M]} |S_j| + 4N - 2 \cdot \Big|\bigcup_{S \in \calR} S \Big|\\
  &\leq \sum_{j \in [M]} |S_j| + 4N - 2T.
\end{align*}
In total, $\swap(V,\Lambda) \leq D$, as required.

Now, let us assume that there is $\Lambda \subseteq \{\pref_{v}: v\in V\}, |\Lambda| = k $ such that $\swap(V,\Lambda) \leq D$.

First of all, we observe that $\Lambda$ may contain only rankings of set-voters.
Let us assume, by contradiction, that there is an element-voter in $\Lambda$.
It means that at most $L(k-1)$ set-voters realize the swap distance $0$.
Furthermore, at least $L(M-k+1)$ set-voters realize the swap distance at least $2$
(it is exactly $2$ when the closest ranking comes from a set-vote, and it is at least $3$ when the closest ranking comes from an element-vote).
Hence, we would have
$\swap(V,\Lambda) \geq 2L(M-k+1) > 2L(M-K)+NM+4N \geq D$,
which is a contradiction with $\swap(V,\Lambda) \leq D.$

Using the same calculation as in the previous paragraph, we can conclude that $\Lambda$ does not contain two copies of the same set-voter.
Because of that, we can define $\calR \subseteq \calS$ containing exactly $k=K$ subsets corresponding to votes from $\Lambda$, i.e., $\calR = \{S_j:\ \pref_{v_j} \in \Lambda\}$.

We will show that $\calR$ covers at least $T$ elements.
Let us assume, by contradiction, that $\calR$ covers at most $T-1$ elements.
Then we would have:
\begin{align*}
  \swap(V,\Lambda)
  &= 2L(M-K) + \sum_{j \in [M]} |S_j| + 4N - 2 \cdot \Big|\bigcup_{S \in \calR} S \Big|\\
  &\geq 2L(M-K) + \sum_{j \in [M]} |S_j| + 4N - 2(T-1)\\
  &= D+2,
\end{align*}
which is a contradiction with $\swap(V,\Lambda) \leq D$.

\section[Propositions from Theorem 3]{Propositions from Theorem~\ref{thm:w2-kkemeny}}
\label{app:propositions-thm3}
Let us define $M = \max_{v,u \in V} \swap(v,u)$, i.e., the
maximum distance between votes. The value of~$M$ is small in instances with
similar votes. Unfortunately, small values of $M$ do not make the problem easy.
\begin{proposition}\label{prop:kkemeny-w1-km}
  \kkemenyamongvotes is $\mathrm{W}$[1]-hard when parameterized by $k+M$.
\end{proposition}
\begin{proof}
  It is known that \maxkcover is W[1]-hard w.r.t.~$K+f$~\citep{BonnetPS16}, where $f$ is \emph{the maximum frequency of an element}, i.e., $f = \max_{i \in [N]} |\{S_j \in \calS : x_i \in S_j\}|$.
  We can observe that $M \leq 2f+4$ in the reduction given in the proof of Theorem~\ref{thm:w2-kkemeny} hence we obtain the proposition statement.
\end{proof}

By adapting results regarding \maxkcover~\cite[Observation
7]{SornatWX22}, we also obtain the following bound that uses the
Strong Exponential Time Hypothesis (SETH).\footnote{SETH is one of
  popular complexity assumptions in parameterized complexity.  For a
  formal statement see, e.g., the book of~\citespecial{CyganFKLMPPS15}{Conjecture
  14.2}.}

\begin{proposition}\label{prop:seth-kkemeny-m}
  There is no $1.4^m \cdot \poly(n,m)$-time algorithm for
  \kkemenyamongvotes, where $m$ is the number of candidates and $n$ is
  the number of voters, unless SETH fails.
\end{proposition}
\begin{proof}
  Let us assume, by contradiction, that there is a $1.4^m \cdot \poly(n,m)$-time algorithm for \kkemenyamongvotes.
  We take an instance of \maxkcover and reduce it (in $\poly(N,M)$ time) to \kkemenyamongvotes using the reduction from the proof of Theorem~\ref{thm:w2-kkemeny}.
  We solve the obtained instance of \kkemenyamongvotes in time $1.4^m \cdot \poly(n,m)$ and we output the same response to \maxkcover.
  Due to Theorem~\ref{thm:w2-kkemeny}, we obtained a correct response to the instance of \maxkcover.
  Recall that $m = 2M+3$ and $n = N(M^2+4M+1)$.
  Therefore, the running time of our algorithm for \maxkcover is at most
  $1.4^m \cdot \poly(n,m) \cdot \poly(N,M) \leq 1.4^{2M+3} \cdot \poly(N,M) \leq 1.96^M \cdot \poly(N,M)$.
  This would show that SETH is false because under SETH \maxkcover has no $1.99^M \cdot \poly(N,M)$ time algorithm~\cite[Observation 7]{SornatWX22}.
\end{proof}
On the other hand, \kkemenyamongvotes (and \kkemeny) is FPT w.r.t.~$m$ by a brute-force evaluation of all $k$-size subsets of $m!$ possible linear orders as a solution, each in polynomial time.
Hence, the running time is $\binom{m!}{k} \cdot \poly(n,m) \leq 2^{m!} \cdot \poly(n,m) \leq 2^{m^m} \cdot \poly(n,m)$.
This is a double-exponential dependence.
An open question is to provide a single-exponential time algorithm.

\begin{table}[t]
\centering
 \caption{The ingredients of the standard dataset. \label{tbl:map-cultures}}
\small
\begin{tabular}{rll}
  \toprule
  model & variants/parameters & \#elcs\\\midrule
  Impartial Culture & & 16\\
  normalized Mallows & $\phi \in \textrm{unif. over } [0,1]$ & 48\\
  urn model & $\alpha \in \Gamma(0.8)$ & 48\\
  \midrule 
  single-peaked (Conitzer) & & 16\\
  single-peaked (Walsh) & & 16\\
  \midrule  
  1-cube (Interval)  & uniform interval & 16\\
  2-cube (Square)  & uniform square & 16\\
  3-cube (Cube)  & uniform cube & 16\\
  5-cube  & uniform 5D-cube & 8 \\
  10-cube & uniform 10D-cube & 8 \\
  circle & circle in 2D & 16\\
  sphere & sphere in 3D & 16\\
  \midrule
  Irish dataset & & 8\\
  Sushi dataset & & 8\\
  Grenoble dataset & & 8\\
  \midrule
  uniformity ($\appUN$) & &4\\
  identity (ID) & & 1\\
  antagonism (AN) & & 1\\
  \midrule
  ID-AN mixture & AN fractions: $\nicefrac{1}{12} \ldots
                  \nicefrac{11}{12}$ & 11 \\
  AN-$\appUN$ mixture& $\appUN$ fractions: $\nicefrac{1}{12} \ldots
                 \nicefrac{11}{12}$ & 11\\
  \bottomrule
\end{tabular}
\end{table}

\section{Standard Dataset Composition}
\label{app:assorted-dataset}
The map of elections from Fig.~\ref{fig:swap-map:standard} consists of elections from
various statistical cultures.
In Table~\ref{tbl:map-cultures} we
specify how many elections come from each culture and how their
parameters were chosen.
From now on, we will call this collection of elections
(i.e., the elections depicted in Fig.~\ref{fig:swap-map:standard})
as the \emph{standard dataset}, to distinguish it from
the \emph{extended dataset} and the \emph{Mallows dataset}
presented in the following sections.
In what follows, we describe how we generate elections
that were not covered in Section~\ref{sec:cultures} (or Appendix~\ref{app:preprocessing}).

Before we begin, let us describe a general technique
that is sampling elections from a position matrix.
A \emph{position matrix} \citep{szu-fal-sko-sli-tal:c:map,boe-bre-fal-nie-szu:c:compass} is an integer $m \times m$
matrix, in which the values of each row and each column
sum up to some constant $n \in \mathbb{N}$.
An election, $E=(C,V)$, \emph{realizes} a given
position matrix $X$, if  $|C|=m$, $|V|=n$, and
for every $i,j \in [m]$,
the value in $i$-th row and $j$-th column
of matrix $X$, i.e., $X_{i,j}$,
is equal to the number of voters in $V$ that
ranks the $j$-th candidate at the $i$-th position
(note that one position matrix can be realized
by multiple elections).
For example, a position matrix realizing 
\UN election with $m$ candidates,
is an $m \times m$ matrix
with each element equal to $(m-1)!$.
\citet{pos-matrices},
provide a technique to sample elections
realizing given position matrix $X$, which starts from
an empty election without any votes, and then, iteratively:
\begin{enumerate}
    \item finds a vote $v$ that can belong to an election realizing $X$,
    \item adds $v$ to the election, and then
    \item updates the values of matrix $X$
        (by subtracting one from $X_{i,j}$ for every
        $j \in [m]$ and $i$ being the position of $j$-th candidate
        according to vote $v$),
\end{enumerate}
until $X$ is a zero matrix.
We note that this procedure returns every election
realizing given matrix with positive probability,
but the exact distribution we obtain is unknown
(\citet{pos-matrices} argue that
obtaining a P-time uniform sampler
is challenging).
We use this sampling technique to generate
$\appUN$ and AN-$\appUN$ mixture elections.

\paragraph{UN*}
To generate $\appUN$ elections,
we sample an election realizing an $8 \times 8$ position matrix
in which every element is equal to $12$.

\paragraph{ID-AN mixture.}
Elections from \ID-\AN mixture model with \AN share
$i \in \{\nicefrac{1}{12}, \dots, \nicefrac{11}{12}\}$
come from merging \AN election with $96i$ voters and
\ID election with $96(1-i)$ voters.
Hence, we have $96 - 48i$ voters with a given preference order
and $48i$ voters with exactly opposing views.

\paragraph{AN-UN* mixture.}
Elections from \AN-$\appUN$ mixture model with $\appUN$ share
$i \in \{\nicefrac{1}{12}, \dots, \nicefrac{11}{12}\}$
come from merging $\appUN$ election with $96i$ voters and
\AN election with $96(1-i)$ voters.
Hence, we have $48(1-i)$ voters with a given preference order,
$48(1-i)$ voters with exactly opposing views,
and on top of that we add $96i$ voters that we get
by sampling election realizing $8 \times 8$ matrix
in which every element is equal $12i$.

\section{Preprocessing of Real-life Data}\label{app:preprocessing}

\paragraph{Grenoble.}
In the Grenoble field experiment, $760$ people were asked to place $11$ candidates on the $[0,1]$ line. The higher the value, the more a given candidate is liked by a voter. We converted each participant's \textit{line preference} into ordinal ranking, by choosing the candidate being closest to one as a first choice, the second closest to one as a second choice and so on.

\paragraph{Sushi.}
In the survey about Sushi there were $5000$ participants and $10$ different types of sushi (i.e., candidates). The original data consists of full ordinal rankings without ties.

\paragraph{Irish.}
In the election held in Dublin North constituency, there were $43942$ voters and $12$ candidates. In the original data many votes were incomplete, hence, we filled them using the same procedure as~\citet{boe-bre-fal-nie-szu:c:compass}, in order to obtain complete preference orders.

\paragraph{Sampling procedure.}
We decided to conduct experiments with $8$ candidates, hence, for all three dataset we selected $8$ candidates having the highest Borda score. We treat all three datasets as statistical cultures. To sample an 
election from a given dataset, we simply sample a given number of votes (in our case $1000$) uniformly at random (sequentially with returning).

\section{Extended Dataset}
\label{app:extended_dataset}

In this section, we introduce our extended dataset.
This dataset consists of all 292 elections
from the standard dataset and 74 new elections
generated using 4 additional statistical cultures
(single-peaked on a circle,
single-crossing,
group-separable balanced, and
group separable caterpillar)
and 2 special models
($\alpha$-stratification and \ID-$\appST$ mixture).
The exact composition of the extended dataset is presented in Table~\ref{tbl:map-extended-cultures}.
In what follows we describe each new culture and model.

We present also map of preferences
for elections from these cultures in Fig.~\ref{fig:microscope:extended}
(some additional maps for cultures and models
already appearing in the standard dataset are also included).
In order to obtain maps of preferences more representative for their models,
we generated elections with 1000 voters instead of 96
(but we present also the version with 96 voters in
Appendix~\ref{app:microscope_96}).
Finally, a map of elections
generated in the same way as that in Fig.~\ref{fig:swap-map:standard},
but for elections in the extended dataset
is presented in Fig.~\ref{fig:swap-map:extended}.

\begin{table}[t]
\centering
 \caption{The ingredients of the extended dataset
 (elections not appearing in the standard dataset are in \textbf{boldface}). \label{tbl:map-extended-cultures}}
\small
\begin{tabular}{rll}
  \toprule
  model & variants/parameters & \#elcs\\\midrule
  Impartial Culture & & 16\\
  normalized Mallows & $\phi \in \textrm{unif. over } [0,1]$ & 48\\
  urn model & $\alpha \in \Gamma(0.8)$ & 48\\
  \midrule 
  single-peaked (Conitzer) & & 16\\
  single-peaked (Walsh) & & 16\\
  \textbf{single-peaked on a circle} & & 16\\
  \textbf{single-crossing} & & 16\\
  \textbf{group-separable} & balanced & 16  \\
  \textbf{group-separable} & caterpillar & 16 \\
  \midrule  
  1-cube (Interval)  & uniform interval & 16\\
  2-cube (Square)  & uniform square & 16\\
  3-cube (Cube)  & uniform cube & 16\\
  5-cube  & uniform 5D-cube & 8 \\
  10-cube & uniform 10D-cube & 8 \\
  circle & circle in 2D & 16\\
  sphere & sphere in 3D & 16\\
  \midrule
  Irish dataset & & 8\\
  Sushi dataset & & 8\\
  Grenoble dataset & & 8\\
  \midrule
  uniformity ($\appUN$) & &4\\
  $\mathbf{\nicefrac{\mathbf{1}}{2}}$\textbf{-stratification} ($\mathbf{ST^*}$) & & 4\\
  identity (ID) & & 1\\
  antagonism (AN) & & 1\\
  $\boldsymbol{\alpha}$\textbf{-stratification} & $\alpha \in \{\nicefrac{1}{8}, \nicefrac{2}{8},
                            \nicefrac{3}{8}\}$ & 3\\
  \midrule
  ID-AN mixture & AN share: $\nicefrac{1}{12} \ldots
                  \nicefrac{11}{12}$ & 11 \\
  AN-$\appUN$ mixture& $\appUN$ share: $\nicefrac{1}{12} \ldots
                 \nicefrac{11}{12}$ & 11\\
  \textbf{ID-}$\mathbf{ST^*}$ \textbf{mixture}& no. blocks: $3$, $4$, $6$ & 3 \\
  \bottomrule
\end{tabular}
\end{table}

\subsubsection*{Single-Peaked On a Cycle Elections (SPOC)}
Elections single-peaked on a circle~\citep{pet-lac:j:spoc} are
analogous to single-peaked ones, except that the societal axis is
cyclic (so a vote is SPOC with respect to axis $\rhd$ if for every
$t \in [m]$ its $t$ top-ranked candidates either form an interval with
respect to $\rhd$ or a complement of an interval; an election is SPOC
if there is an axis with respect to which all its votes are
SPOC). Such preferences occur, e.g., when choosing a virtual meeting
time and voters are in different time zones.  We generate SPOC
elections by choosing SPOC votes uniformly at random (for SPOC, this
is equivalent to using the Conitzer approach).
The shape of the SPOC election in Fig.~\ref{fig:microscope:extended}
naturally corresponds to the cyclic nature of the axis.

\subsubsection*{Single-Crossing Elections}
Single-crossingness captures a similar idea as single-peakedness,
but based on ordering the voters.
\begin{definition}[\citet{mir:j:single-crossing,rob:j:tax}]
  An election is single-crossing if it is possible to order the voters
  so that for each two candidates $a$ and $b$ either every voter who
  prefers $a$ to $b$ comes before every voter who prefers $b$ to $a$,
  or the other way round.
\end{definition}
We generate single-crossing elections using the approach of
\citet{szu-fal-sko-sli-tal:c:map}. First, we generate a
\emph{single-crossing domain}, i.e., a set of votes such that any
multisubset of them is single-crossing.  Then we draw the required
number of votes from the domain, uniformly at random. To obtain the
domain (for candidate set $C = \{c_1, \ldots, c_m\}$), we first
generate vote $v_1 \colon c_1 \pref c_2 \pref \cdots \pref
c_m$, 
and for each $i \in [n]\setminus\{1\}$ we obtain $v_i$ by copying
$v_{i-1}$ and swapping a random pair of adjacent candidates, but so
that $v_1, \ldots, v_i$ are single-crossing (for this
order). Unfortunately, this is not a uniform sampling procedure
(obtaining a $\mathrm{P}$-time one is an open problem).

The map of a single-crossing election in Fig.~\ref{fig:microscope:extended}
shows a linear spectrum of opinions, from one vote to its
reverse. Indeed, the consecutive votes in the single-crossing domain
differ by single swaps, and this is exactly what we see.

\subsubsection*{Group-Separable Elections}
We define group-separable elections following the tree-based approach
of \citet{kar:j:group-separable} (see also the work of
\citet{elk-fal-sli:c:decloning}) rather than the original one
\citep{ina:j:group-separable,ina:j:simple-majority}.  The idea is that
candidates have features (organized hierarchically in a tree) and
voters have preferences over these features.

Let $C$ be a candidate set and let $\calT$ be a rooted, ordered tree
whose each leaf is labeled with a unique candidate (intuitively, each
internal node represents a feature and a candidate has the features
that form its path to the root).  
A vote is consistent
with $\calT$ if we can obtain it by reading the leaves of $\calT$ from
left to right after, possibly, reversing the order of some nodes'
children.

\begin{definition}
  An election is group-separable if there is a rooted, ordered tree
  $\calT$ whose each leaf is associated with a unique candidate, such
  that each vote of the election is consistent with $\calT$.
\end{definition}

For a tree $\calT$, we generate consistent elections uniformly at
random: We obtain each vote by, first, reversing the order of each
internal node's children with probability $\nicefrac{1}{2}$ and, then,
reading off the candidates from the leaves left to right.
We focus on complete binary trees (where every level except, possibly,
the last one is completely filled)
and on binary caterpillar trees (where each internal node has two
children, of which at least one is a leaf).  These trees give,
respectively, balanced and caterpillar group-separable elections.

In Fig.~\ref{fig:microscope:extended}, the group-separable elections are very
distinct from all the other ones and reflect the structures of their
trees. While it seems that they had only a few distinct votes,
this is not the case (it is known that for a binary tree with $m$
candidates, there are $2^{m-1}$ consistent votes), but many of their
votes are similar; they are less fragmented than they appear, but
there is a level of polarization
(especially in the balanced ones).

\begin{figure*}[t]
    \centering
    \includegraphics[width=0.98\linewidth]{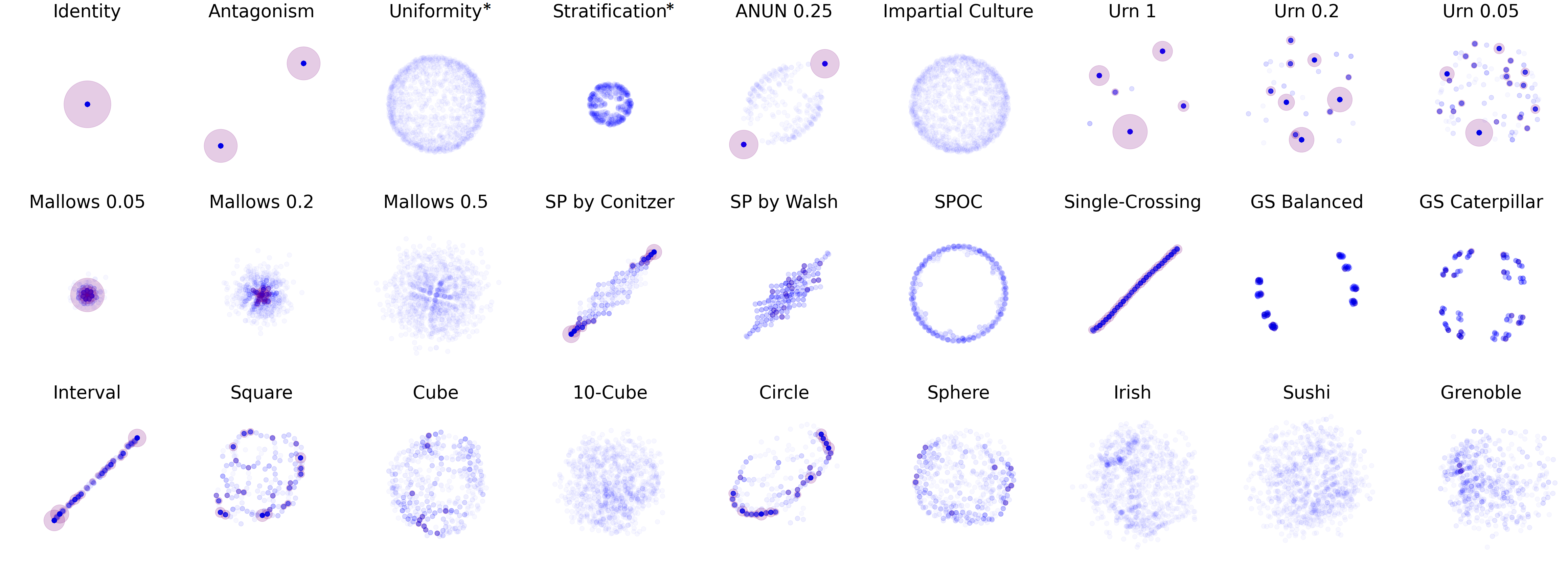}
    \caption{Maps of Preferences (8 candidates, 1000 voters).}
    \label{fig:microscope:extended}
\end{figure*}

\begin{figure}[t]
    \centering
    \includegraphics[width=0.48\textwidth]{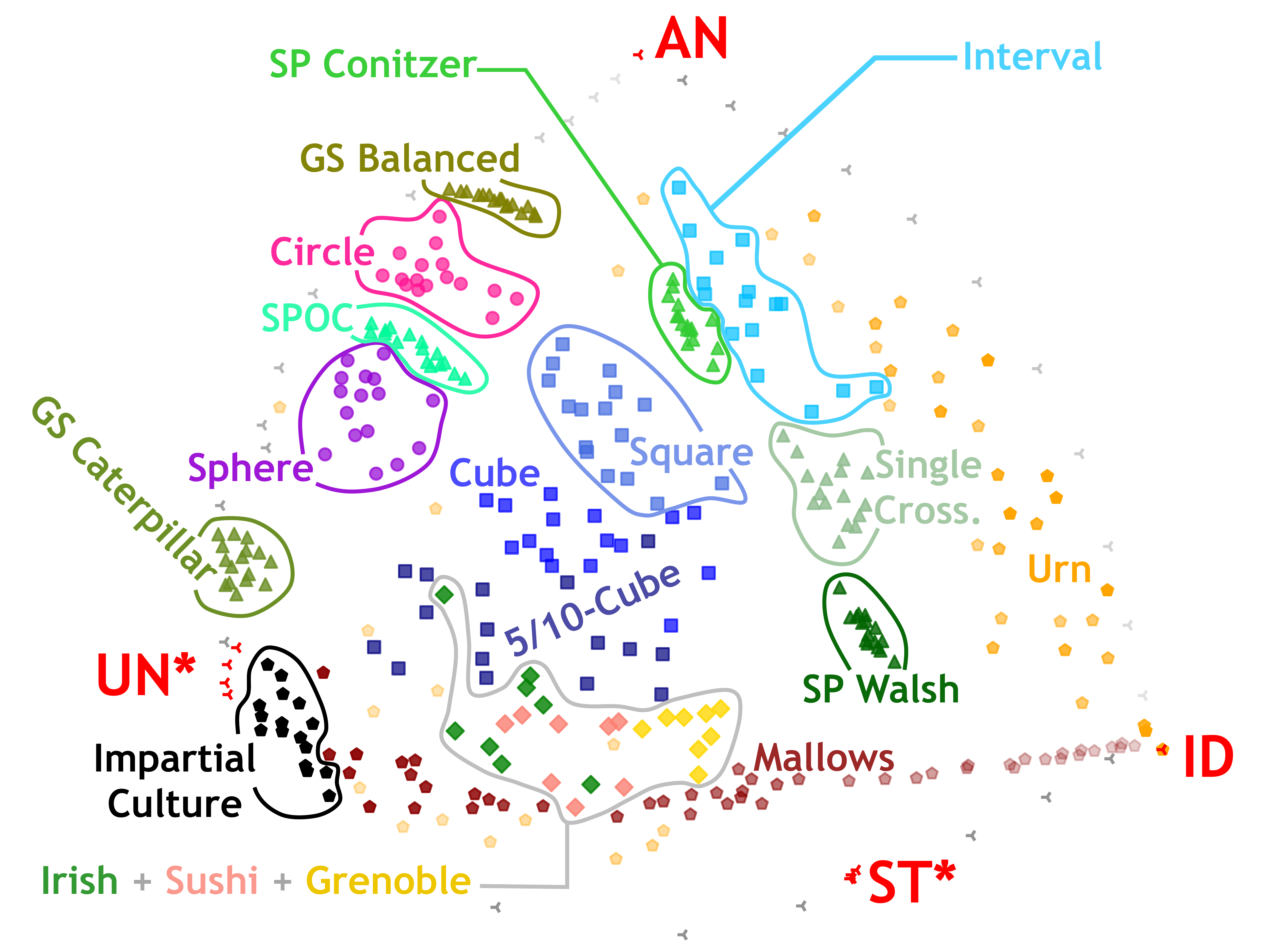}
    \caption{A map of elections in the extended dataset
    obtained using isomorphic swap distance and MDS.}
    \label{fig:swap-map:extended}
\end{figure}

\subsubsection*{Stratification}
In $\alpha$-\emph{stratification election} ($\alpha$-ST)~\citep{boe-bre-fal-nie-szu:c:compass}
the set of candidates, $C$, is partitioned into two subsets $D_1$ and $D_2$,
where the first group contains $\alpha$ fraction of candidates, i.e.,
$|D_1| / |C| = \alpha$
(if no $\alpha$ is given it is assumed that $\alpha=\nicefrac{1}{2}$).
Intuitively, in such election all voters agree that candidates $D_1$
are better than $D_2$, but all orderings of candidates inside the subsets
are equally represented.
Hence, every possible vote that ranks all candidates in $D_1$
above all candidates in $D_2$ (but with arbitrary orderings inside subsets)
appears exactly the same number of times.
However, this means that $\alpha$-stratification  election requires at least
\(
    (\alpha |C|)! \cdot ((1-\alpha)|C|)!
\)
voters.
To cope with this problem,
we consider approximated $\alpha$-stratification elections ($\alpha$-$\appST$)
that we generate using the same sampling technique as described in Appendix~\ref{app:assorted-dataset},
but whit different matrices.
In particular, for $\alpha \in \{\nicefrac{1}{8},\nicefrac{1}{4},\nicefrac{3}{8},\nicefrac{1}{2}\}$
we generate $\alpha$-$\appST$ election by
sampling an election realizing the matrix $X^\alpha$
given as follows:
{\small\begin{align*}
  X^{\nicefrac{1}{2}} &= \left[
    \begin{array}{cccccccc}
         24 & 24 & 24 & 24 & & & & \\
         24 & 24 & 24 & 24 & & & & \\
         24 & 24 & 24 & 24 & & & & \\
         24 & 24 & 24 & 24 & & & & \\
         & & & & 24 & 24 & 24 & 24 \\
         & & & & 24 & 24 & 24 & 24 \\
         & & & & 24 & 24 & 24 & 24 \\
         & & & & 24 & 24 & 24 & 24
    \end{array}
  \right],\\
    X^{\nicefrac{3}{8}} &= \left[
    \begin{array}{cccccccc}
         32 & 32 & 32 & & & & & \\
         32 & 32 & 32 & & & & & \\
         32 & 32 & 32 & & & & & \\
         & & & 20 & 19 & 19 & 19 & 19 \\
         & & & 19 & 20 & 19 & 19 & 19 \\
         & & & 19 & 19 & 20 & 19 & 19 \\
         & & & 19 & 19 & 19 & 20 & 19 \\
         & & & 19 & 19 & 19 & 19 & 20 
    \end{array}
  \right],
\end{align*}
\begin{align*}
    X^{\nicefrac{2}{8}} &= \left[
    \begin{array}{cccccccc}
         48 & 48 & & & & & & \\
         48 & 48 & & & & & & \\
         & & 16 & 16 & 16 & 16 & 16 & 16 \\
         & & 16 & 16 & 16 & 16 & 16 & 16 \\
         & & 16 & 16 & 16 & 16 & 16 & 16 \\
         & & 16 & 16 & 16 & 16 & 16 & 16 \\
         & & 16 & 16 & 16 & 16 & 16 & 16 \\
         & & 16 & 16 & 16 & 16 & 16 & 16
    \end{array}
  \right],\\
    X^{\nicefrac{1}{8}} &= \left[
    \begin{array}{cccccccc}
         96 & & & & & & & \\
         & 14 & 14 & 14 & 14 & 14 & 13 & 13 \\
         & 13 & 14 & 14 & 14 & 14 & 14 & 13 \\
         & 13 & 13 & 14 & 14 & 14 & 14 & 14 \\
         & 14 & 13 & 13 & 14 & 14 & 14 & 14 \\
         & 14 & 14 & 13 & 13 & 14 & 14 & 14 \\
         & 14 & 14 & 14 & 13 & 13 & 14 & 14 \\
         & 14 & 14 & 14 & 14 & 13 & 13 & 14 
    \end{array}
  \right].
\end{align*}
}

The map for $\appST$ election in Fig.~\ref{fig:microscope:extended} resembles a bit the map for Mallows elections,
and it also lands between \ID and IC. This is expected:
In these elections there is some agreement between the voters (they
distinguish the stronger group from the weaker one) but there is also
room for diversity.

\subsubsection*{ID-ST* mixture.}
Finally, we consider elections that capture a transition from
$\appST$ to \ID.
Specifically, instead of dividing candidates into two subsets
(aka \emph{blocks}) on ordering of which the voters agree
we divide the candidates in $k$ blocks for some $2 \le k \le |C|$.
If we choose $k=2$ we get a standard stratification election
and for $k=|C|$ we get identity.
In the extended dataset, we included one such election for each
$k \in \{3,4,6\}$.
Again, they were obtained by sampling (using procedure described in
Appendix~\ref{app:assorted-dataset}) from the position matrix $X^k$
given as follows:

{\small\begin{align*}
  X^3 &= \left[
    \begin{array}{cccccccc}
         48 & 48 & & & & & & \\
         48 & 48 & & & & & & \\
         & & 32 & 32 & 32 & & & \\
         & & 32 & 32 & 32 & & & \\
         & & 32 & 32 & 32 & & & \\
         & & & & & 32 & 32 & 32 \\
         & & & & & 32 & 32 & 32 \\
         & & & & & 32 & 32 & 32
    \end{array}
  \right],\\
    X^4 &= \left[
    \begin{array}{cccccccc}
         48 & 48 & & & & & & \\
         48 & 48 & & & & & & \\
         & & 48 & 48 & & & & \\
         & & 48 & 48 & & & & \\
         & & & & 48 & 48 & & \\
         & & & & 48 & 48 & & \\
         & & & & & & 48 & 48 \\
         & & & & & & 48 & 48 \\
    \end{array}
  \right],\\
    X^6 &= \left[
    \begin{array}{cccccccc}
         96 & & & & & & & \\
         & 48 & 48 & & & & & \\
         & 48 & 48 & & & & & \\
         & & & 48 & 48 & & & \\
         & & & 48 & 48 & & & \\
         & & & & & 96 & & \\
         & & & & & & 96 & \\
         & & & & & & & 96 \\
    \end{array}
  \right].
\end{align*}}
\noindent
The order of the larger and smaller blocks in matrices $X^3$ and $X^6$
was chosen randomly.

\section{Mallows Dataset}
\label{app:mallows_dataset}
In this section, we introduce the Mallows dataset.

\paragraph{Mallows Mixture Model.}
Mallows mixture model is parameterized by the central vote $u$, 
$\normphi \in [0,1]$, and mixing parameter $\omega \in [0,0.5]$. We
generate votes as follows: With probability $1-\omega$, we use the
Mallows model with central vote $u$ and parameter $\normphi$, and with
probability $\omega$ we use $\normphi$ and the reversed central vote.
Observe that for $\omega=0$ this gives a standard Mallows model
as described in Section~\ref{sec:cultures}
(we speak then of \emph{pure} Mallows election).

Let us analyze maps of preferences for Mallows mixture model
as seen in Fig.~\ref{fig:microscope:mallows}.
As noted in Section~\ref{sec:cultures},
pure Mallows elections
form a spectrum between \ID and IC.
However, for
$\omega \in \{0.25, 0.5\}$, polarization appears (the maps for
$\omega \in \{0.25, 0.5\}$ show how the central vote and its reverse
are at maximum swap distance and their noisy incarnations are closer
to each other).
Note that for $\omega=0.5$ and $\normphi=0$,
the election we obtain is basically \AN election
(with possible random fluctuation in the sizes of the opposite groups).

\paragraph{The Mallows Dataset Composition.}
The Mallows dataset includes elections (with 8 candidates and 96 voters) generated from mixtures of Mallows
models (and the four special elections, \ID, \AN, $\appUN$, and $\appST$, for
orientation).
We present this dataset on map of elections in Fig.~\ref{fig:swap-map:mallows}.
There, each dot represents an election generated
from the Mallows mixture model with $\omega$ drawn uniformly at random
from $[0,0.5]$ and $\normphi$ drawn from $[0,1]$ in such
a way that
$\mathbb{P}[1 - \normphi \le x] = x^2$.  This 
allows us to avoid 
high congestion of
elections near $\appUN$ (intuitively, we can think of one minus
$\normphi$ as a distance from $\appUN$ and of $\omega$ as a direction in
which we move away from $\appUN$---%
by taking the probability of the distance proportional to its square,
we ensure the uniform distribution of the dots on the map).

\begin{figure}[t]
    \centering
    \begin{subfigure}[t]{0.5\textwidth}
        \includegraphics[width = \textwidth]{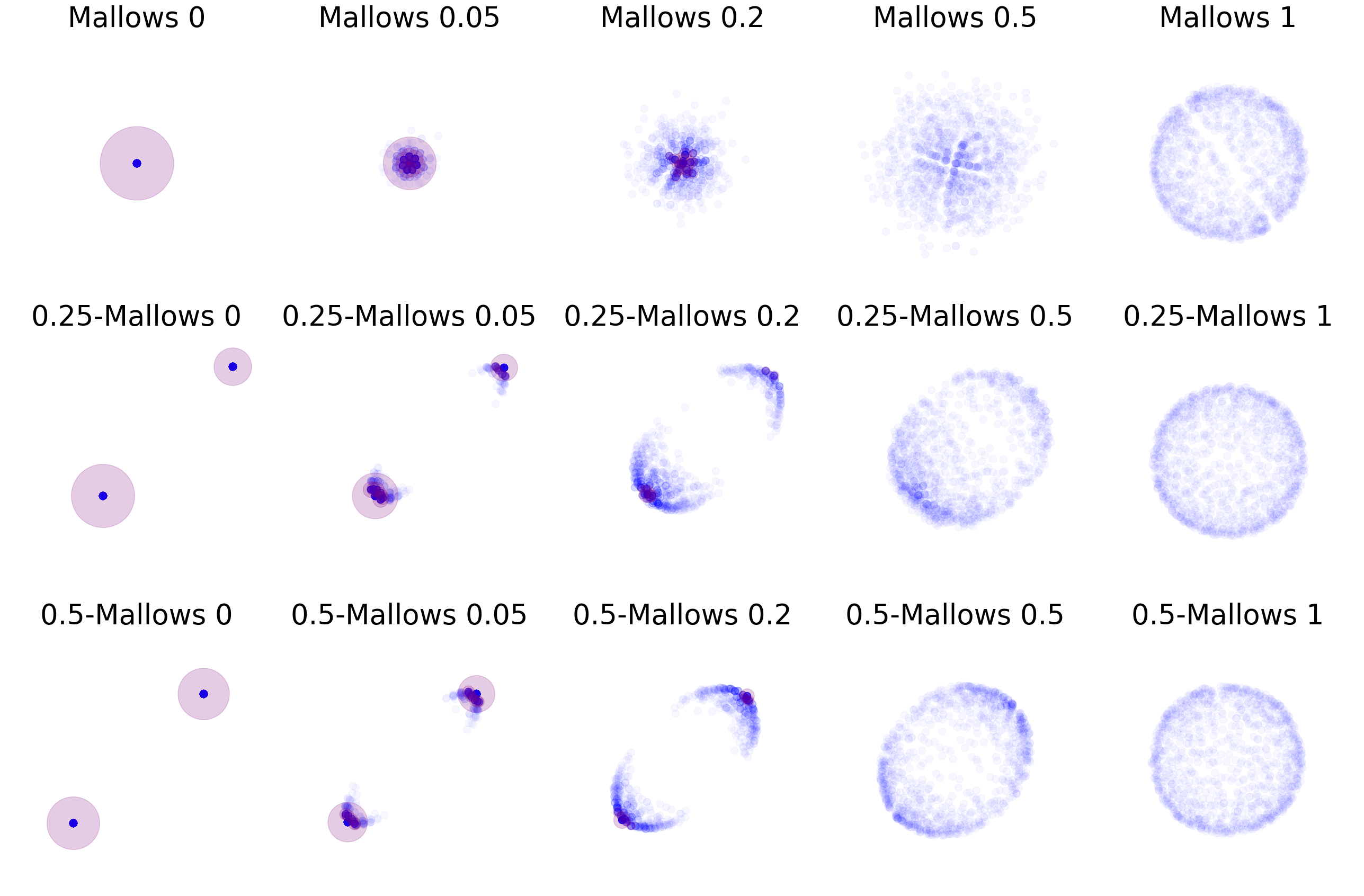}
        \caption{Maps of Preferences (8 candidates, 1000 voters). A title ``$\boldsymbol{x}$-Mallows $\boldsymbol{y}$'' denotes an election from Mallows mixture
        model with norm-$\boldsymbol{\phi = y}$ and $\boldsymbol{\omega = x}$. We omit $\boldsymbol{x}$- when $\omega=0$.}
        \label{fig:microscope:mallows}
    \end{subfigure}
    \hfill
    \begin{subfigure}[t]{0.46\textwidth}
        \includegraphics[width=\textwidth]{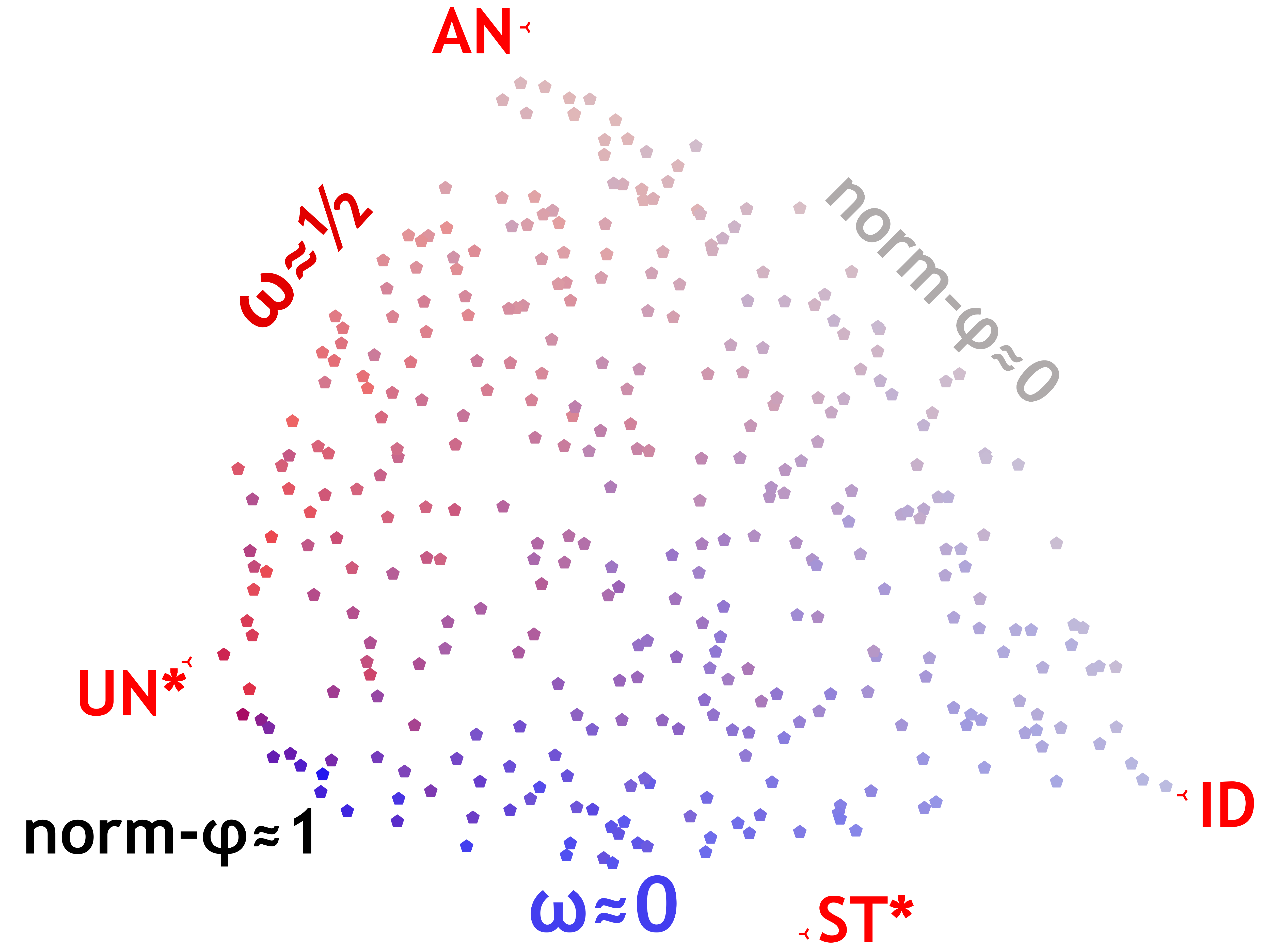}
        \caption{A map of election in the Mallows dataset obtained using isomorphic swap distance and MDS. The color corresponds to $\omega$ parameter (blue for $\omega=0$ and red for $\omega=0.5$) and the color intensity is greater for greater $\normphi$ values.}
        \label{fig:swap-map:mallows}
    \end{subfigure}
    \caption{Mallows mixture model.}
\end{figure}

\section{Maps of Preferences}
\label{app:microscope_96}
In this section,
we present analogues of the pictures
form Fig.~\ref{fig:microscope:extended}
(hence, including all elections from Fig.~\ref{fig:microscope}),
but for elections with 8 candidates and 96 voters.
The method through which it is obtained is exactly the same,
i.e., first we compute swap distance between every pair of votes
in an election, and then we project the votes
onto a 2D plane using MDS.
The results are presented in Fig.~\ref{fig:microscope_96}.

\begin{figure*}[b]
    \centering
    \includegraphics[width = \linewidth]{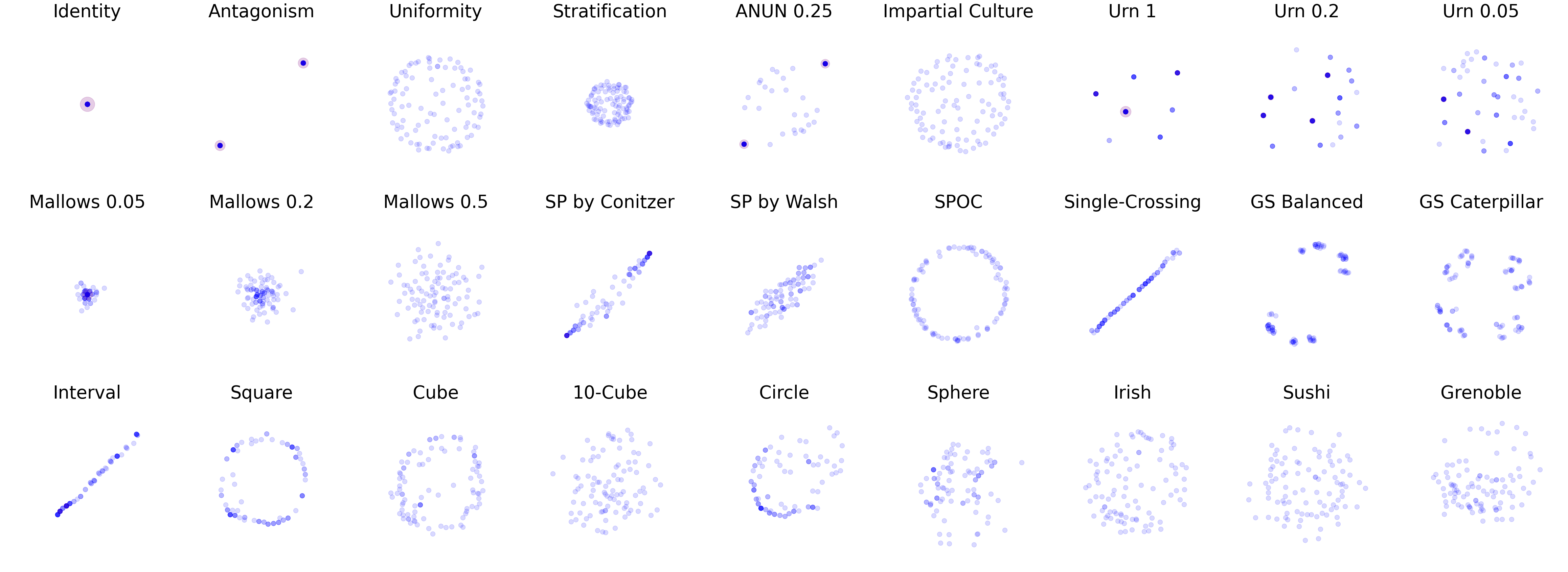}
    \caption{Maps of Preferences (8 candidates, 96 voters).}
    \label{fig:microscope_96}
\end{figure*}

\begin{figure*}[t]
    \centering
    \begin{tabular}{ccc}
        \multicolumn{3}{c}{{\small Improvement of the combined heuristic over greedy approach}} \\
        \begin{subfigure}[t]{0.32\textwidth}
            \centering
            \includegraphics[width=\textwidth]{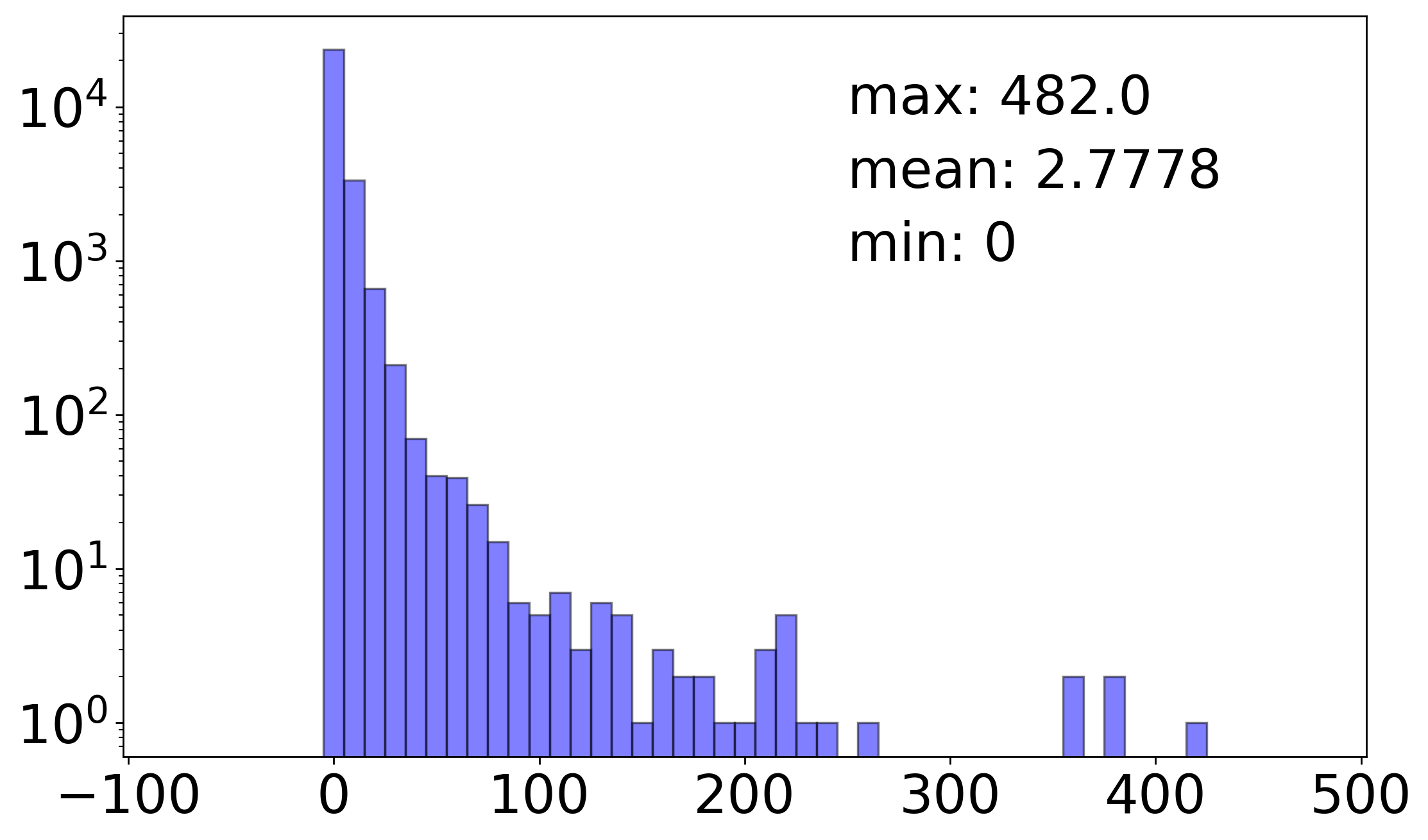}
        \end{subfigure}    &  
        \begin{subfigure}[t]{0.32\textwidth}
            \centering
            \includegraphics[width=\textwidth]{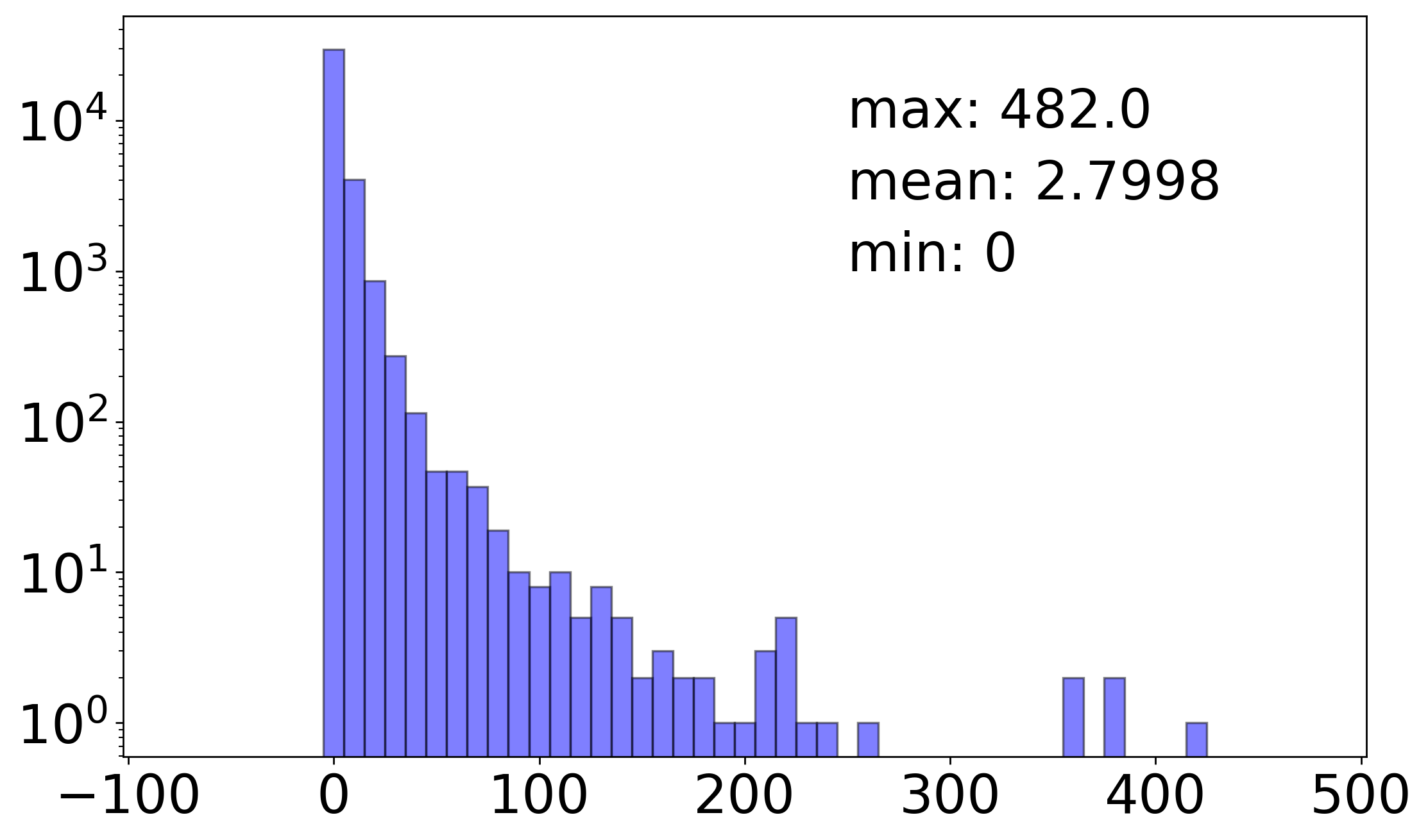}
        \end{subfigure}    &  
        \begin{subfigure}[t]{0.32\textwidth}
            \centering
            \includegraphics[width=\textwidth]{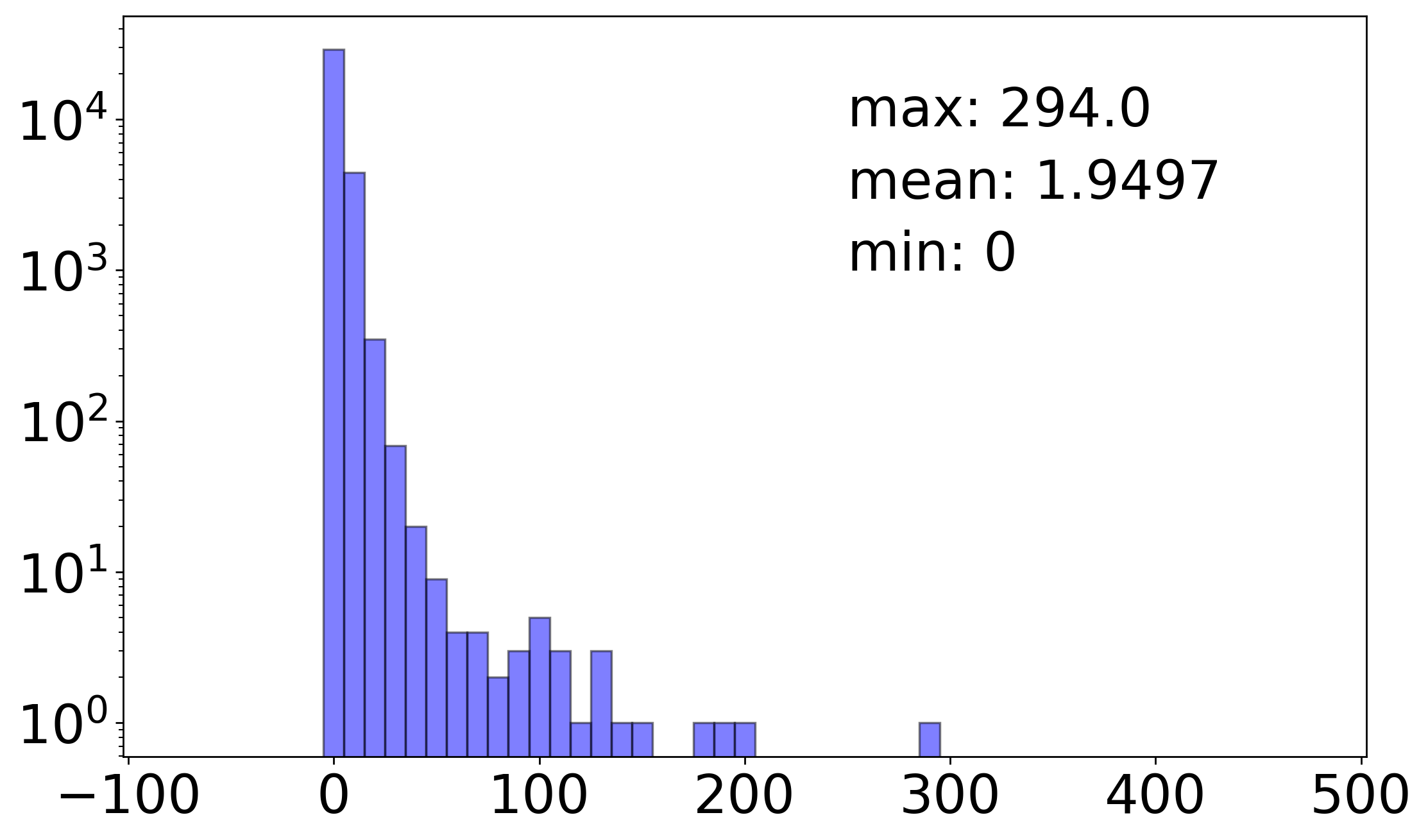}
        \end{subfigure}\\
        \multicolumn{3}{c}{{\small Improvement of local search over greedy approach}}\\
        \begin{subfigure}[t]{0.32\textwidth}
            \centering
            \includegraphics[width=\textwidth]{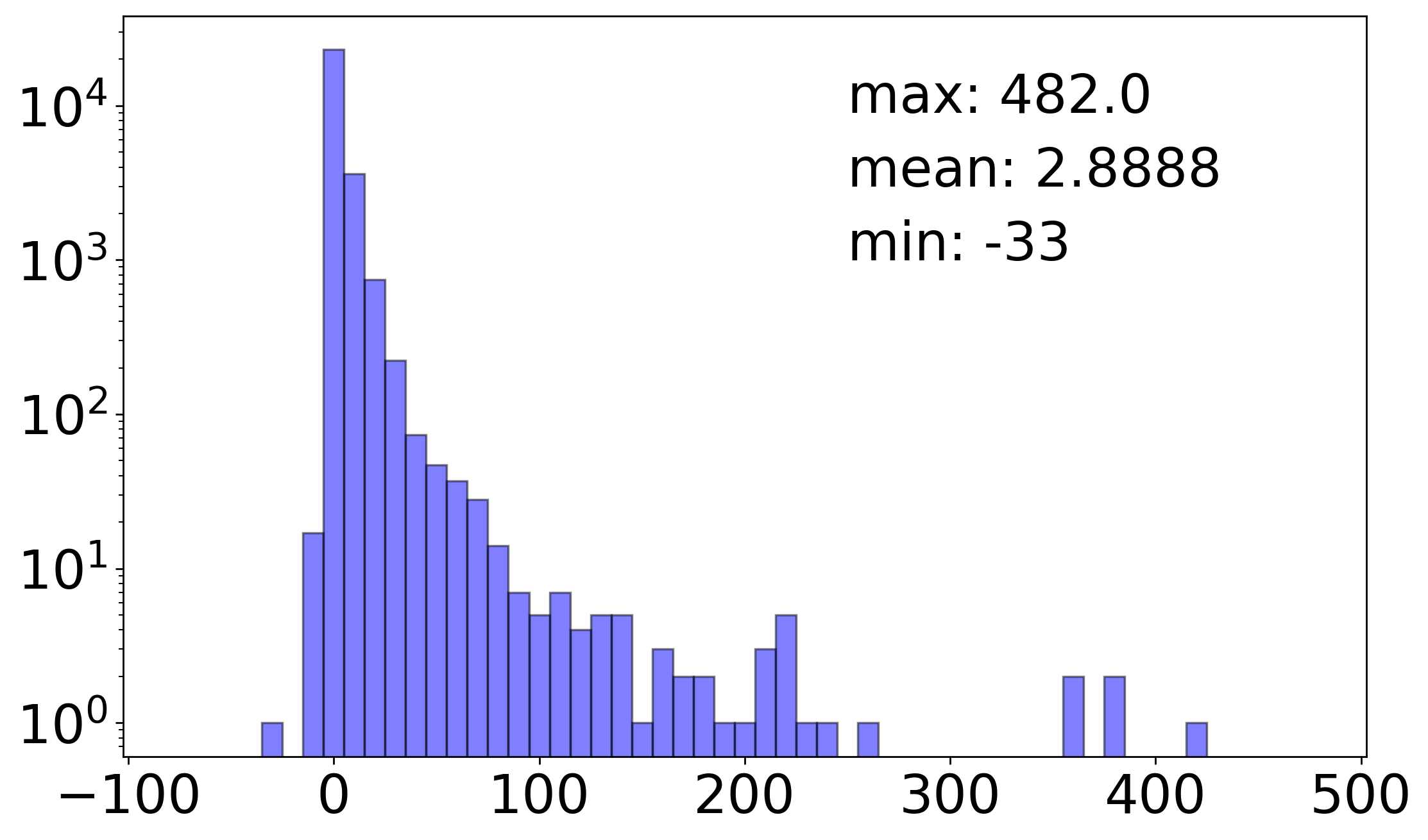}
        \end{subfigure}    &  
        \begin{subfigure}[t]{0.32\textwidth}
            \centering
            \includegraphics[width=\textwidth]{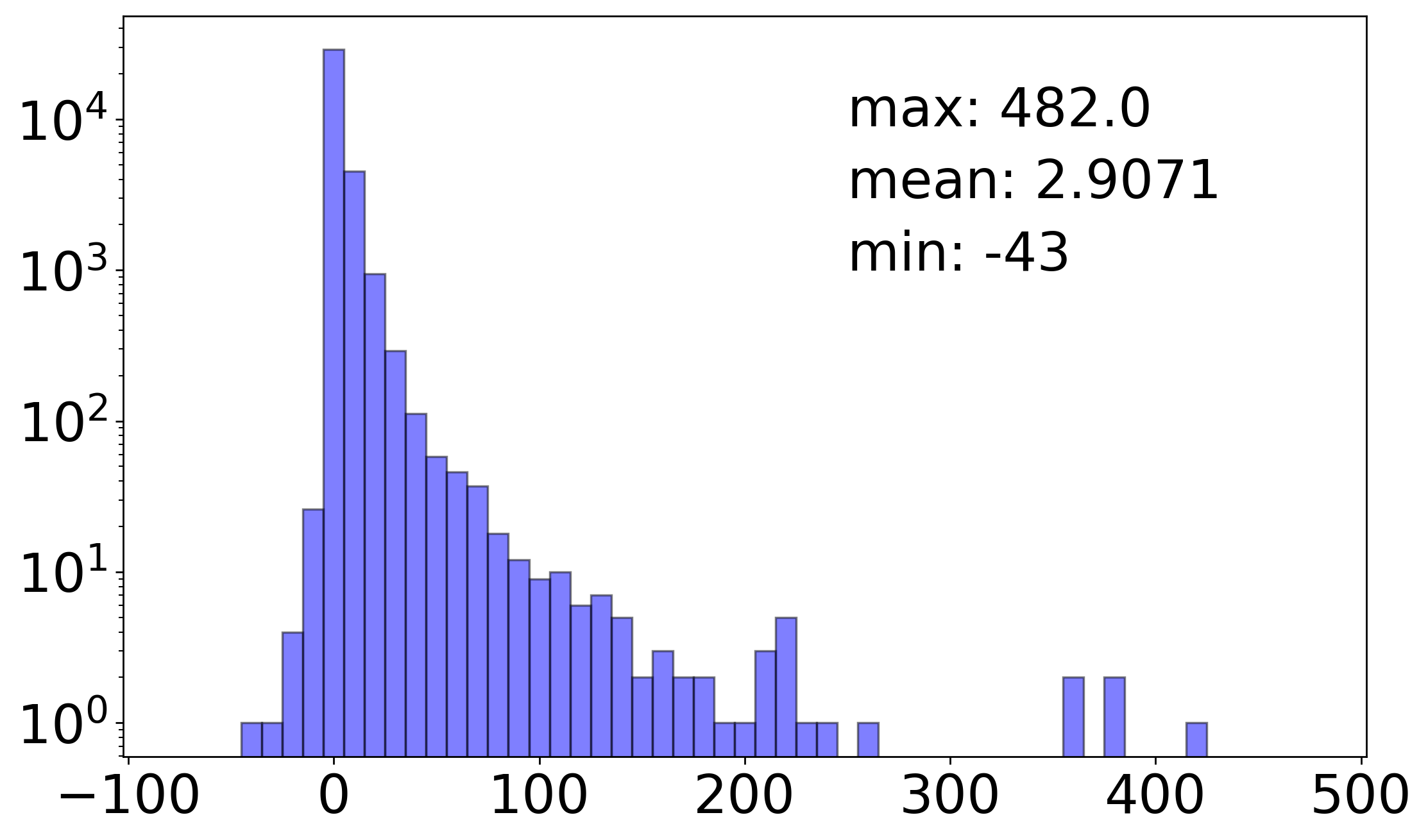}
        \end{subfigure}    &  
        \begin{subfigure}[t]{0.32\textwidth}
            \centering
            \includegraphics[width=\textwidth]{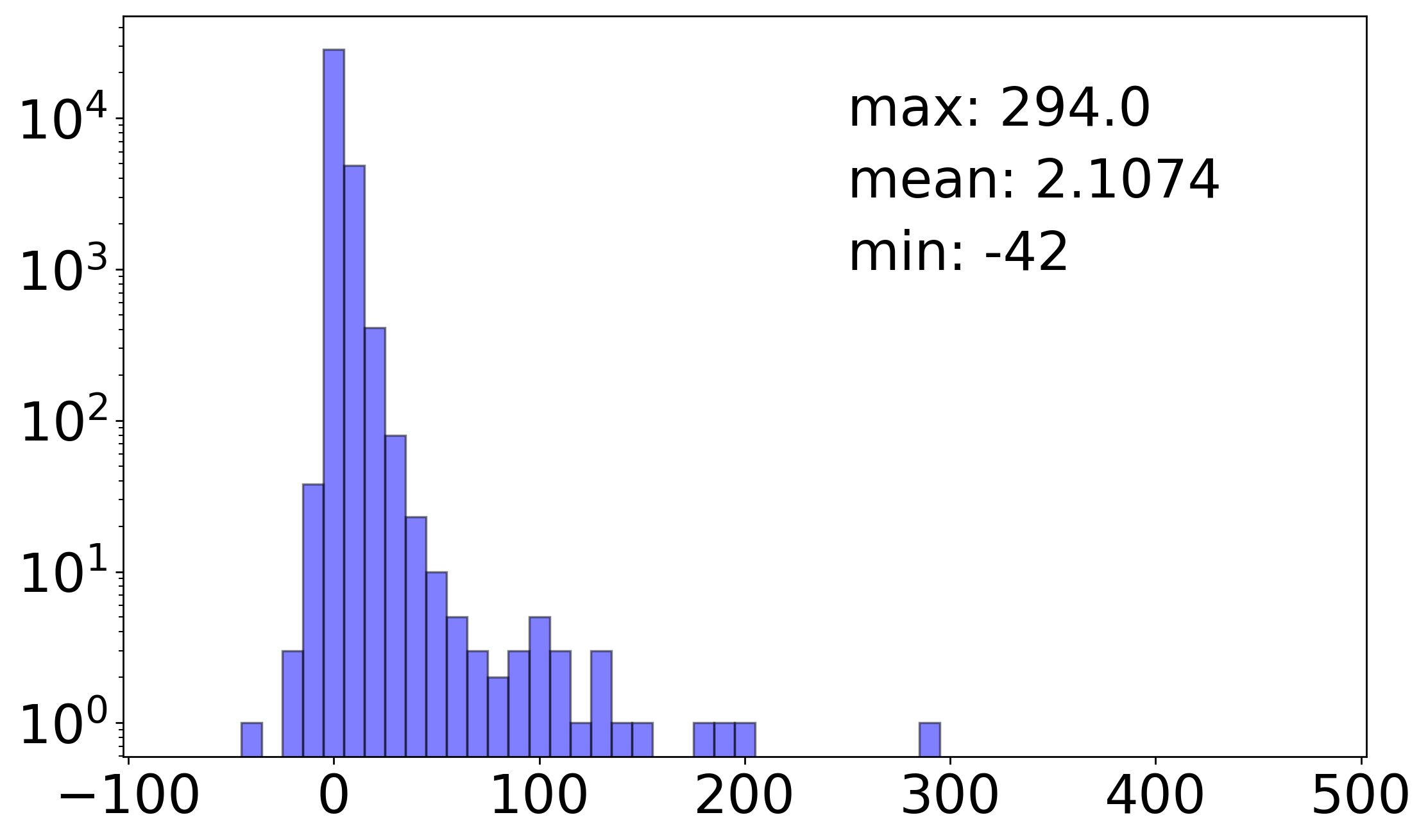}
        \end{subfigure}\\
        \multicolumn{3}{c}{{\small Improvement of the combined heuristic over local search}}\\
        \begin{subfigure}[t]{0.32\textwidth}
            \centering
            \includegraphics[width=\textwidth]{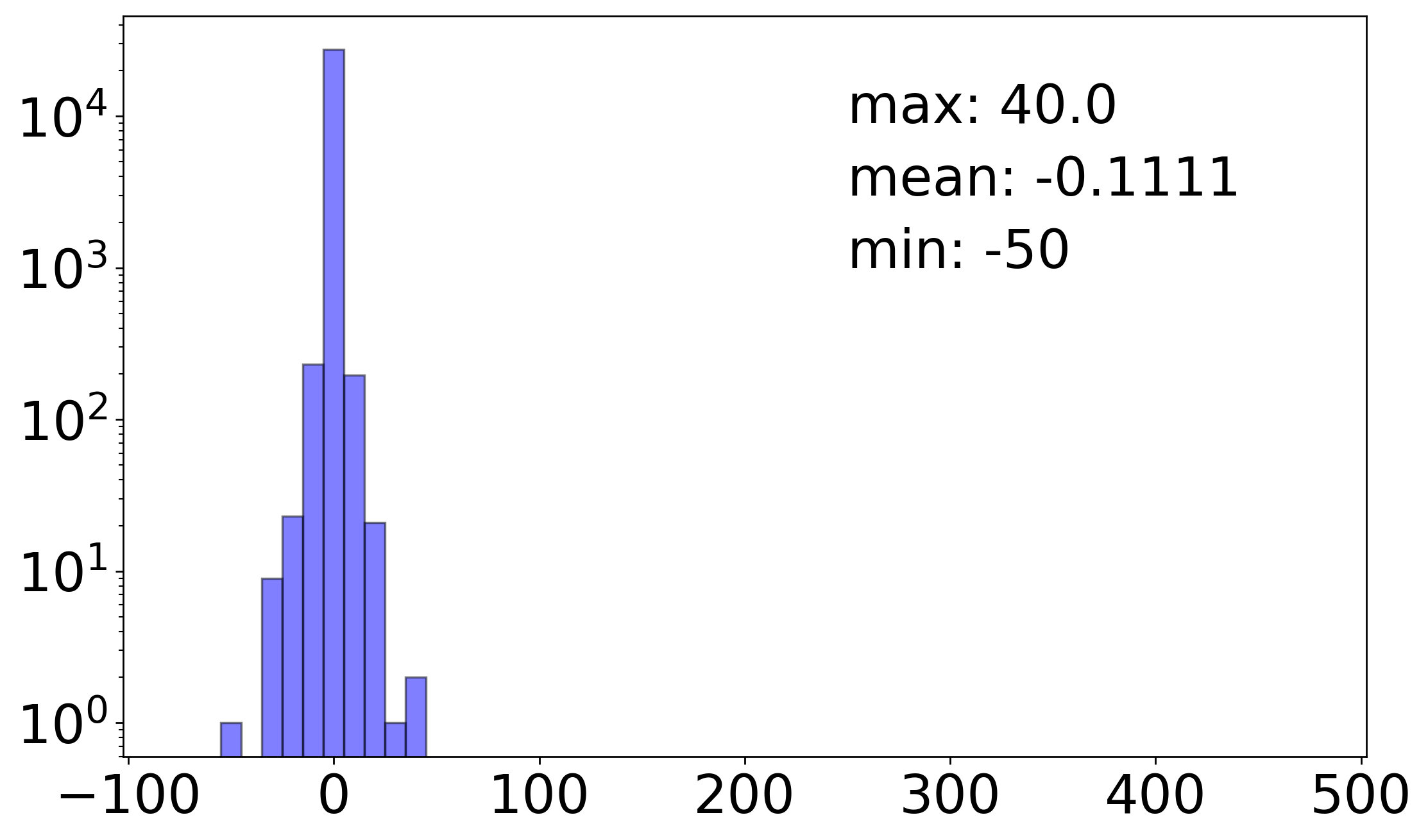}
        \end{subfigure}    &  
        \begin{subfigure}[t]{0.32\textwidth}
            \centering
            \includegraphics[width=\textwidth]{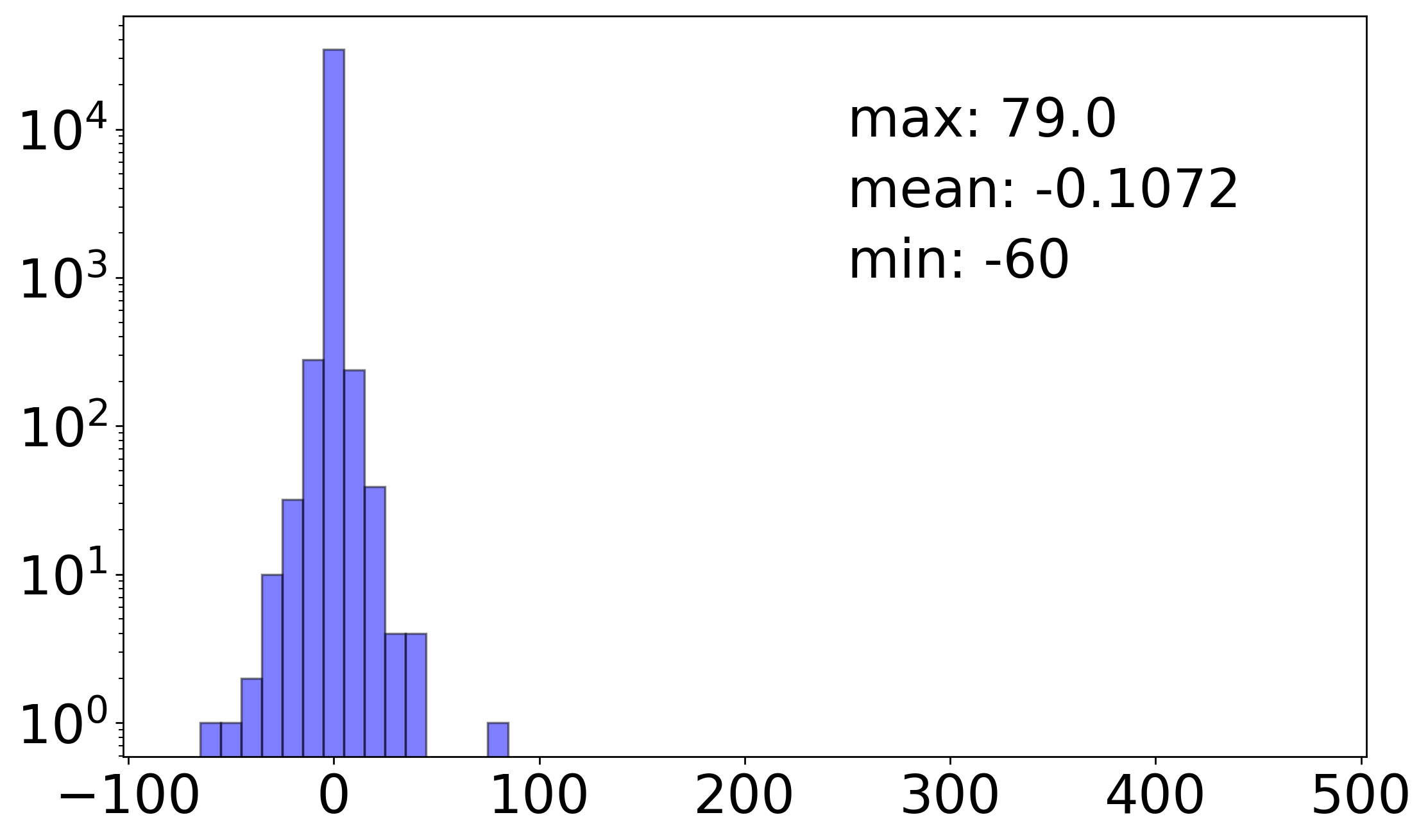}
        \end{subfigure}    &  
        \begin{subfigure}[t]{0.32\textwidth}
            \centering
            \includegraphics[width=\textwidth]{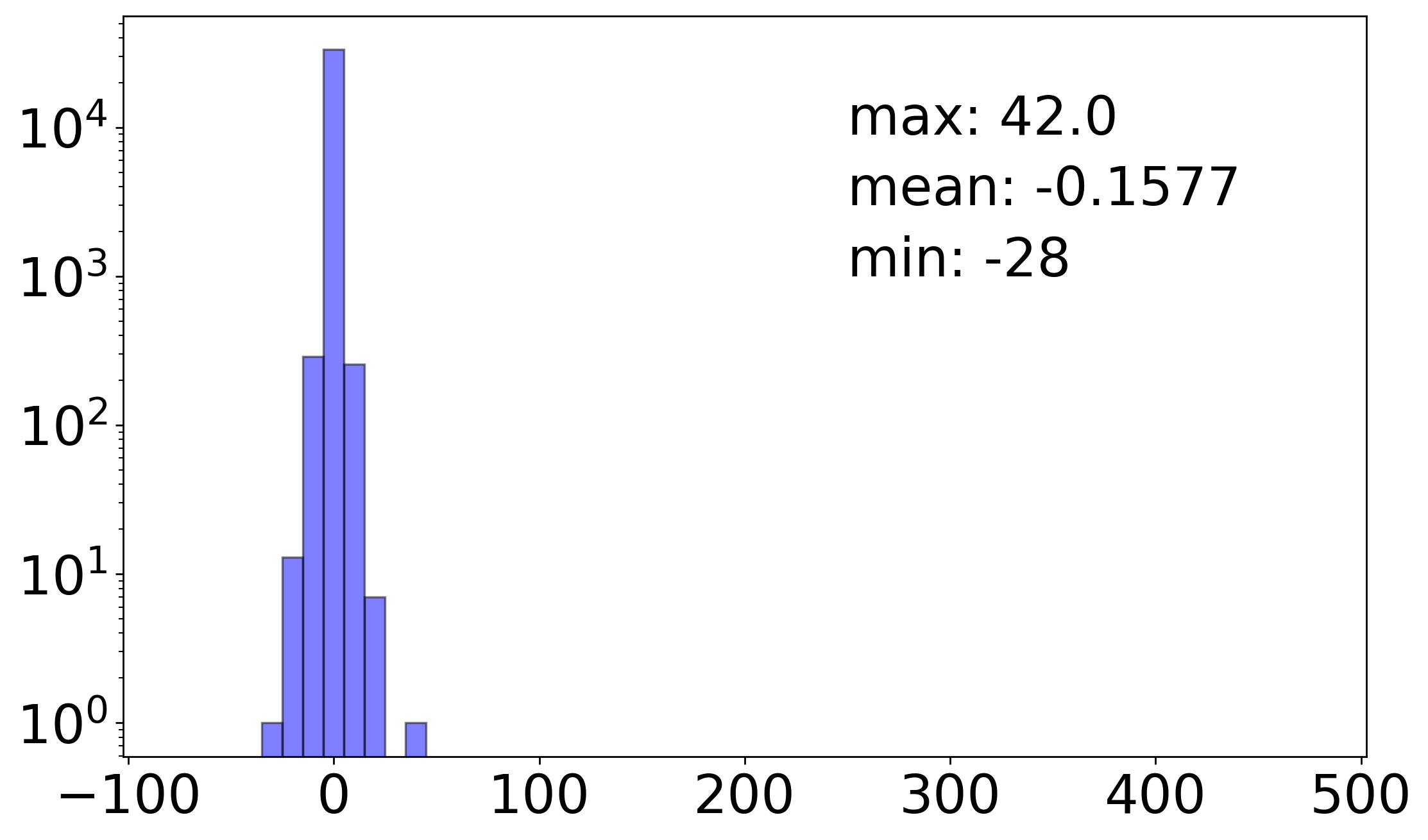}
        \end{subfigure}\\
        Standard dataset & Extended dataset & Mallows dataset
    \end{tabular}
    \caption{The histograms of differences in $k$-Kemeny distances returned by our three algorithms
    (note that the counts are in the logarithmic scale).
    The first column corresponds to the standard dataset,
    the second one to the extended dataset,
    and the third to the Mallows dataset.
    The maximal, average, and minimal value in each case is given.}
    \label{fig:histograms}
\end{figure*}
\section[k-Kemeny Computation Methods]{$\boldsymbol{k}$-Kemeny Computation Methods}
\label{app:kkemeny-computation}
In this section,
we present our experiment
comparing three methods of computing
$k$-Kemeny distances: the greedy approach,
the local search, and the combined heuristic.

For each election in all three of our datasets
and every $k \in [96]$, we calculated
$k$-Kemeny distance using our three methods.
Then, we looked at the differences between the reported values.
The histograms of differences for all three pairs of methods and
all three datasets are presented in Fig.~\ref{fig:histograms}.
We note that in the majority of cases
all three methods returned exactly the same distance.
However, in other cases, the differences between
the reported $k$-Kemeny distance was significant,
especially if we compare the combined heuristic or the local search
against the greedy approach.
In particular, maximal difference, 482,
is observed for an election
that is a mixture of \AN and $\appUN$,
where voters characteristic for \AN dominate by far
(see Appendix~\ref{app:assorted-dataset} for the definition).
Hence, we have two large groups of voters with exactly opposing preferences
and few approximately uniformly spread votes (see Fig.~\ref{fig:microscope_96}
(5th column, 1st row) for an illustration).
Then, for $k=2$, the greedy algorithm first chooses a vote that is somewhere in the middle,
and then a vote that belongs to one of the two opposing groups.
However, we obtain much smaller total distance,
if we just set the rankings at the preference orders of the two opposing groups
what both the local search and the combined heuristic managed to do.

The differences between local search and the combined heuristic are comparatively very small.
On average, local search performed better than the combined heuristic,
but the difference is too small to draw any conclusions.
Hence, in our further calculations
we simply took the better of the outcomes produced by either of these two methods.

\section{Plots}
\label{app:plots}
In this section, we present the values of all three indices
for elections in our datasets.
We do it in three figures
(see their captions for details):
\begin{itemize}
    \item Fig.~\ref{fig:plots} presents plots on which every election is a dot with x/y coordinates corresponding to the values of two out of three of our indices
    (for every pair of indices).
    We include also their affine transformations to show the resemblance 
    to maps of elections from Figs.~\ref{fig:swap-map:standard},
    \ref{fig:swap-map:extended}, and~\ref{fig:swap-map:mallows}.
    \item Fig.~\ref{fig:tenmaps} shows the maps of elections in which the colors of the dots correspond to the values our indices.
    \item Fig.~\ref{fig:correlations:all} depicts the correlation between the values of agreement, diversity, and polarization and the distance from \ID, $\appUN$, and \AN, respectively.
\end{itemize}

\begin{figure*}[t]
     \centering
     \begin{subfigure}[t]{0.177\textwidth}
         \centering
         \includegraphics[width=\textwidth]{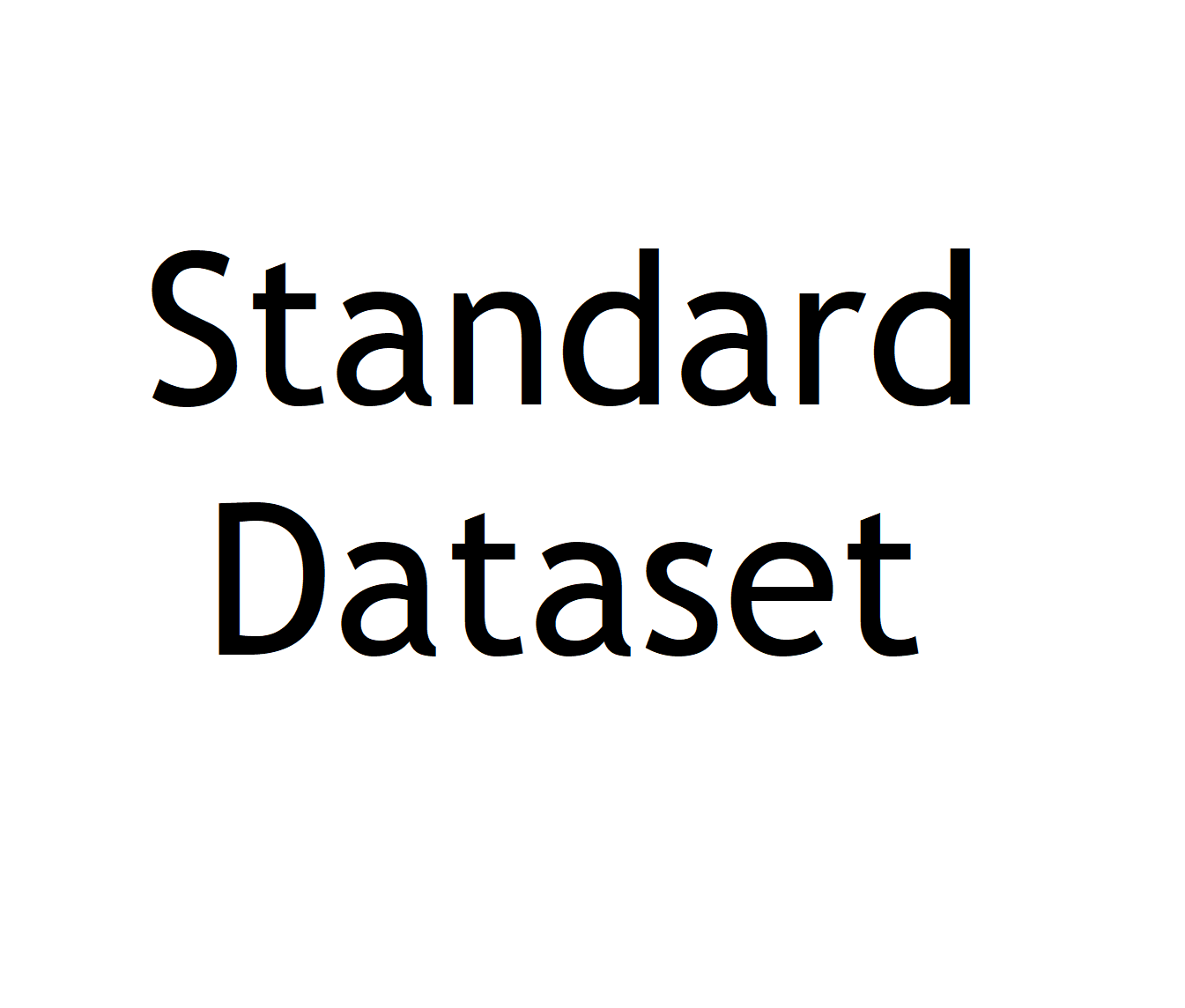}
     \end{subfigure}
     \begin{subfigure}[t]{0.222\textwidth}
         \centering
         \includegraphics[width=\textwidth]{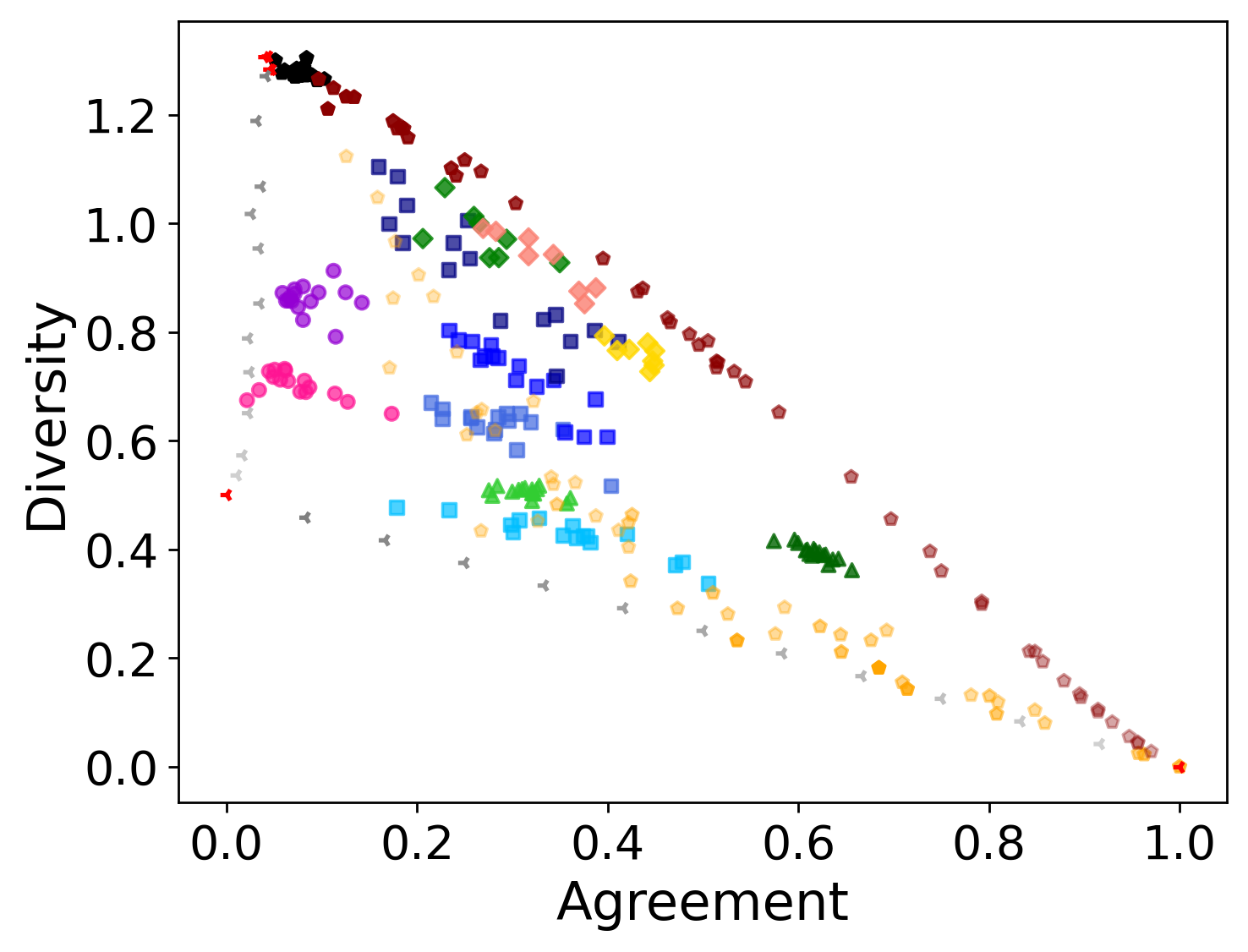}
     \end{subfigure}
     \begin{subfigure}[t]{0.222\textwidth}
         \centering
         \includegraphics[width=\textwidth]{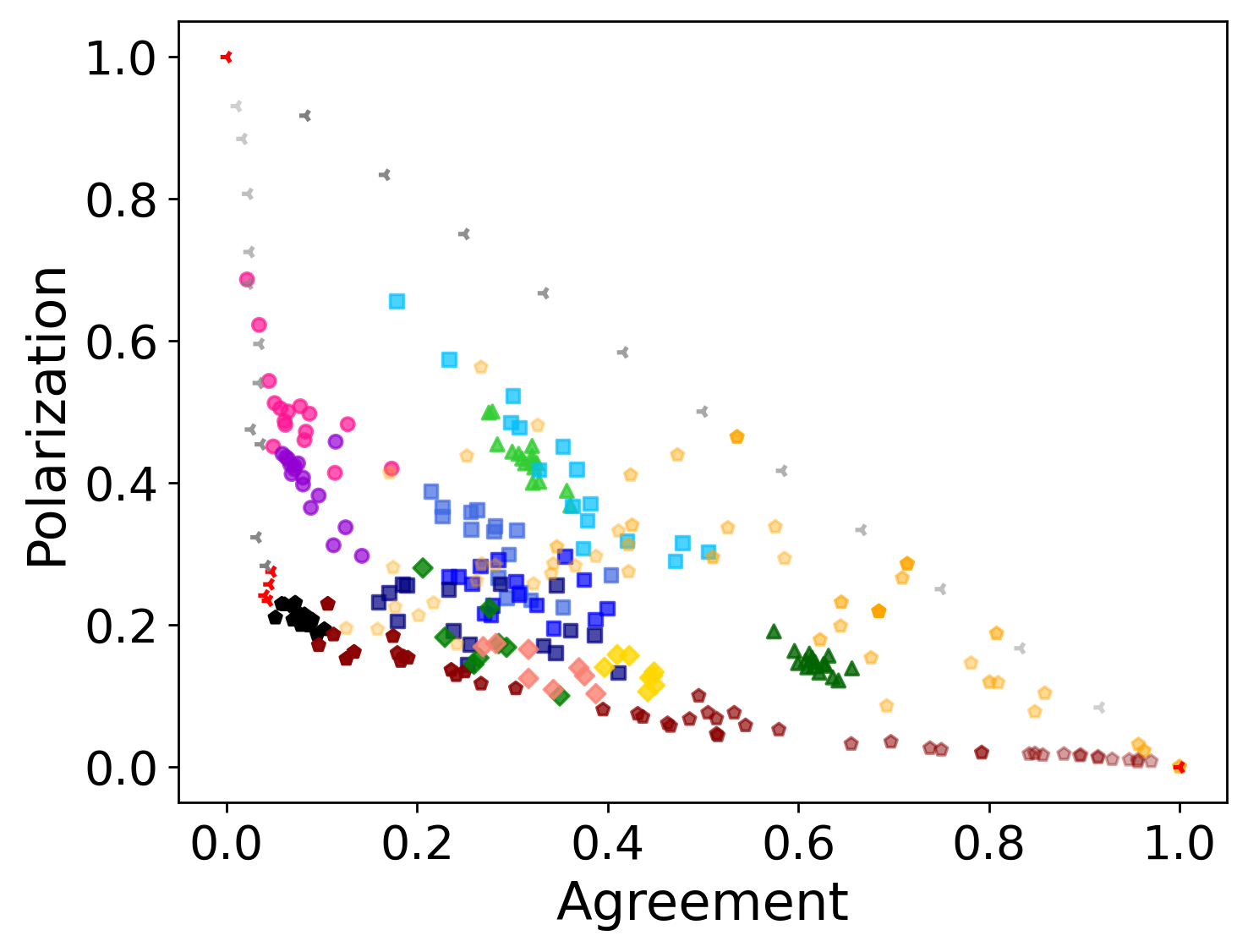}
     \end{subfigure}
     \begin{subfigure}[t]{0.222\textwidth}
         \centering
         \includegraphics[width=\textwidth]{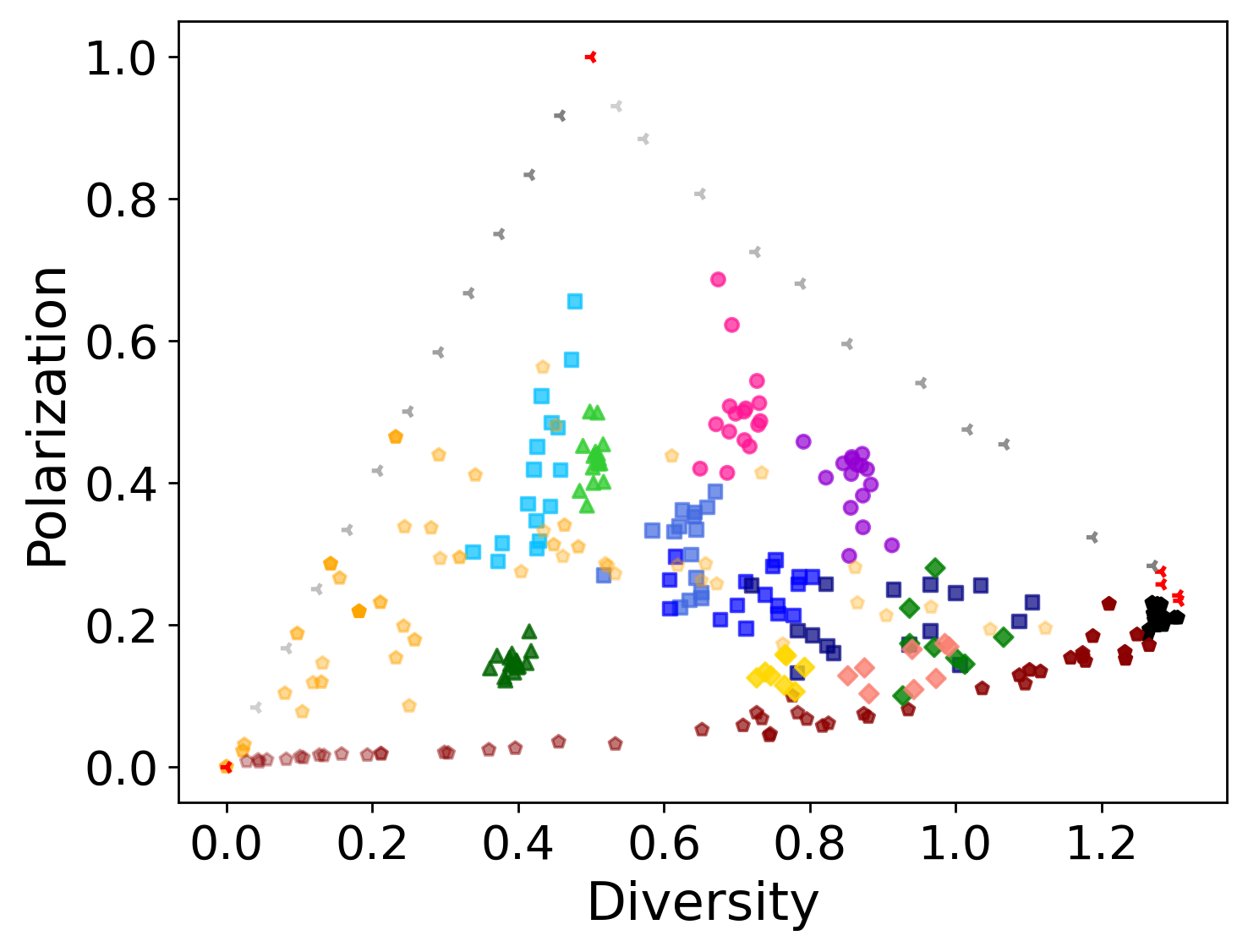}
     \end{subfigure}
\end{figure*}
\begin{figure*}[t]\ContinuedFloat
     \centering
     \begin{subfigure}[t]{0.177\textwidth}
         \centering
         \includegraphics[width=\textwidth]{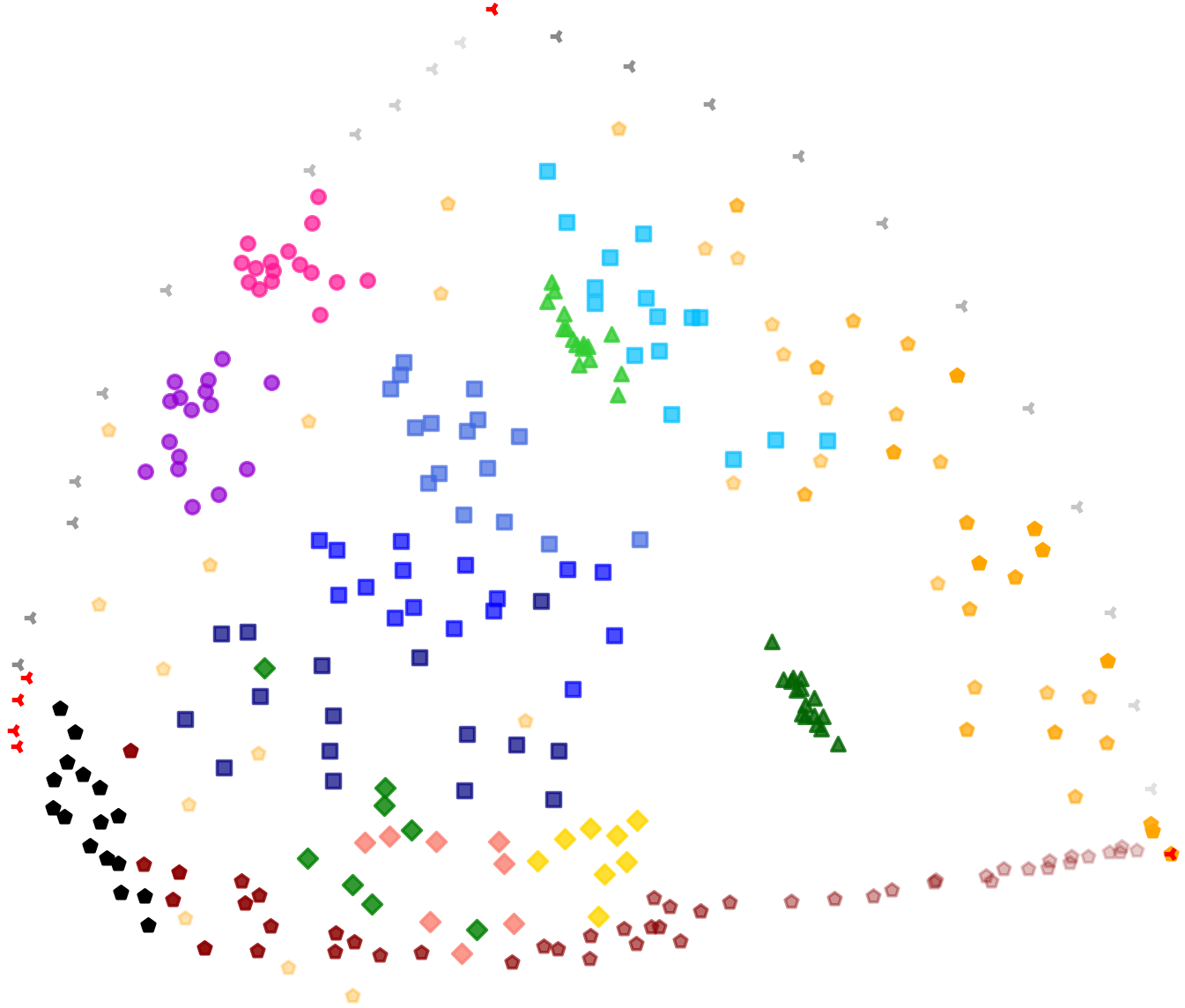}
     \end{subfigure}
     \begin{subfigure}[t]{0.222\textwidth}
         \centering
         \includegraphics[width=\textwidth]{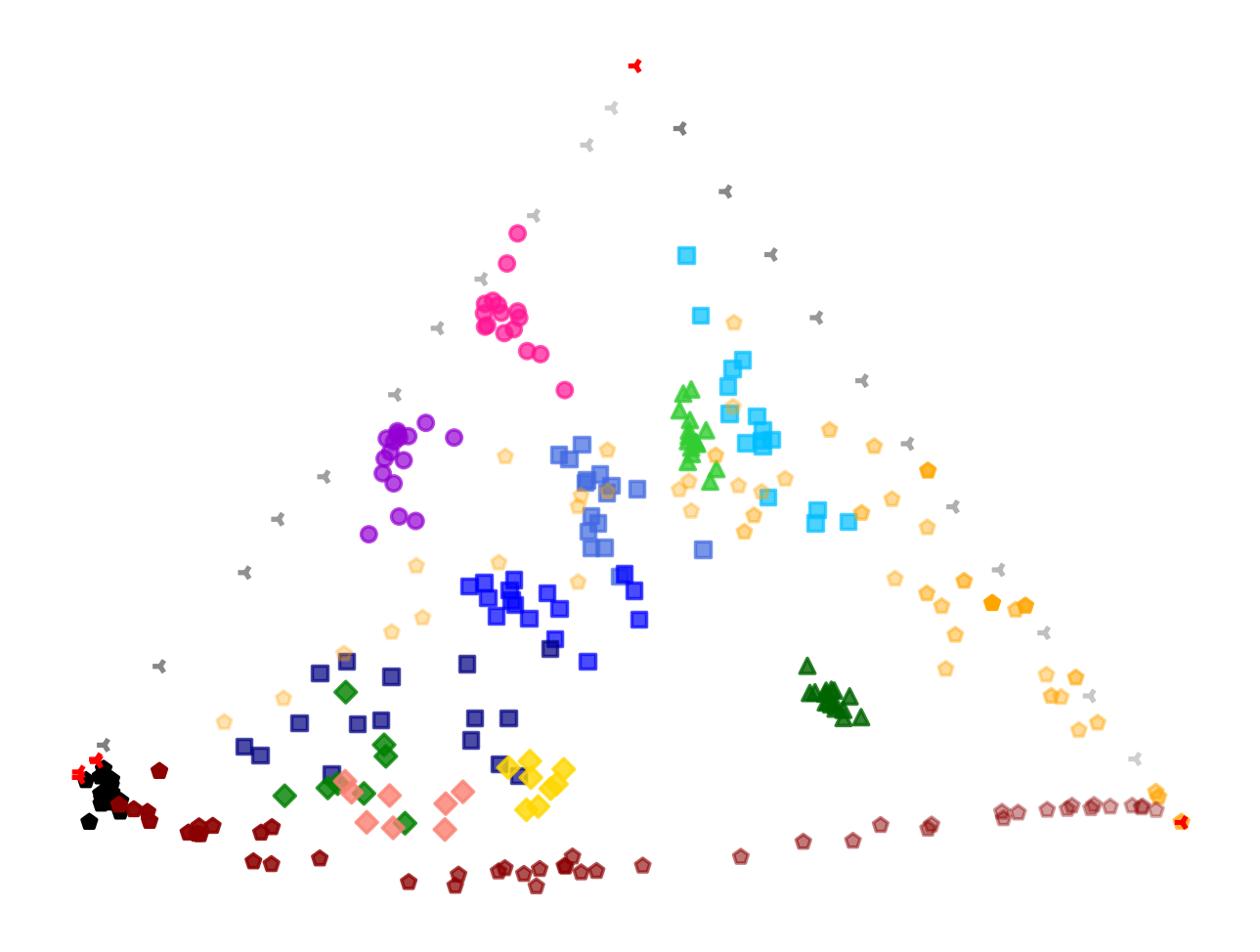}
     \end{subfigure}
     \begin{subfigure}[t]{0.222\textwidth}
         \centering
         \includegraphics[width=\textwidth]{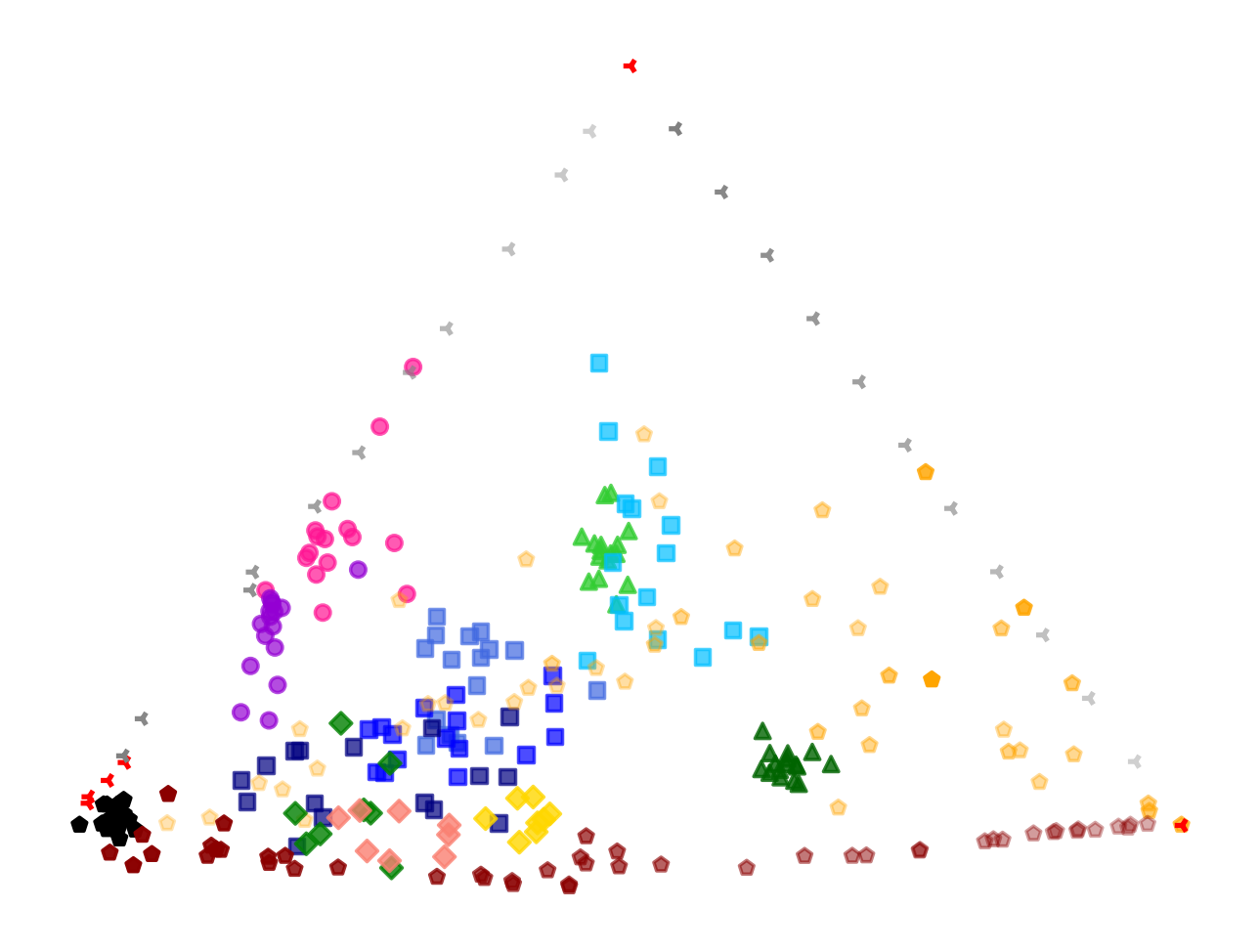}
     \end{subfigure}
     \begin{subfigure}[t]{0.222\textwidth}
         \centering
         \includegraphics[width=\textwidth]{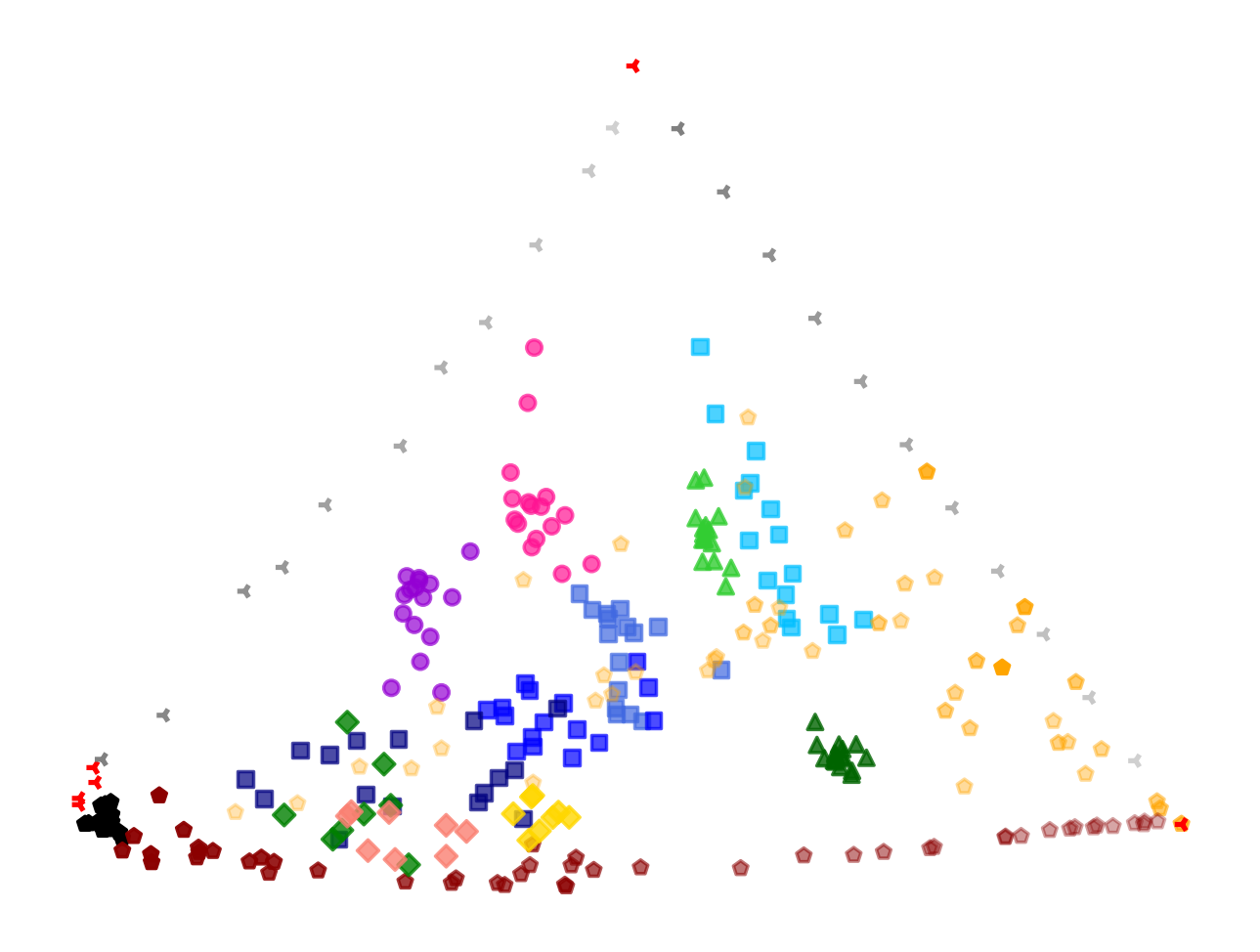}
     \end{subfigure}
\end{figure*}
\begin{figure*}[t]\ContinuedFloat
     \centering
     \begin{subfigure}[t]{0.177\textwidth}
         \centering
         \includegraphics[width=\textwidth]{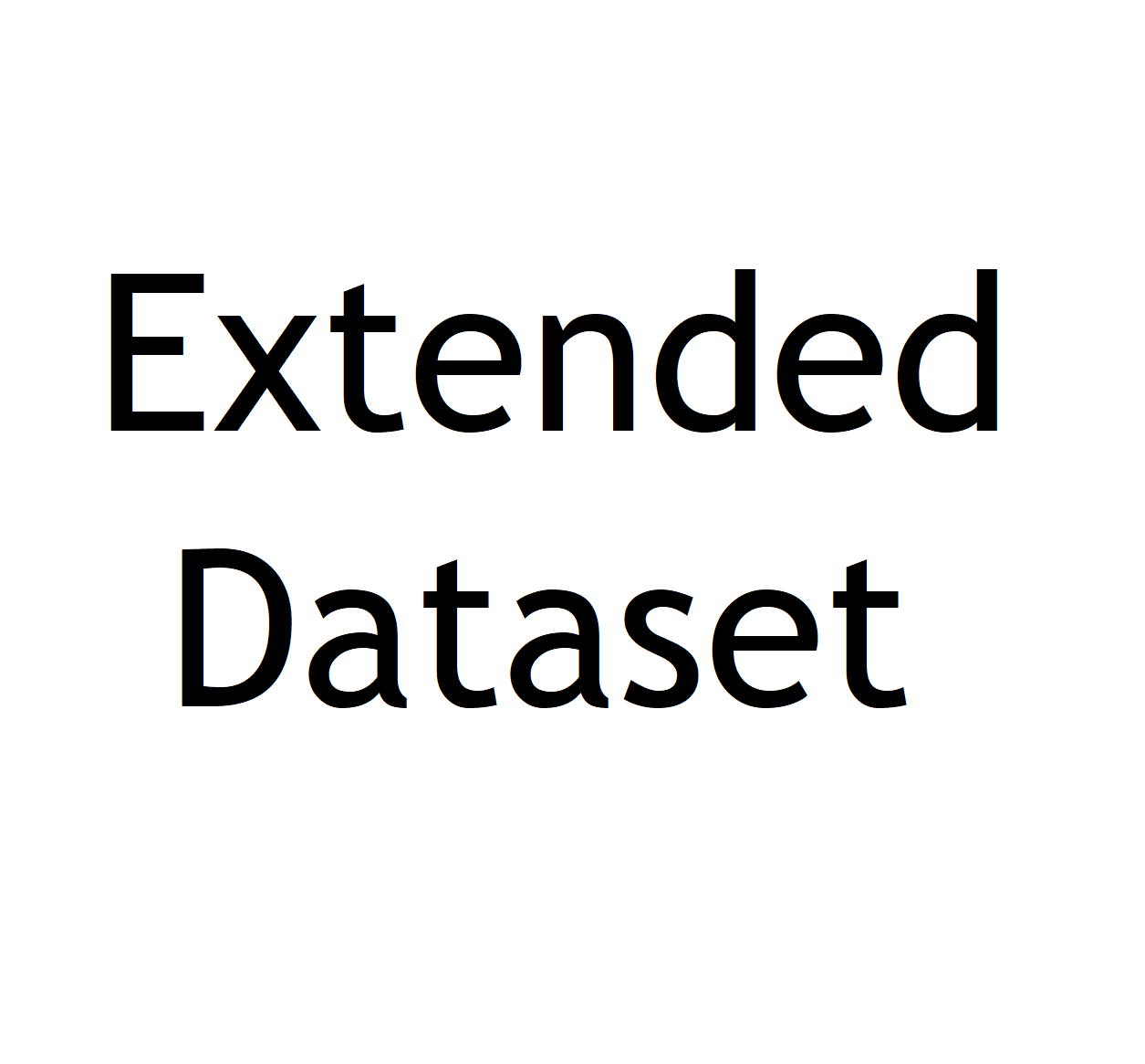}
     \end{subfigure}
     \begin{subfigure}[t]{0.222\textwidth}
         \centering
         \includegraphics[width=\textwidth]{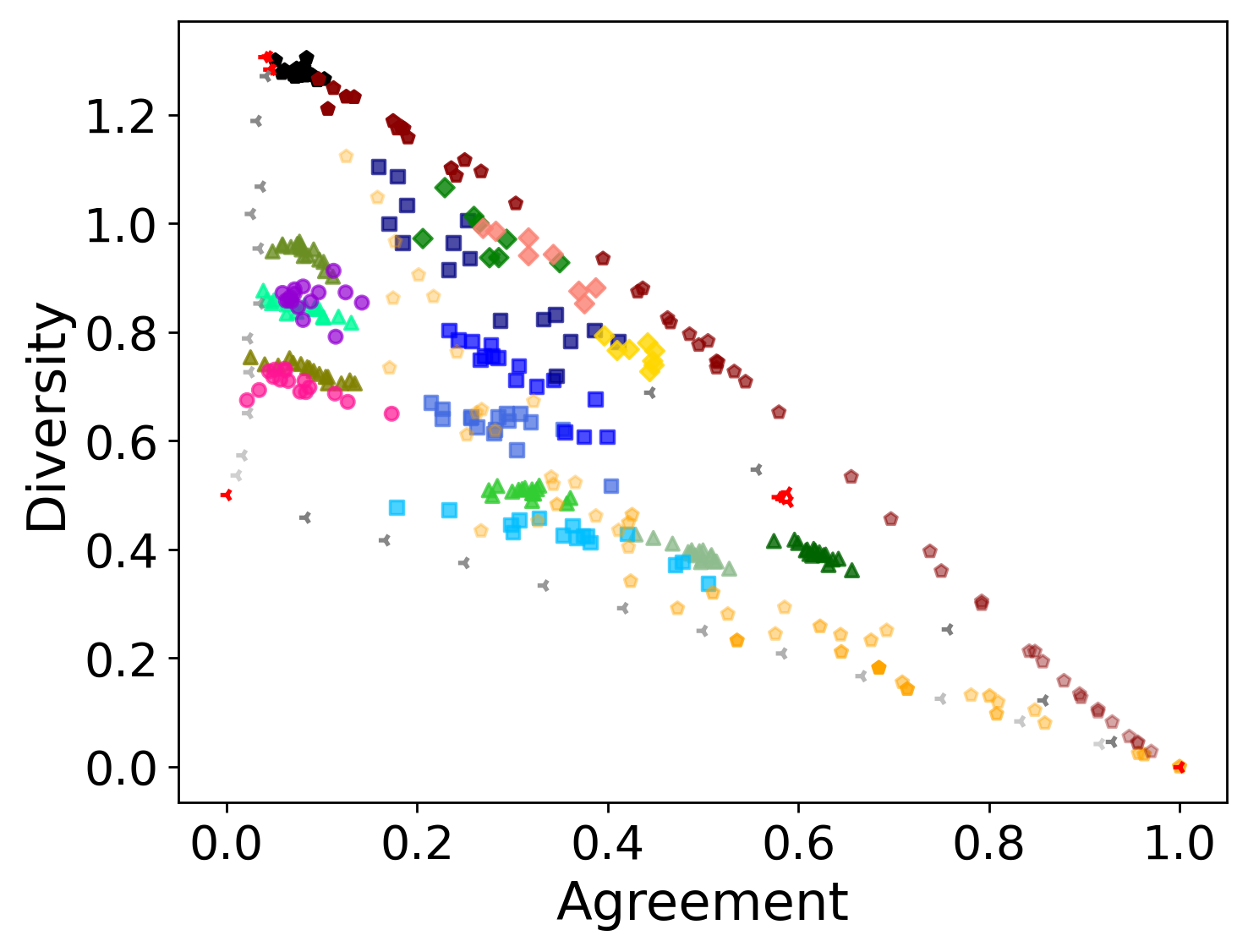}
     \end{subfigure}
     \begin{subfigure}[t]{0.222\textwidth}
         \centering
         \includegraphics[width=\textwidth]{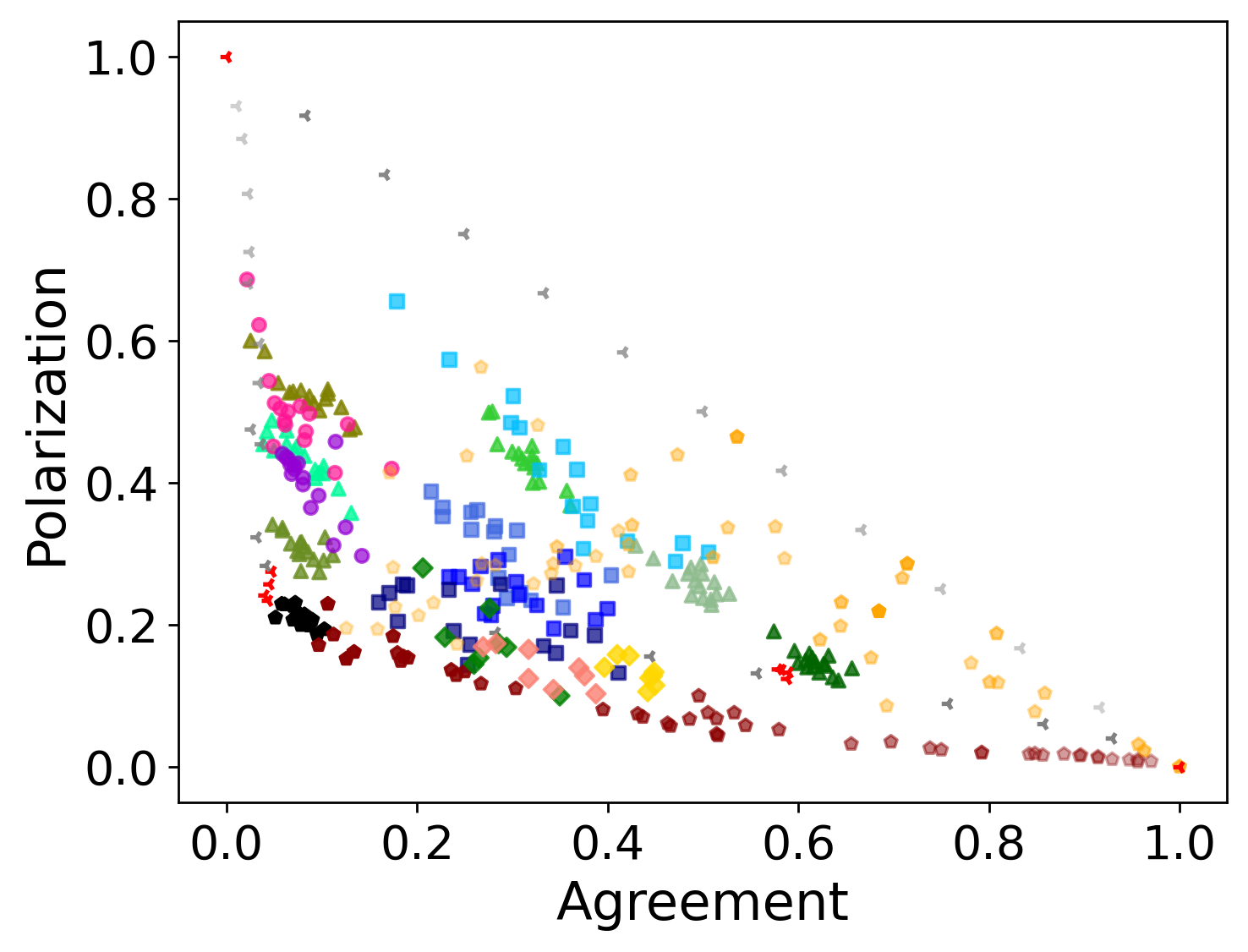}
     \end{subfigure}
     \begin{subfigure}[t]{0.222\textwidth}
         \centering
         \includegraphics[width=\textwidth]{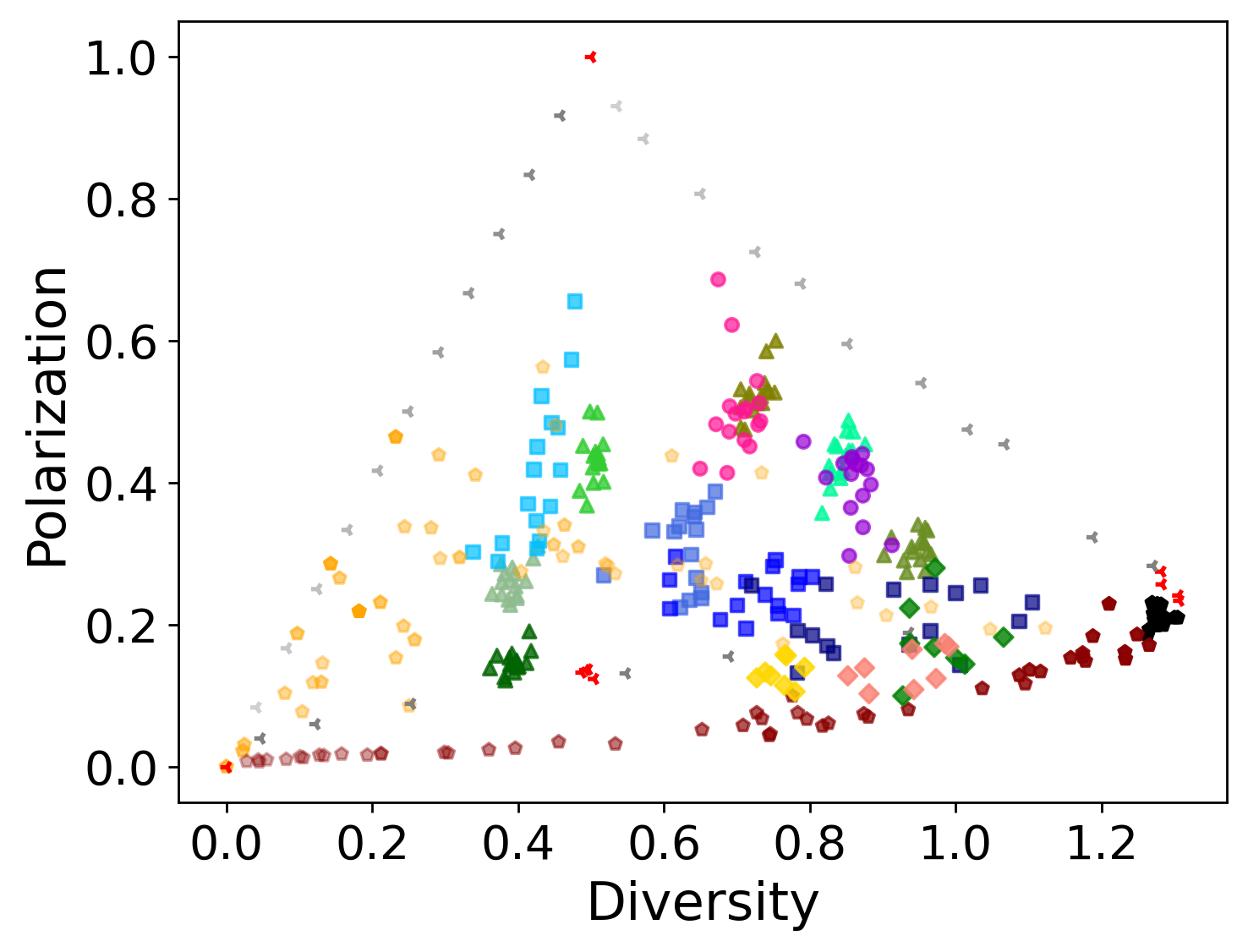}
     \end{subfigure}
\end{figure*}
\begin{figure*}[t]\ContinuedFloat
     \centering
     \begin{subfigure}[t]{0.177\textwidth}
         \centering
         \includegraphics[width=\textwidth]{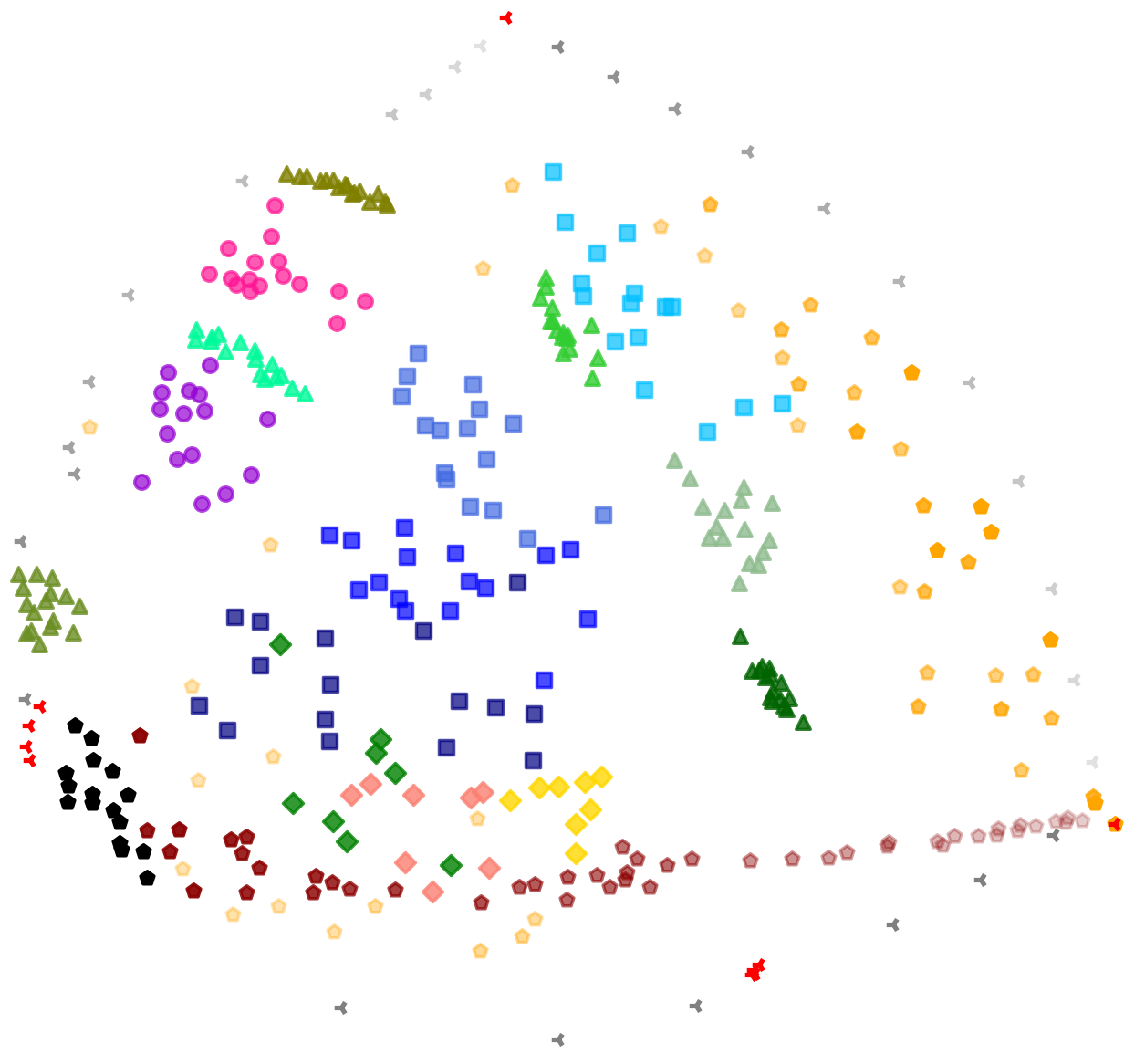}
     \end{subfigure}
     \begin{subfigure}[t]{0.222\textwidth}
         \centering
         \includegraphics[width=\textwidth]{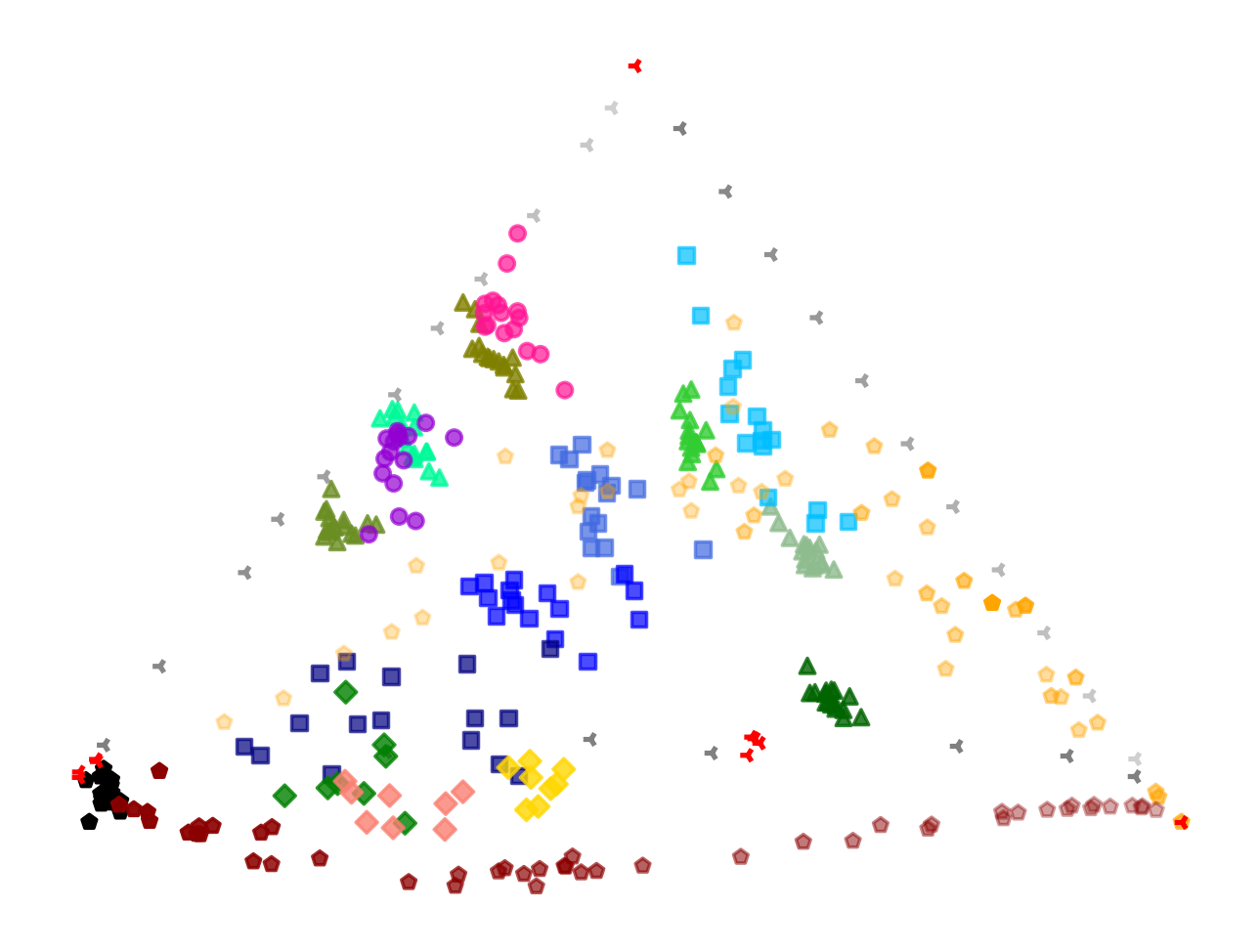}
     \end{subfigure}
     \begin{subfigure}[t]{0.222\textwidth}
         \centering
         \includegraphics[width=\textwidth]{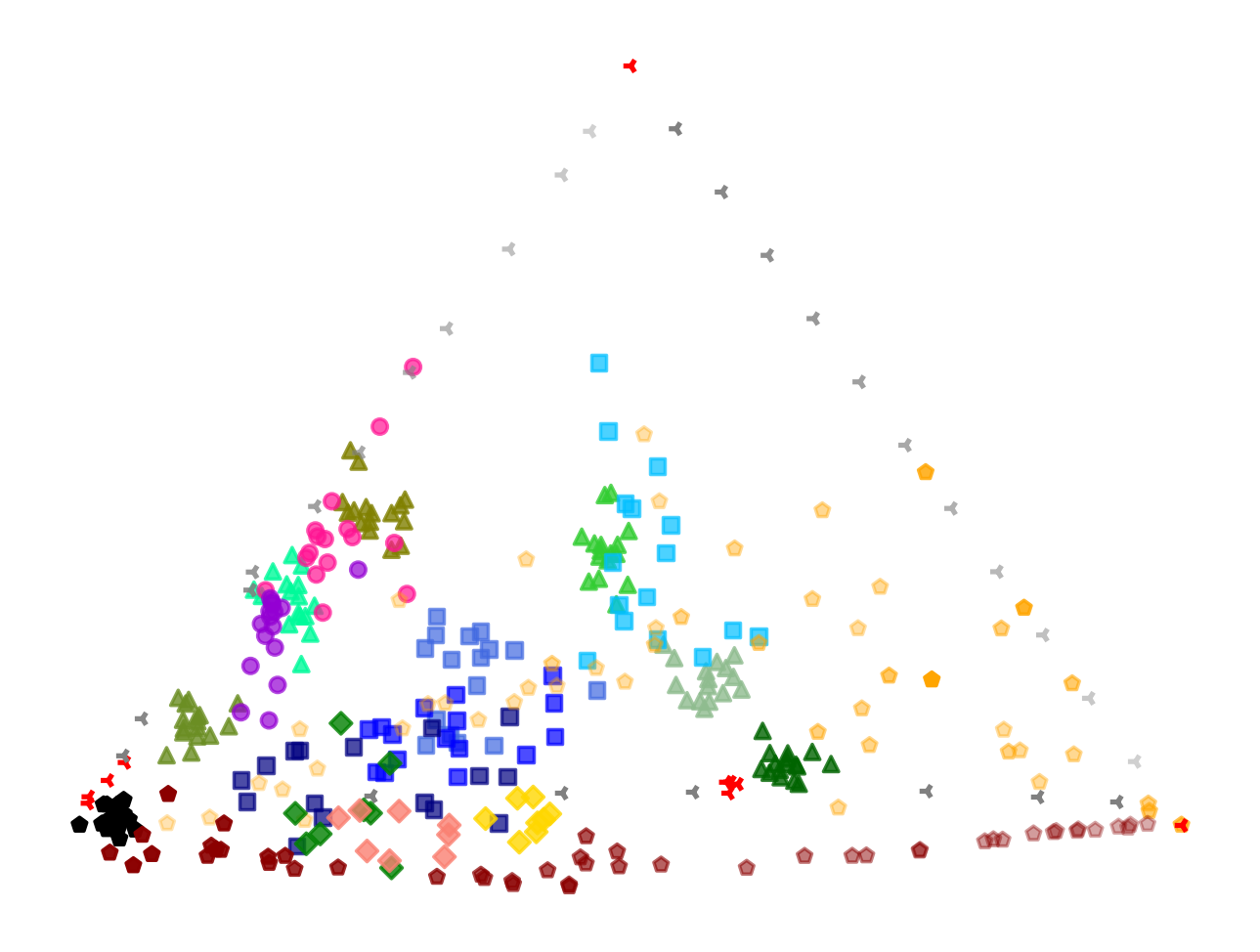}
     \end{subfigure}
     \begin{subfigure}[t]{0.222\textwidth}
         \centering
         \includegraphics[width=\textwidth]{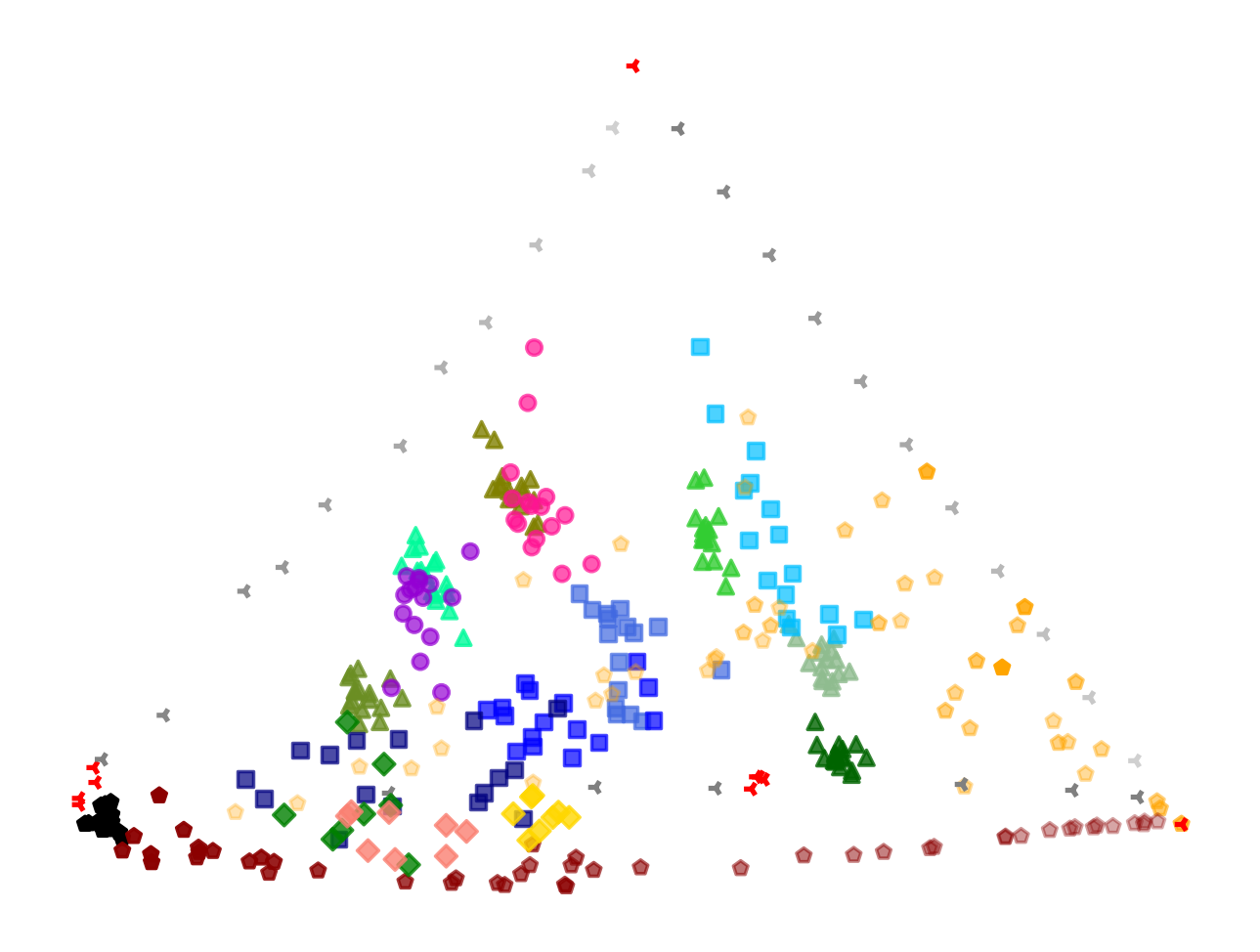}
     \end{subfigure}
\end{figure*}
\begin{figure*}[t]\ContinuedFloat
     \centering
     \begin{subfigure}[t]{0.177\textwidth}
         \centering
         \includegraphics[width=\textwidth]{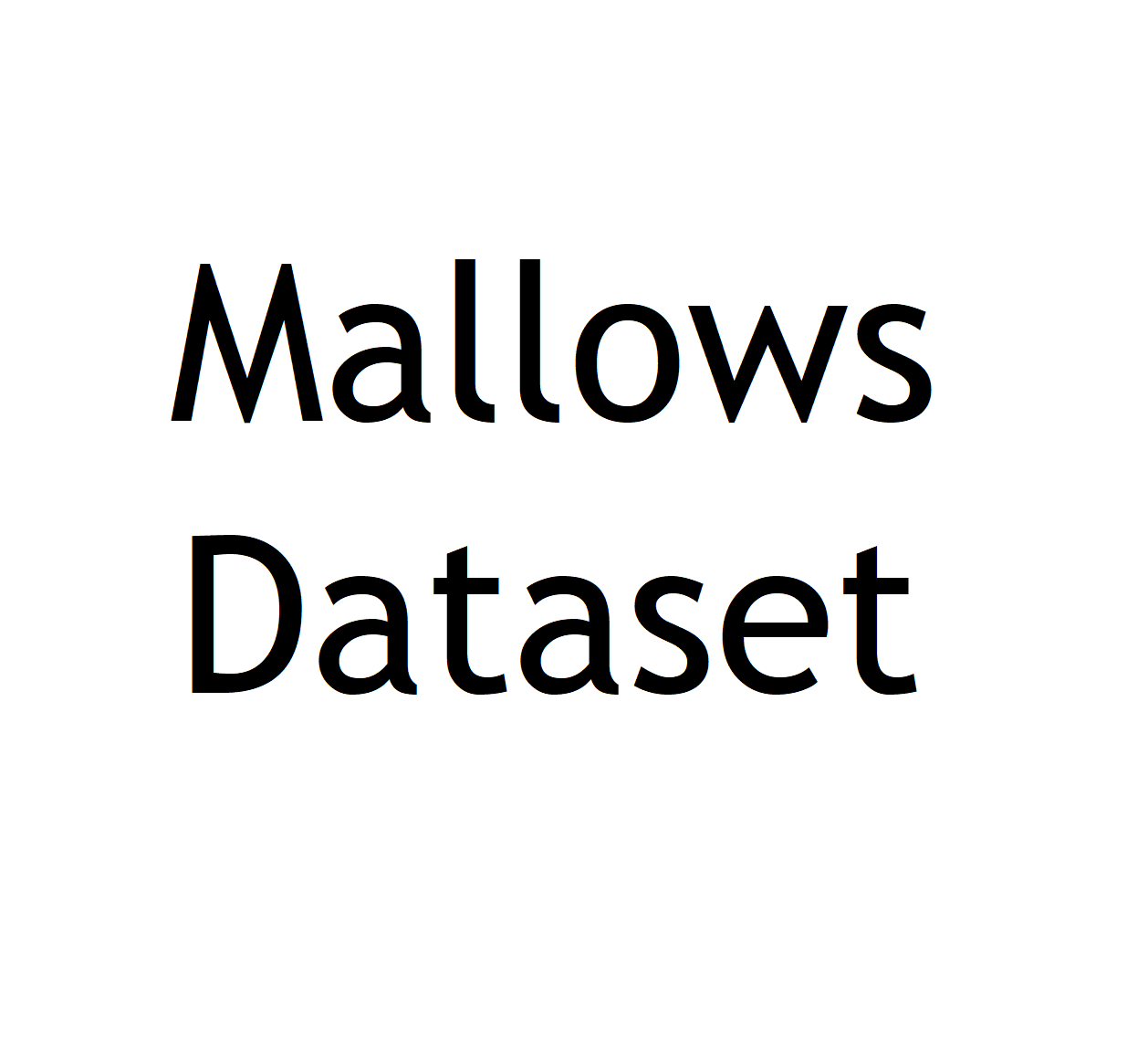}
     \end{subfigure}
     \begin{subfigure}[t]{0.222\textwidth}
         \centering
         \includegraphics[width=\textwidth]{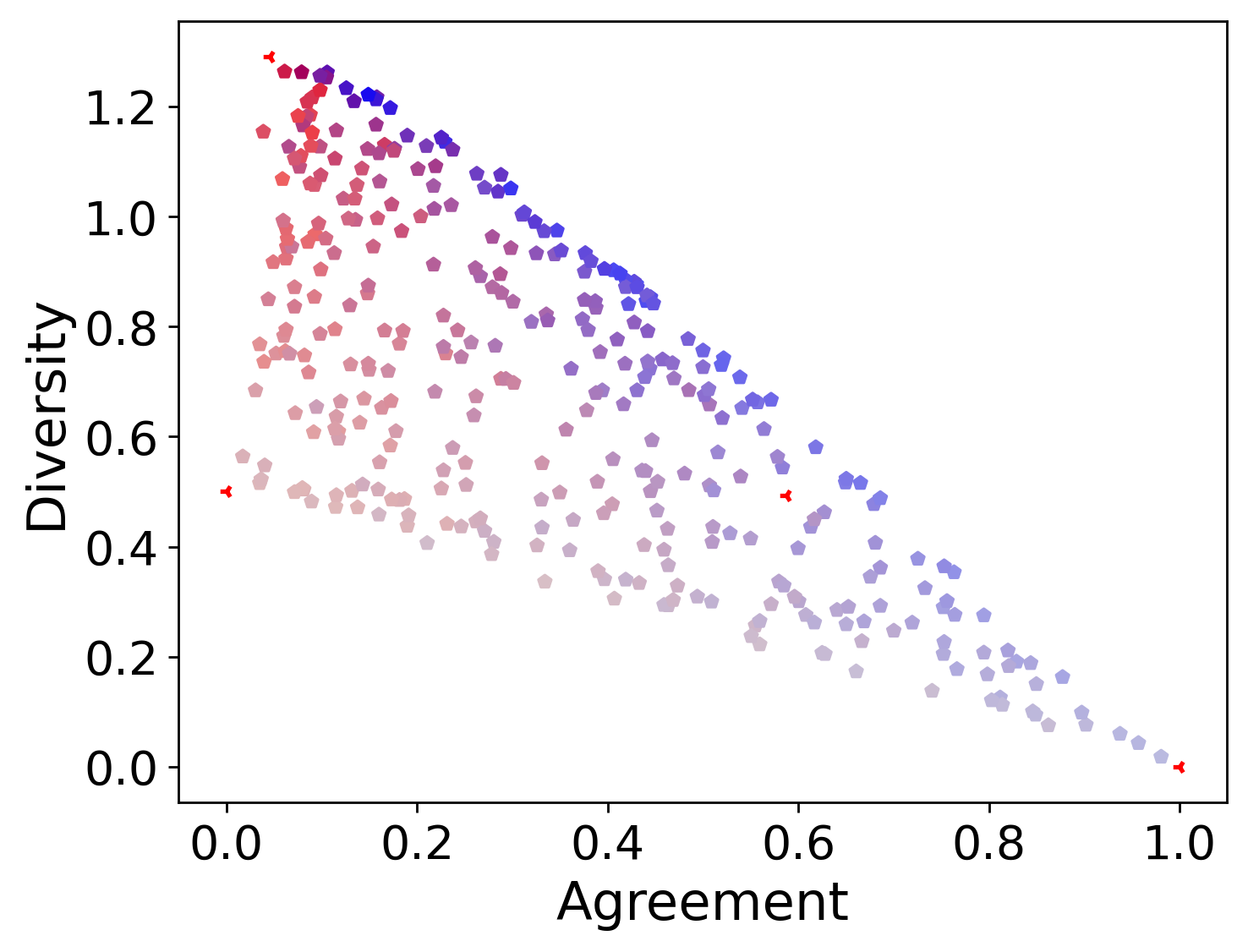}
     \end{subfigure}
     \begin{subfigure}[t]{0.222\textwidth}
         \centering
         \includegraphics[width=\textwidth]{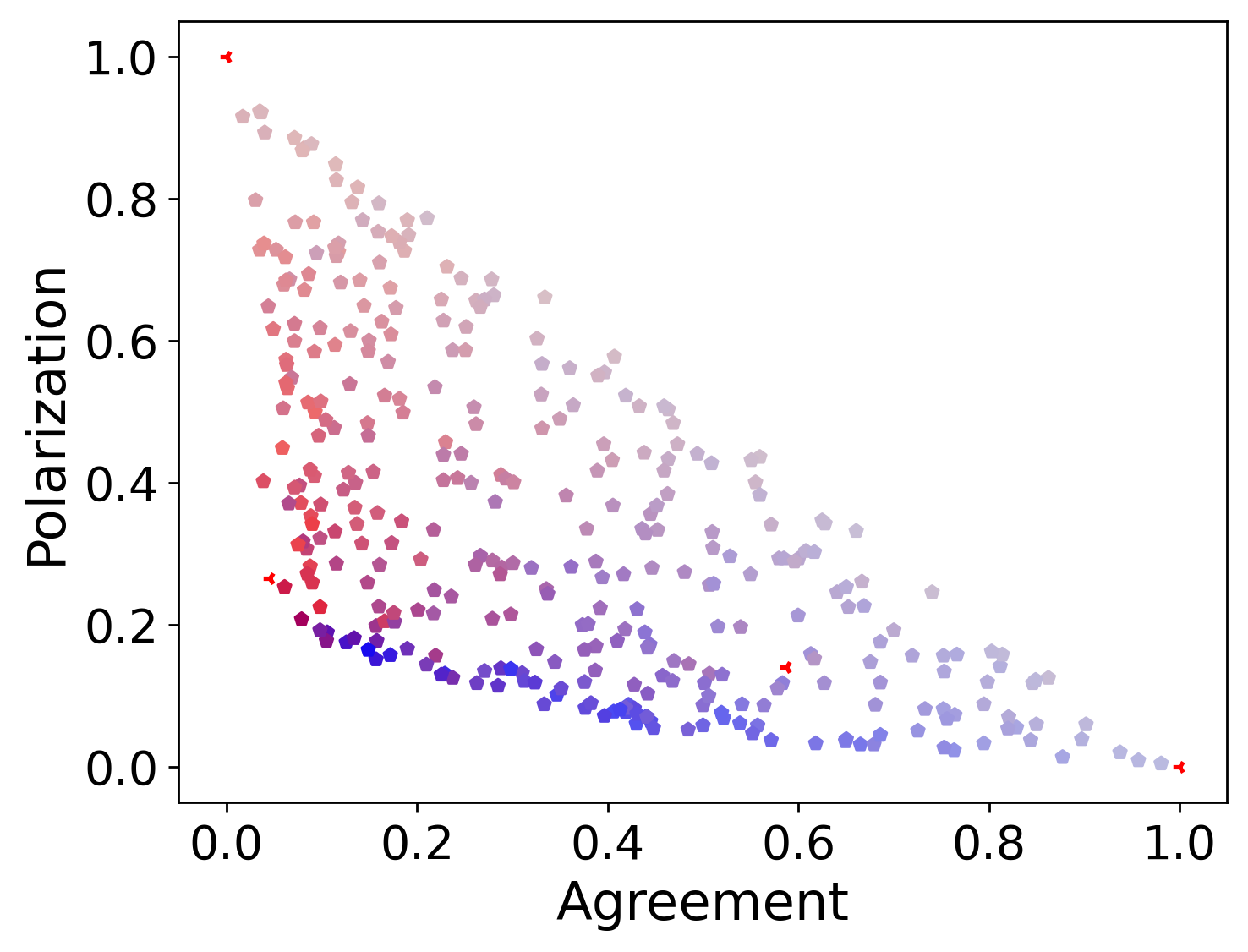}
     \end{subfigure}
     \begin{subfigure}[t]{0.222\textwidth}
         \centering
         \includegraphics[width=\textwidth]{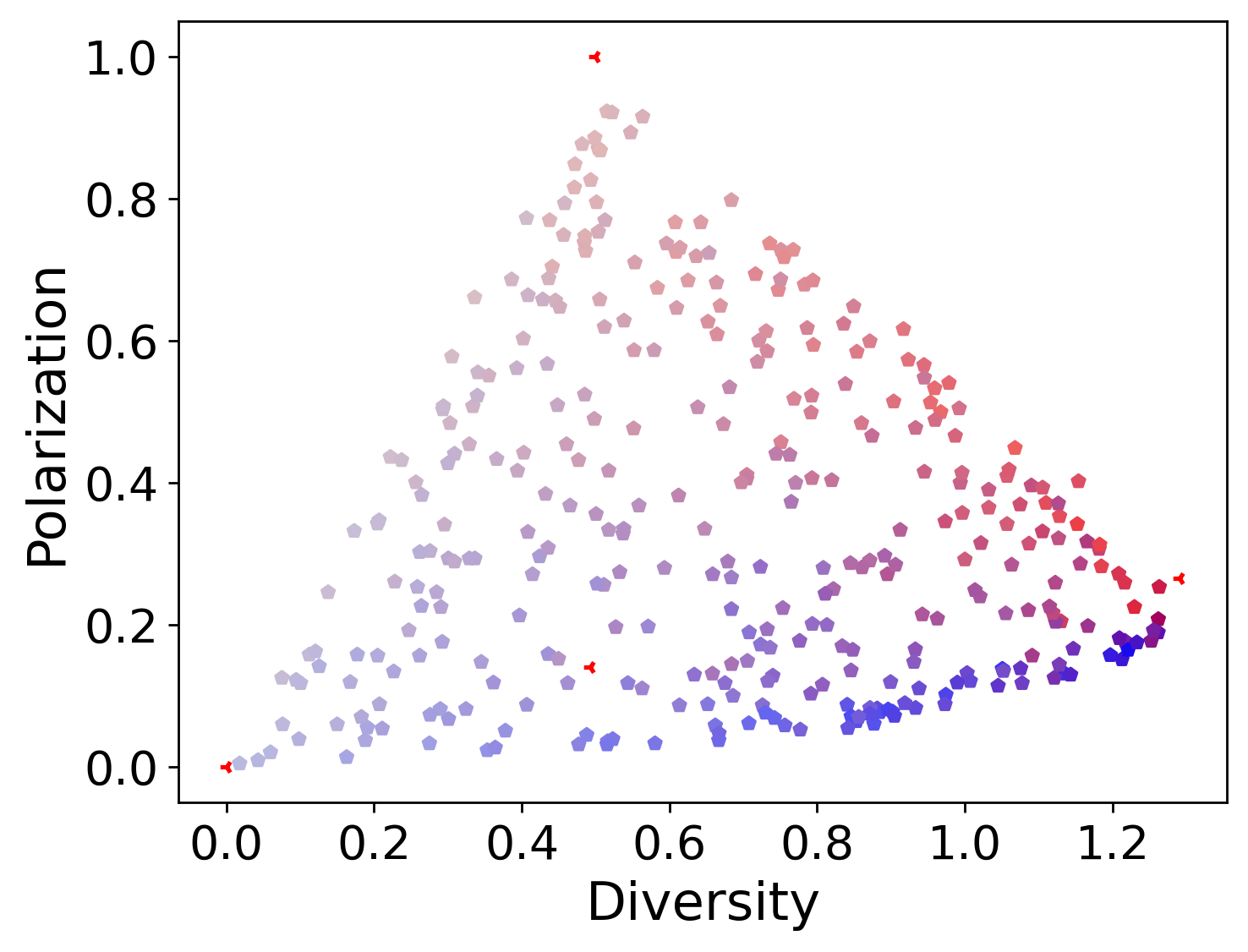}
     \end{subfigure}
\end{figure*}
\begin{figure*}[t]\ContinuedFloat
     \centering
     \begin{subfigure}[t]{0.177\textwidth}
         \centering
         \includegraphics[width=\textwidth]{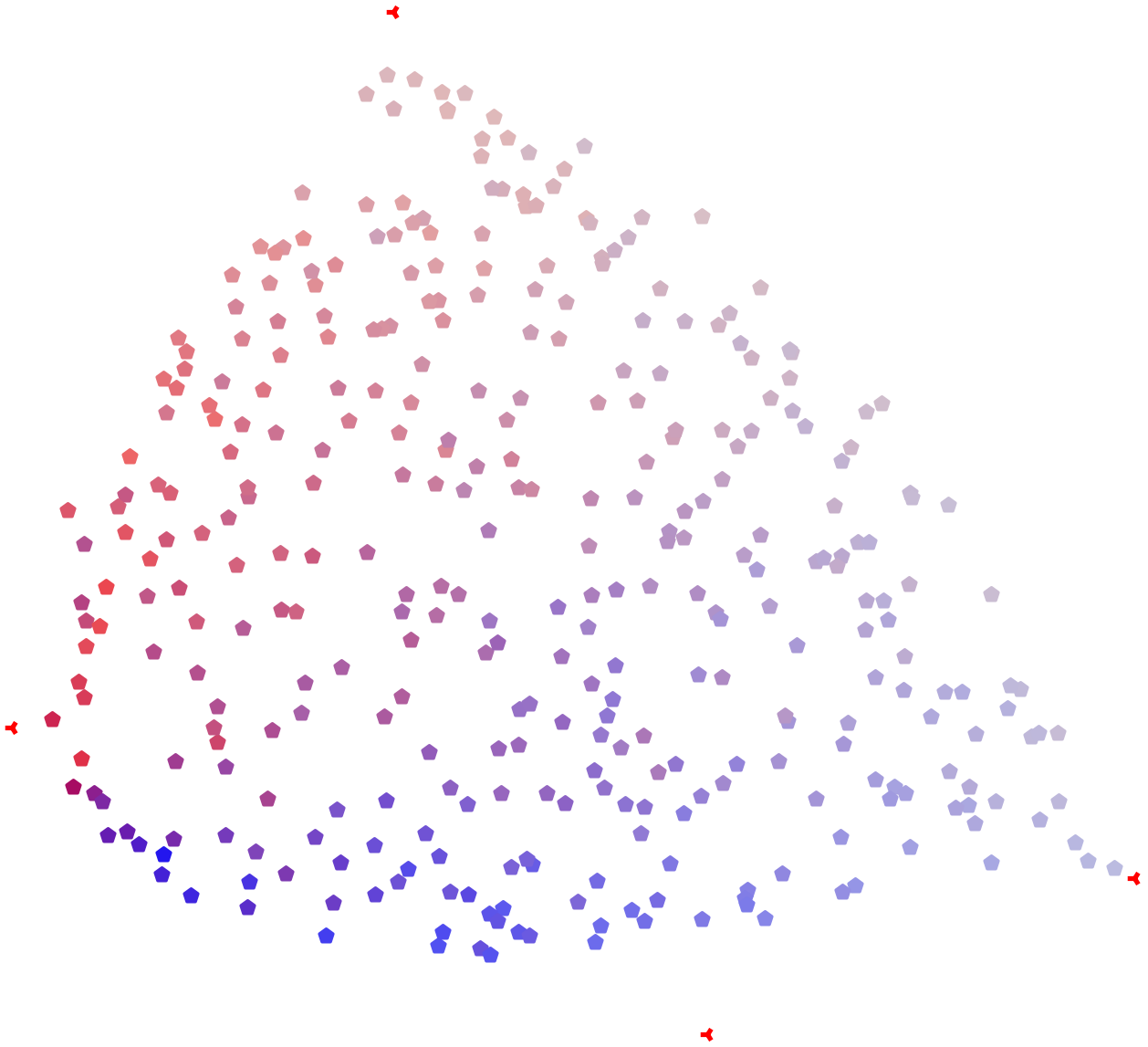}
     \end{subfigure}
     \begin{subfigure}[t]{0.222\textwidth}
         \centering
         \includegraphics[width=\textwidth]{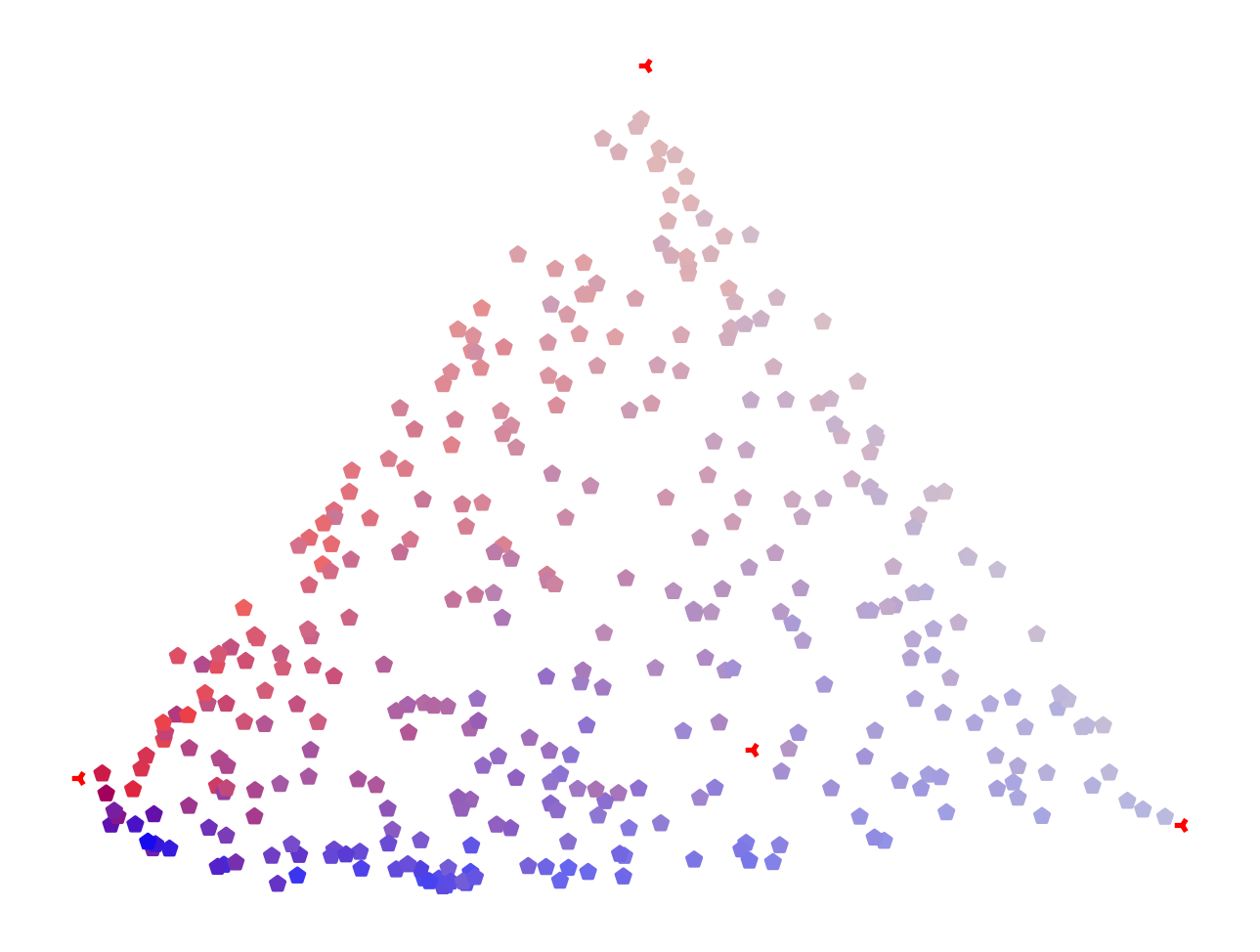}
     \end{subfigure}
     \begin{subfigure}[t]{0.222\textwidth}
         \centering
         \includegraphics[width=\textwidth]{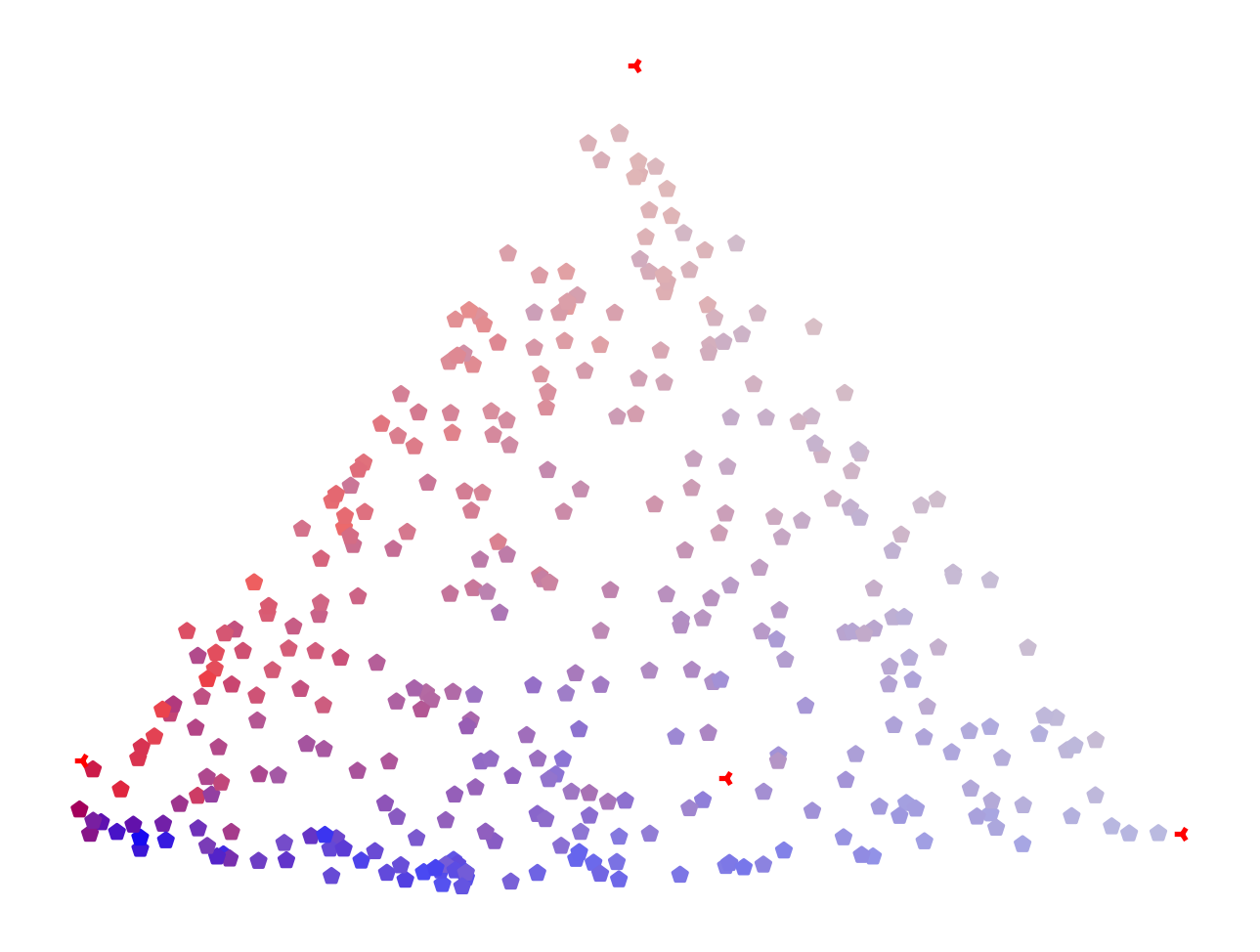}
     \end{subfigure}
     \begin{subfigure}[t]{0.222\textwidth}
         \centering
         \includegraphics[width=\textwidth]{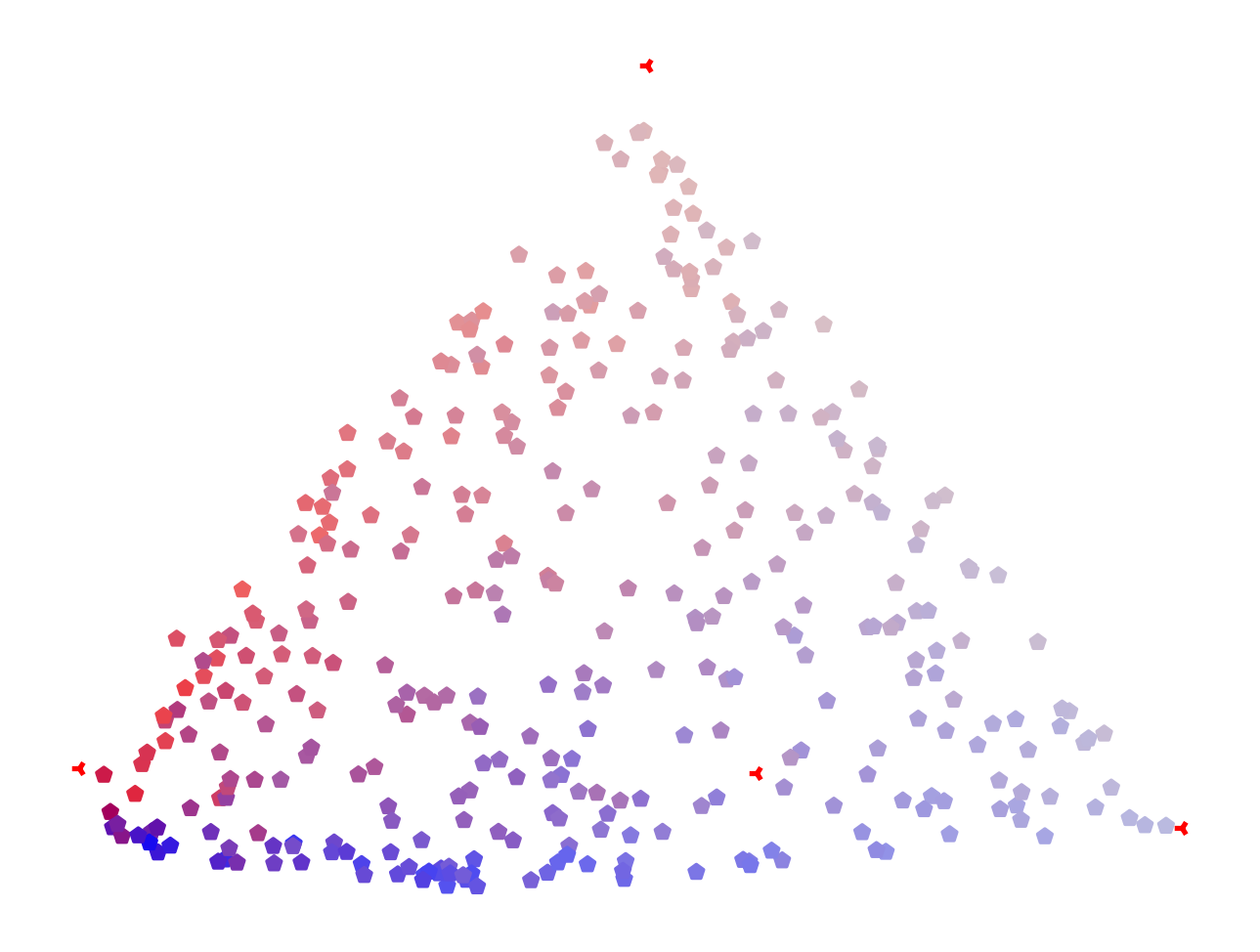}
     \end{subfigure}
    \caption{For each dataset, the first row presents the plots where the position of each dot corresponds to the value of our indices for elections in the dataset.
    In the second row, under each plot, we present its \emph{affine transformation} obtained in a following way: First, we rotate the map in such a way that \ID and \AN form a horizontal line (with \ID on the left hand side). If $\appUN$ is below this line, we take a symmetric reflection with respect to \ID--\AN line.
    Next, we take the dot furthest from \ID--\AN line, $x$, and scale the height of the image, so that the distance from $x$ to \ID--\AN line is approximately 0.87 times the distance from \ID to \AN (the height of the equilateral triangle).
    Then, we make a \emph{shear mapping} to make sure that $x$ is in equal distance to \ID and \AN (i.e., we move $x$ to the right or to the left, so it is in the middle,
    and every other dot we move in the same direction, but less, proportionally to its distance to \ID--\AN line).
    Finally, we rotate the picture by 120 degrees so that $x$ and \ID form a horizontal line.
    For comparison, in the first column we present corresponding maps of elections from 
    Figs.~\ref{fig:swap-map:standard},
    \ref{fig:swap-map:extended}, and~\ref{fig:swap-map:mallows}.
    }
    \label{fig:plots}
\end{figure*}
\clearpage

\begin{figure*}[t]
     \centering
     \begin{subfigure}[t]{0.03\textwidth}
         \centering
         \includegraphics[width=\textwidth]{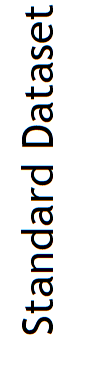}
     \end{subfigure}
     \hfill
     \begin{subfigure}[t]{0.18\textwidth}
         \centering
         \includegraphics[width=\textwidth]{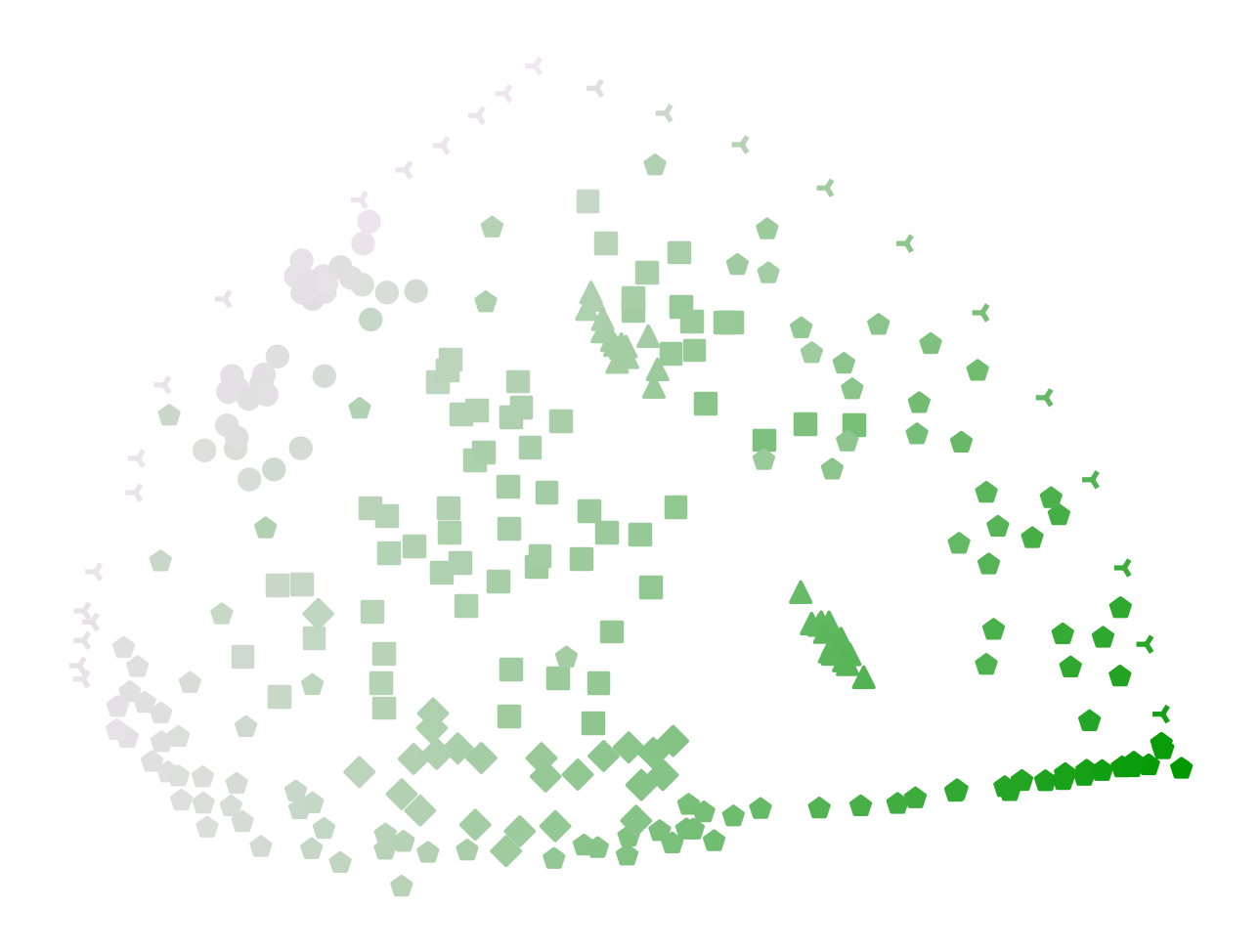}
     \end{subfigure}
     \hfill
     \begin{subfigure}[t]{0.18\textwidth}
         \centering
         \includegraphics[width=\textwidth]{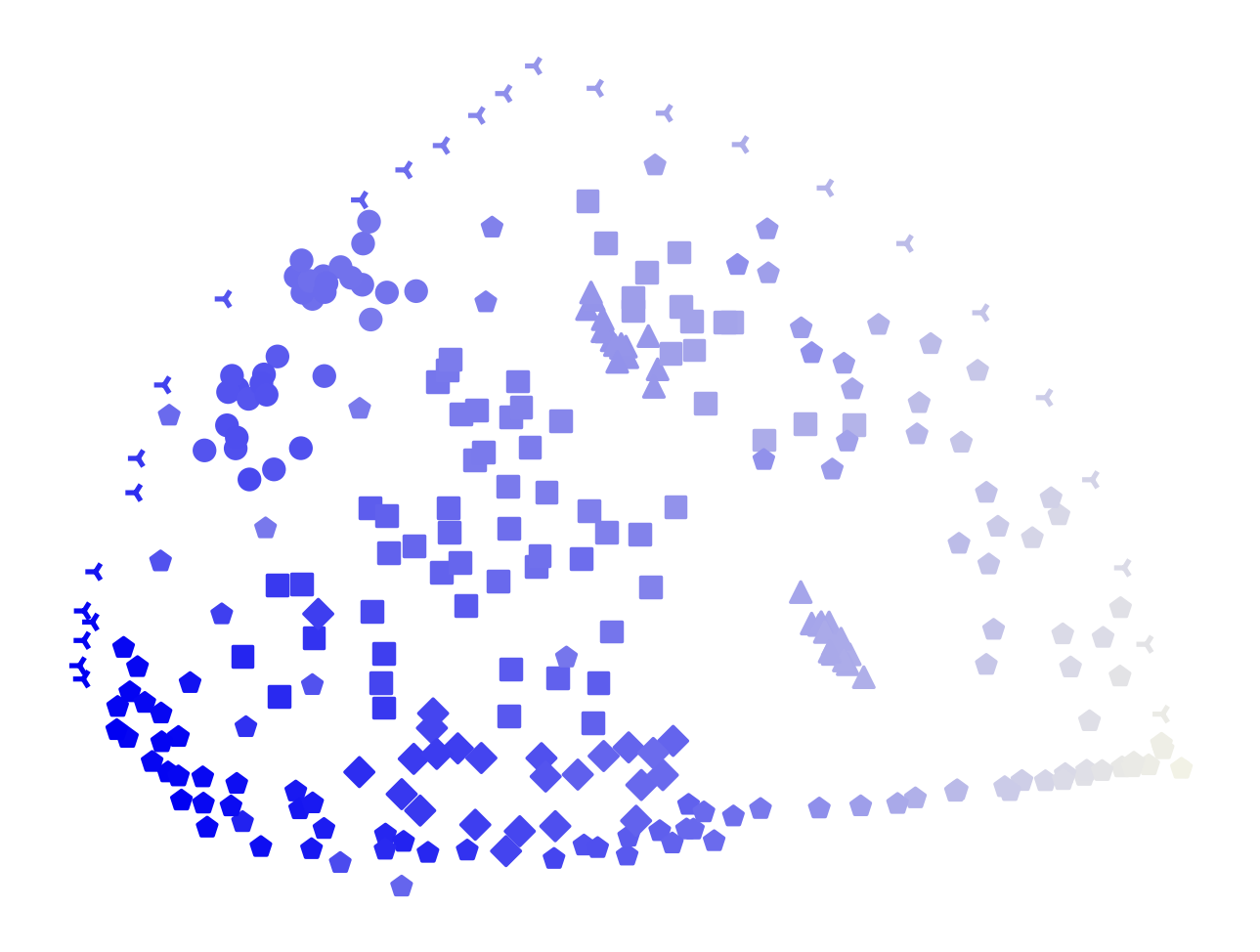}
     \end{subfigure}
     \hfill
     \begin{subfigure}[t]{0.18\textwidth}
         \centering
         \includegraphics[width=\textwidth]{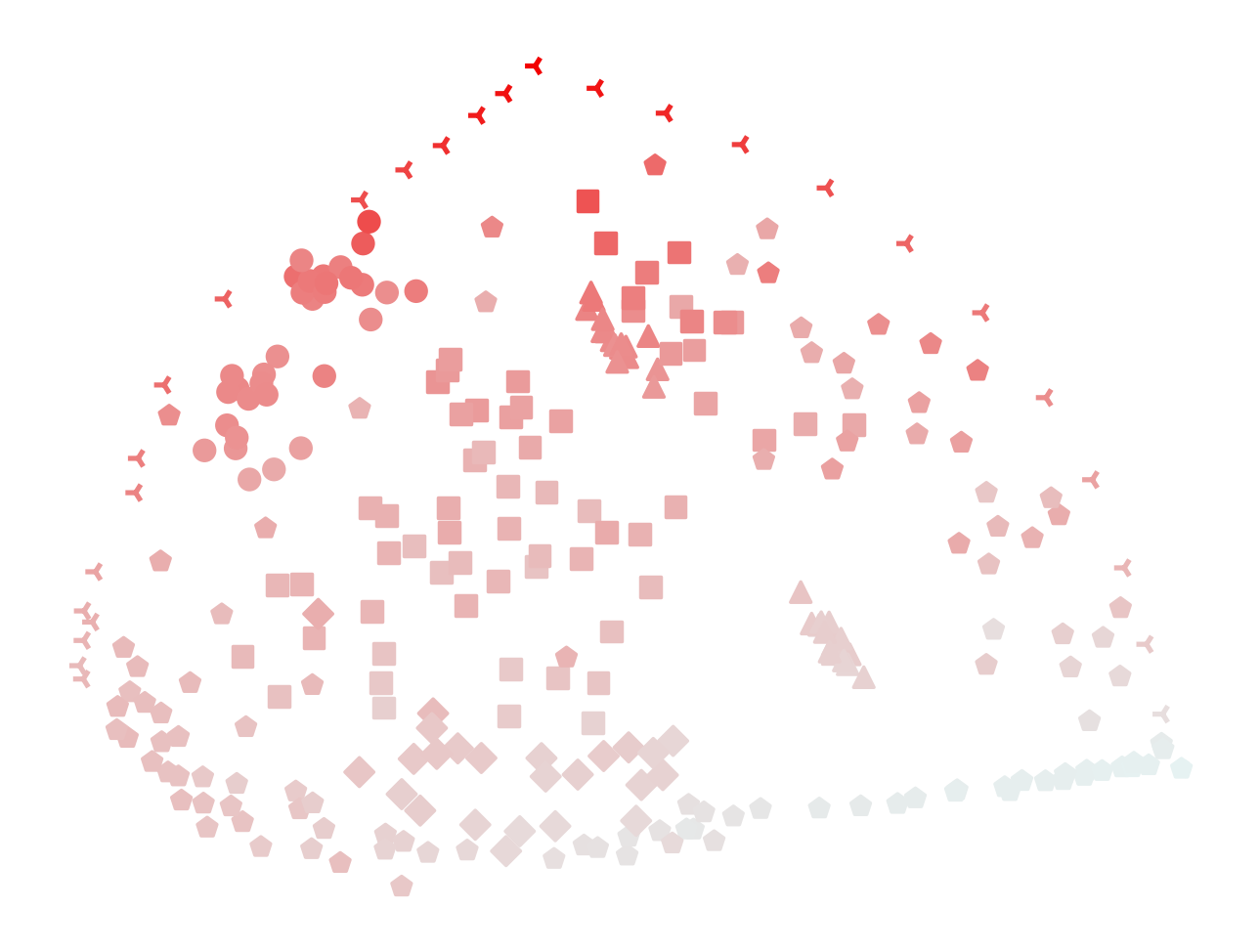}
    \end{subfigure}
    \hfill
    \begin{subfigure}[t]{0.18\textwidth}
     \centering
         \includegraphics[width=\textwidth]{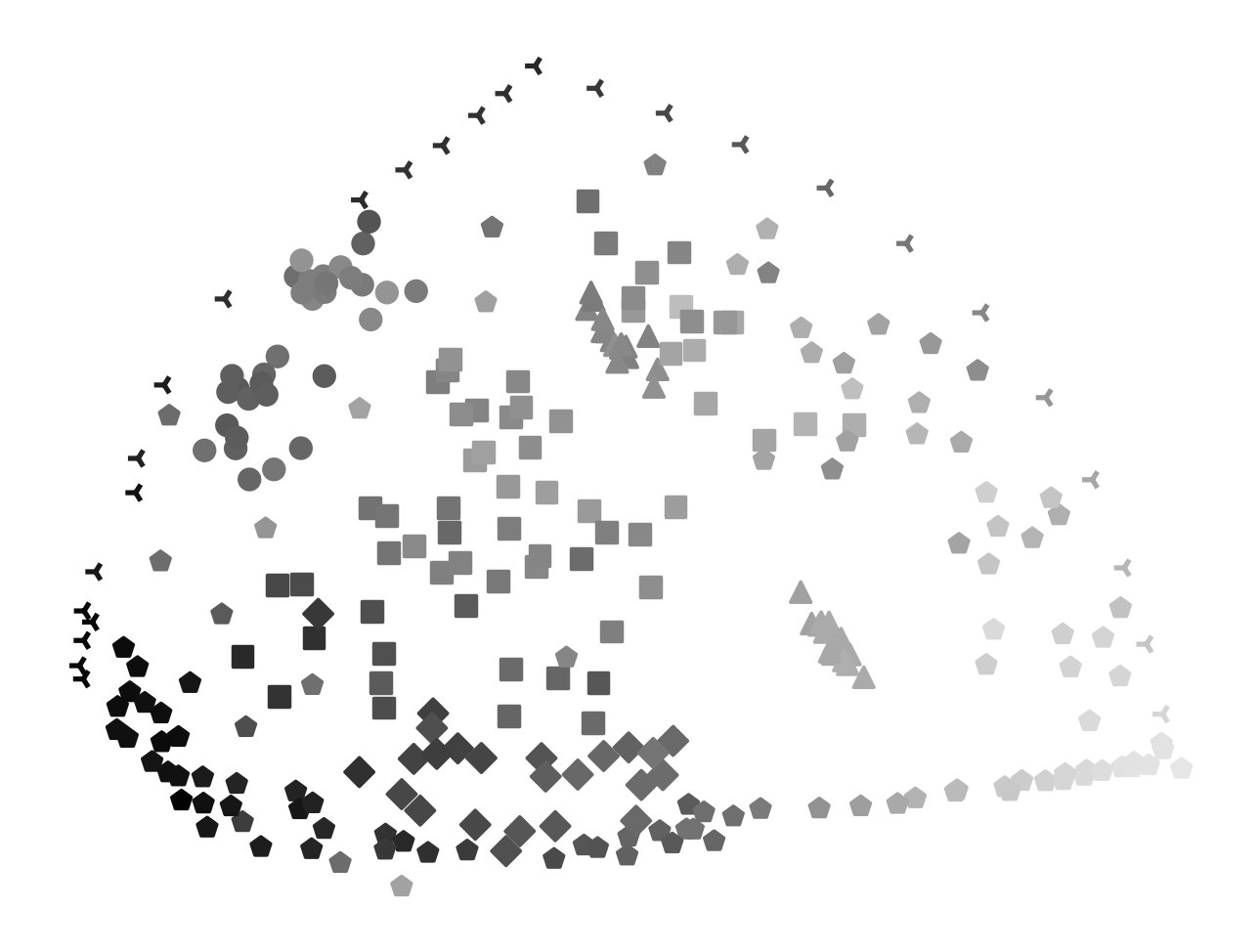}
    \end{subfigure}
    \hfill
     \begin{subfigure}[t]{0.18\textwidth}
     \centering
         \includegraphics[width=\textwidth]{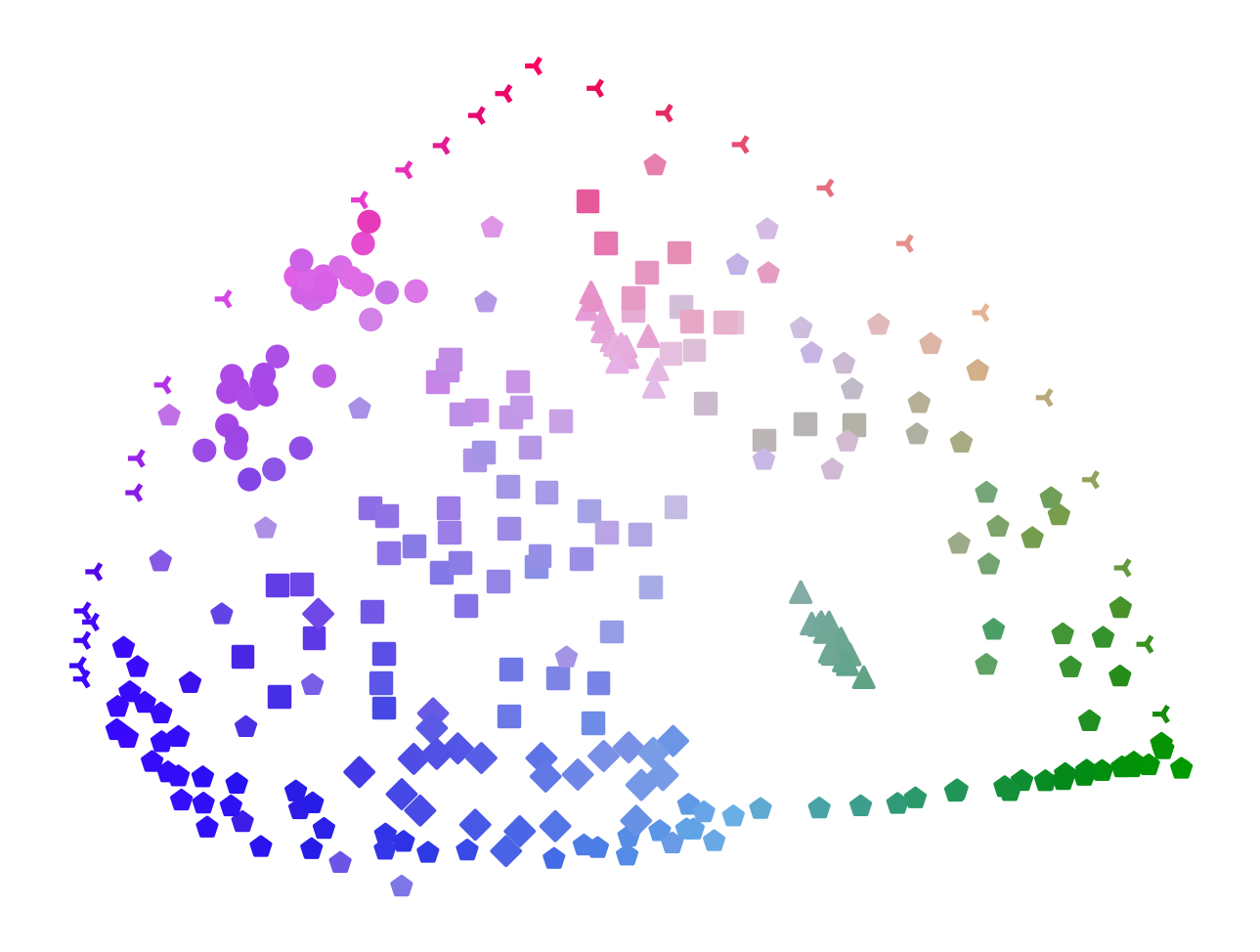}
    \end{subfigure}
\end{figure*}
\begin{figure*}[t]\ContinuedFloat
     \centering
     \begin{subfigure}[t]{0.03\textwidth}
         \centering
         \includegraphics[width=\textwidth]{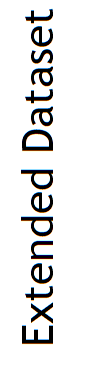}
     \end{subfigure}
     \hfill
     \begin{subfigure}[t]{0.18\textwidth}
         \centering
         \includegraphics[width=\textwidth]{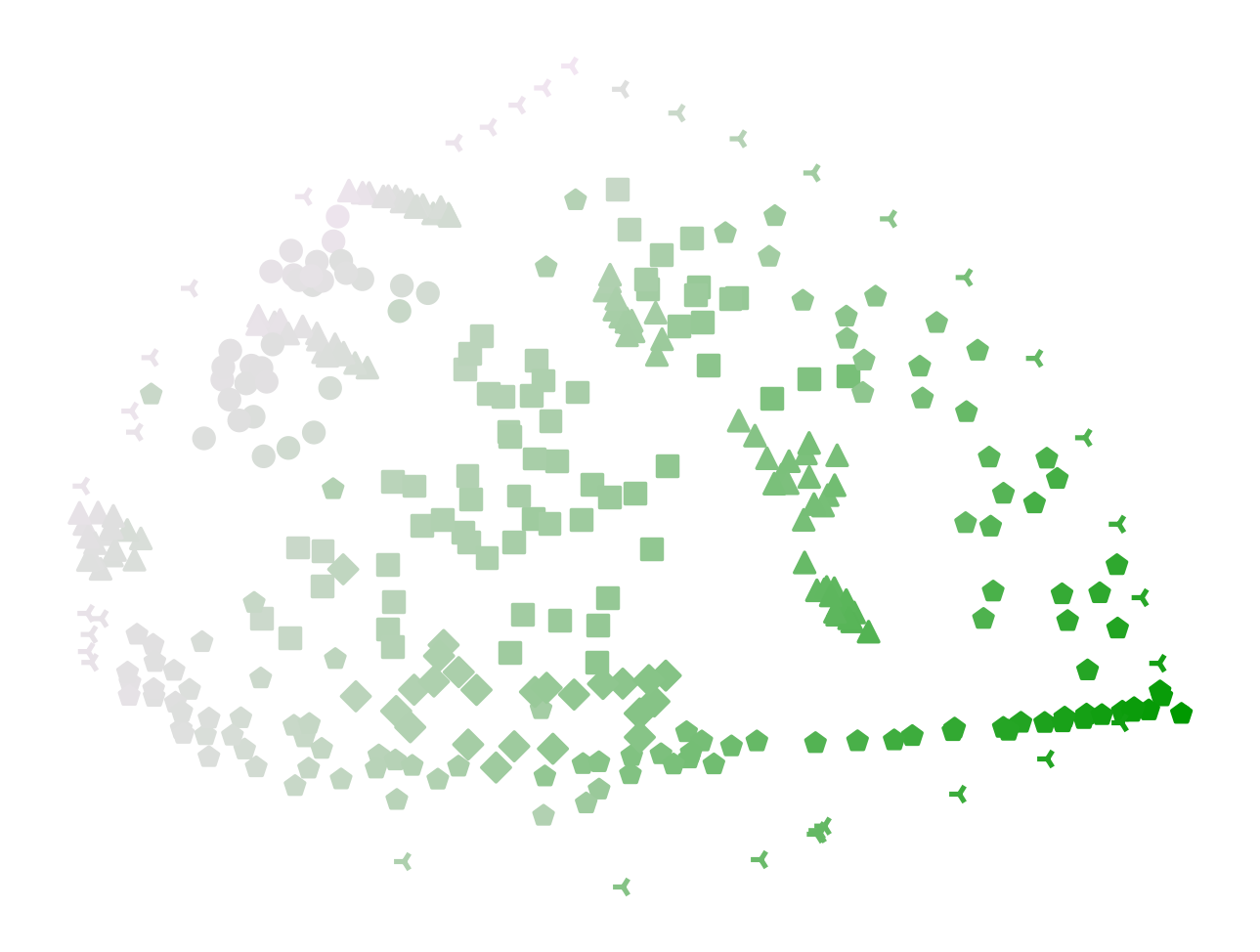}
     \end{subfigure}
     \hfill
     \begin{subfigure}[t]{0.18\textwidth}
         \centering
         \includegraphics[width=\textwidth]{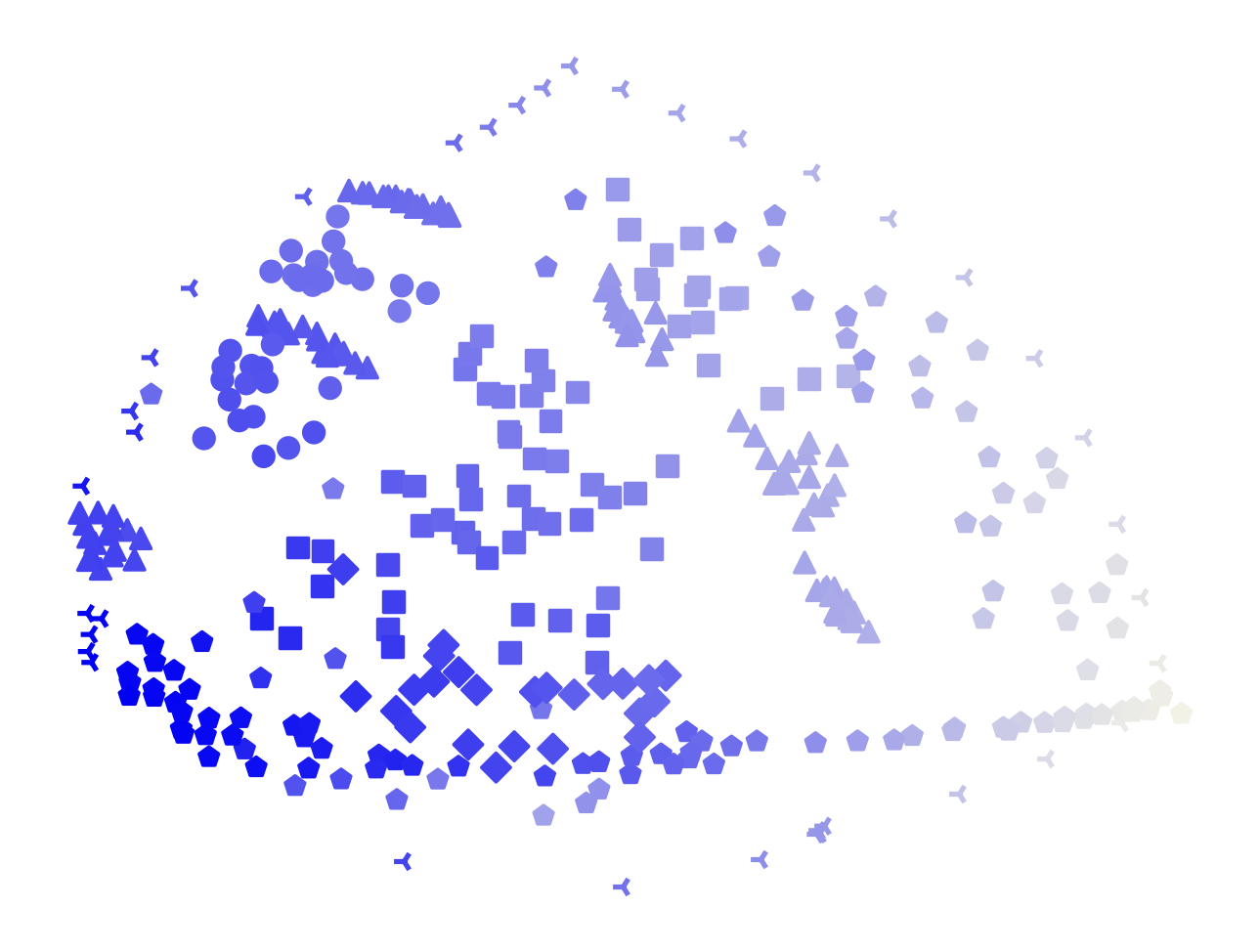}
     \end{subfigure}
     \hfill
     \begin{subfigure}[t]{0.18\textwidth}
         \centering
         \includegraphics[width=\textwidth]{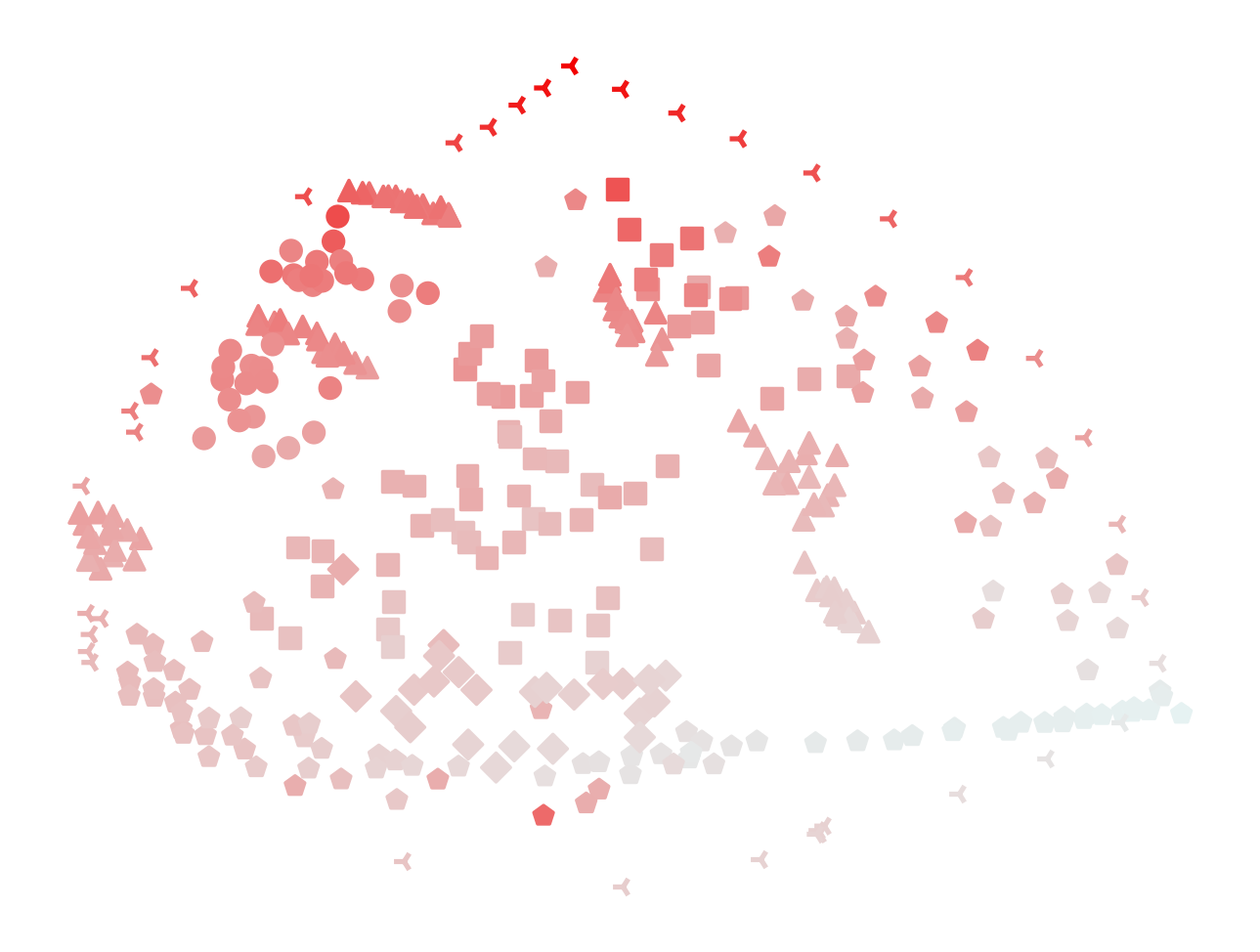}
    \end{subfigure}
    \hfill
    \begin{subfigure}[t]{0.18\textwidth}
     \centering
         \includegraphics[width=\textwidth]{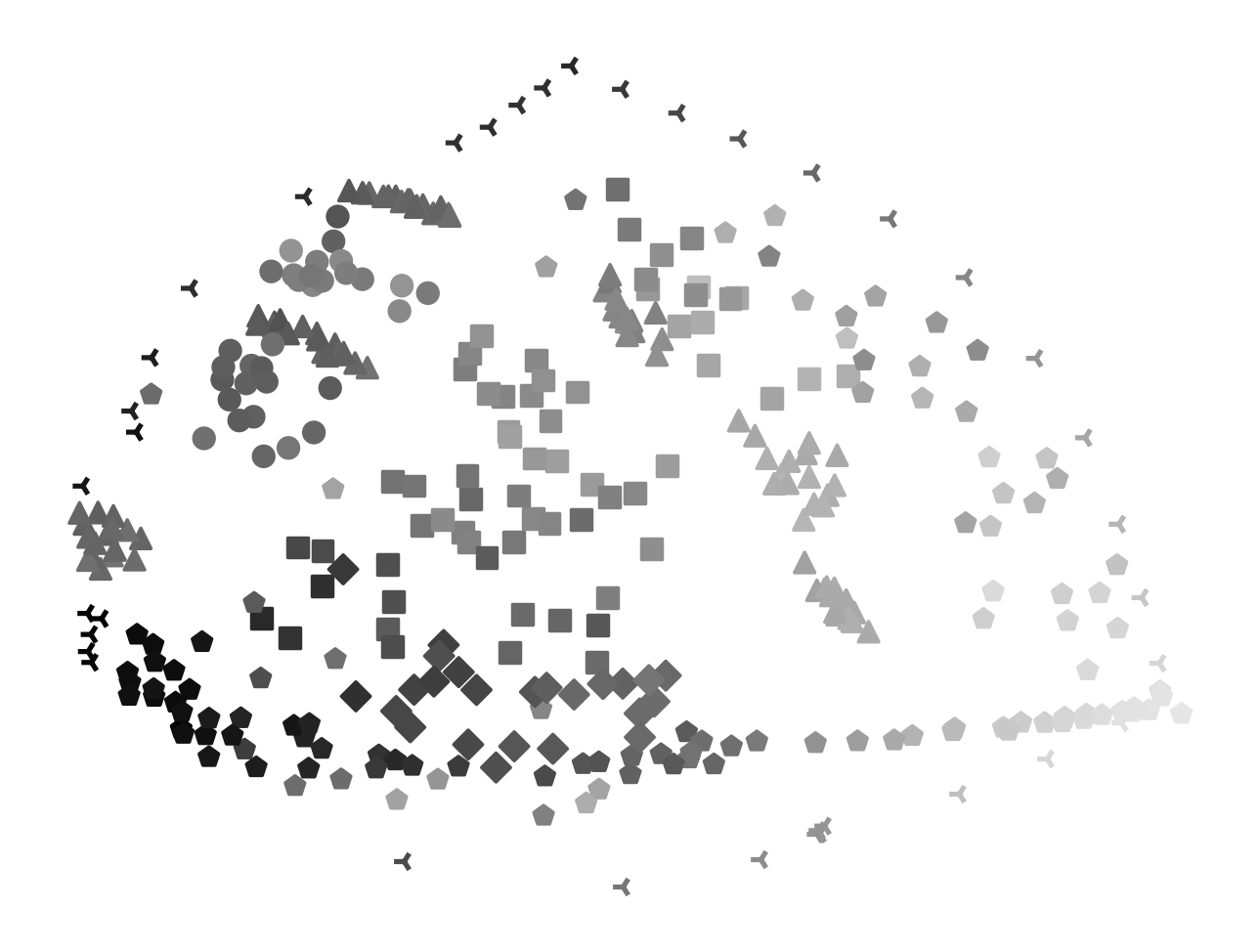}
    \end{subfigure}
    \hfill
     \begin{subfigure}[t]{0.18\textwidth}
     \centering
         \includegraphics[width=\textwidth]{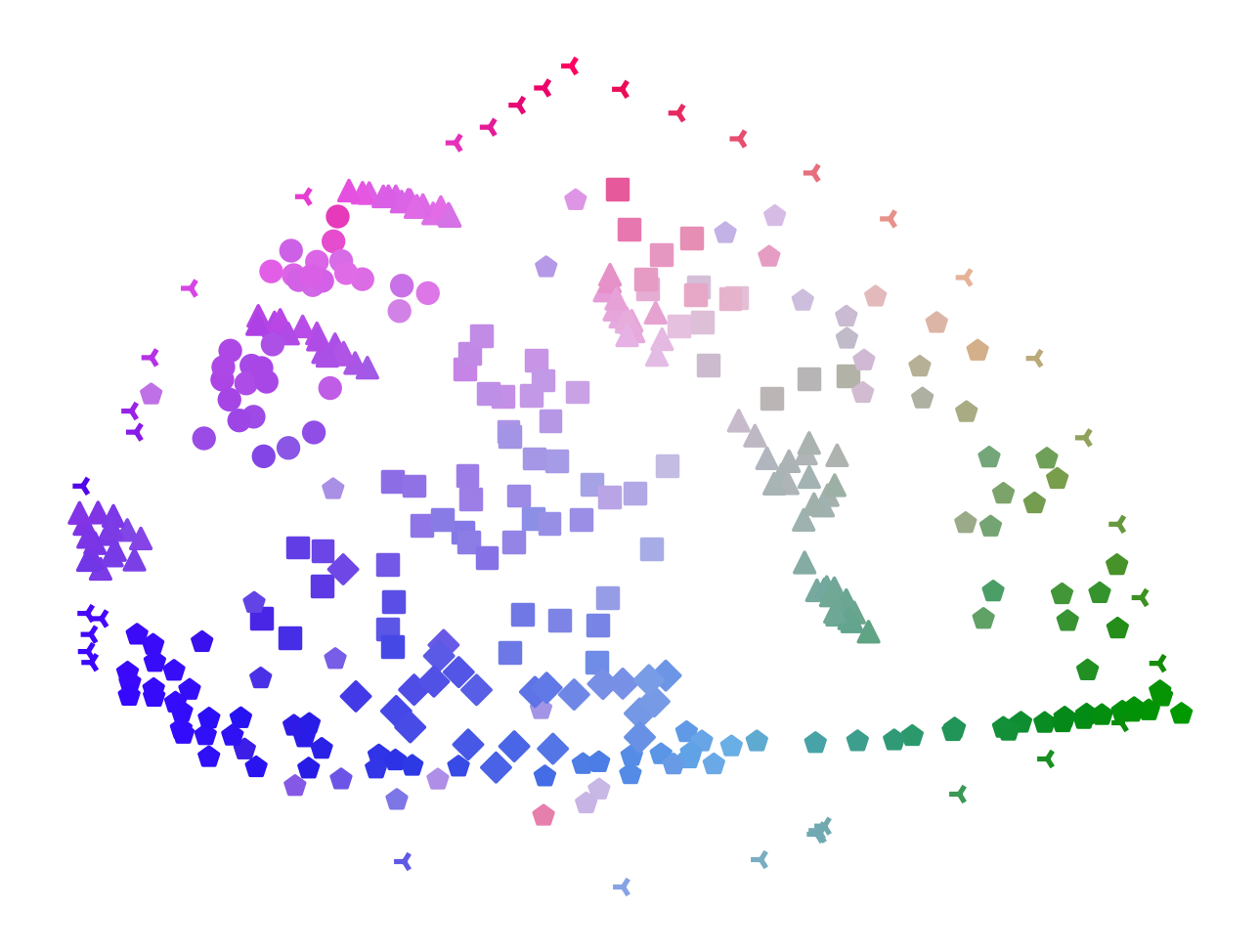}
    \end{subfigure}
\end{figure*}
\begin{figure*}[t]\ContinuedFloat
     \centering
     \begin{subfigure}[t]{0.03\textwidth}
         \centering
         \includegraphics[width=\textwidth]{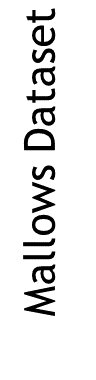}
     \end{subfigure}
     \hfill
     \begin{subfigure}[t]{0.18\textwidth}
         \centering
         \includegraphics[width=\textwidth]{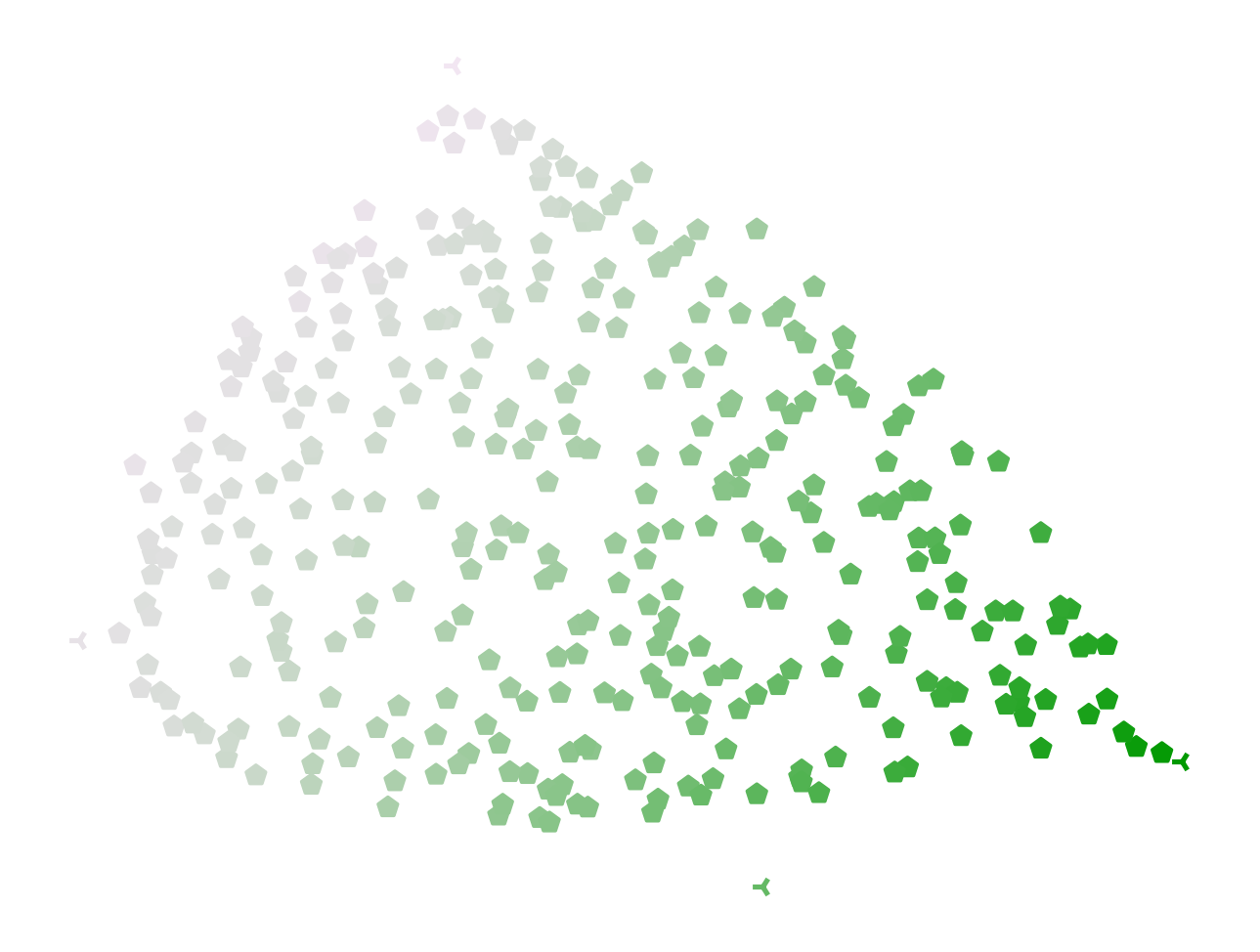}
     \end{subfigure}
     \hfill
     \begin{subfigure}[t]{0.18\textwidth}
         \centering
         \includegraphics[width=\textwidth]{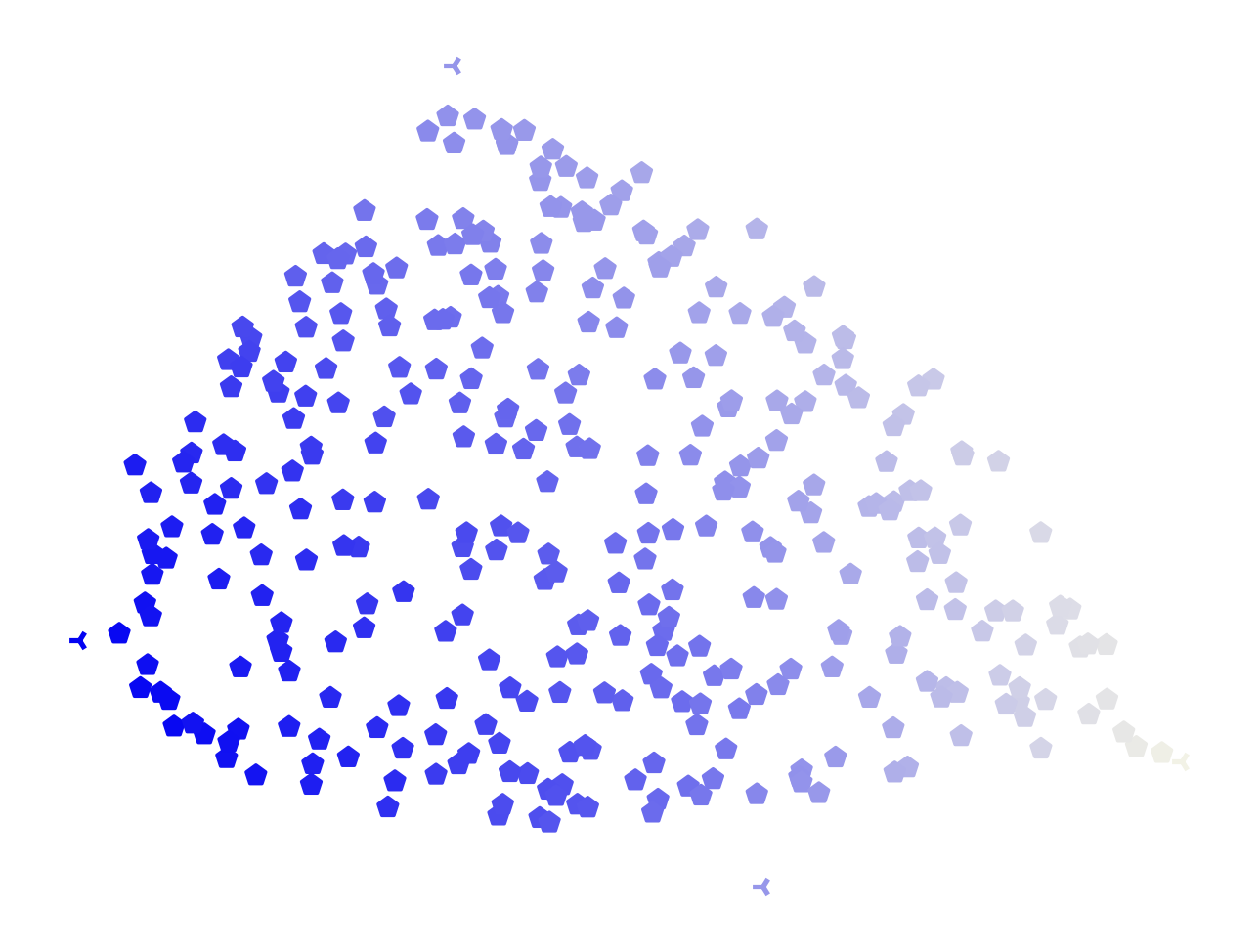}
     \end{subfigure}
     \hfill
     \begin{subfigure}[t]{0.18\textwidth}
         \centering
         \includegraphics[width=\textwidth]{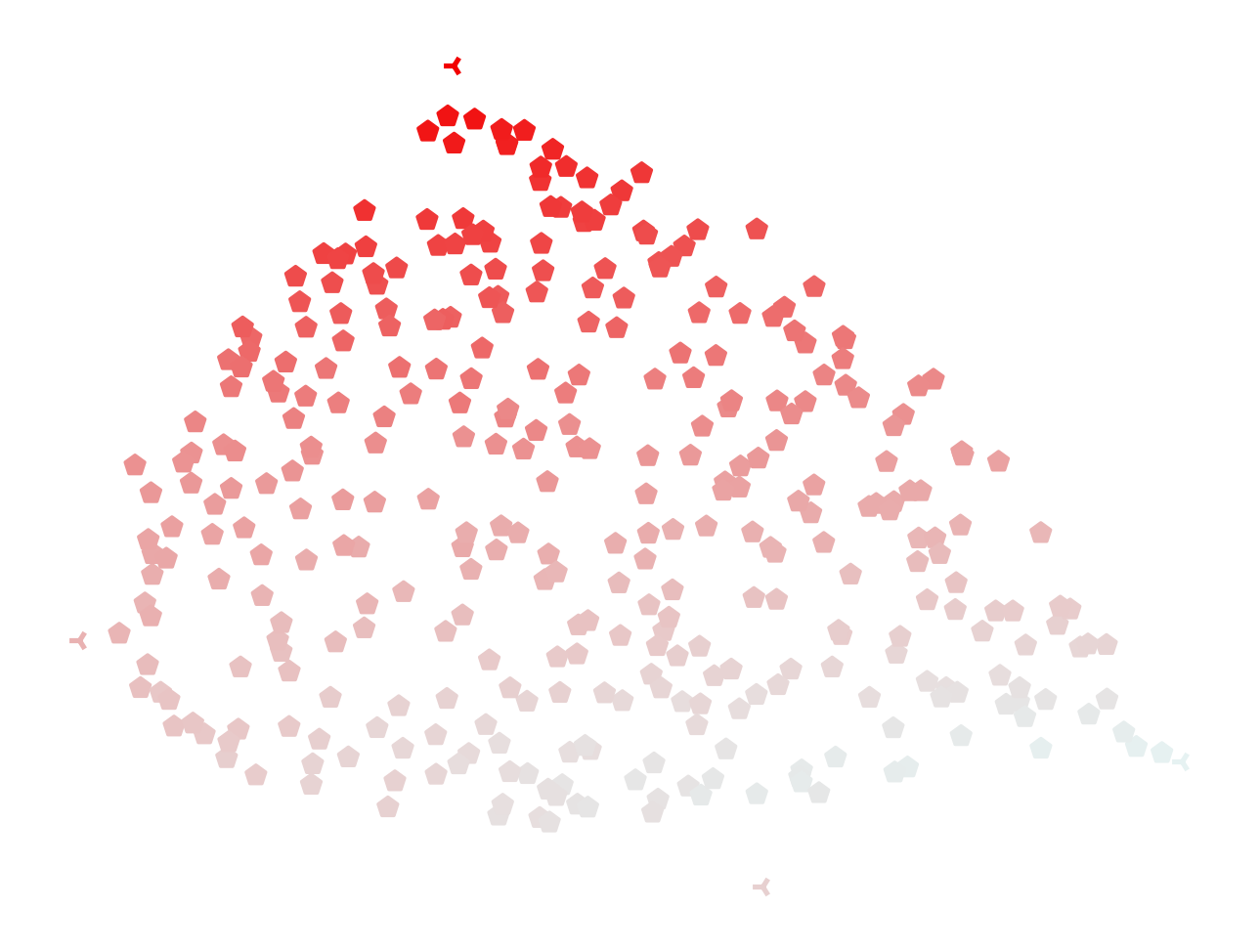}
    \end{subfigure}
    \hfill
    \begin{subfigure}[t]{0.18\textwidth}
     \centering
         \includegraphics[width=\textwidth]{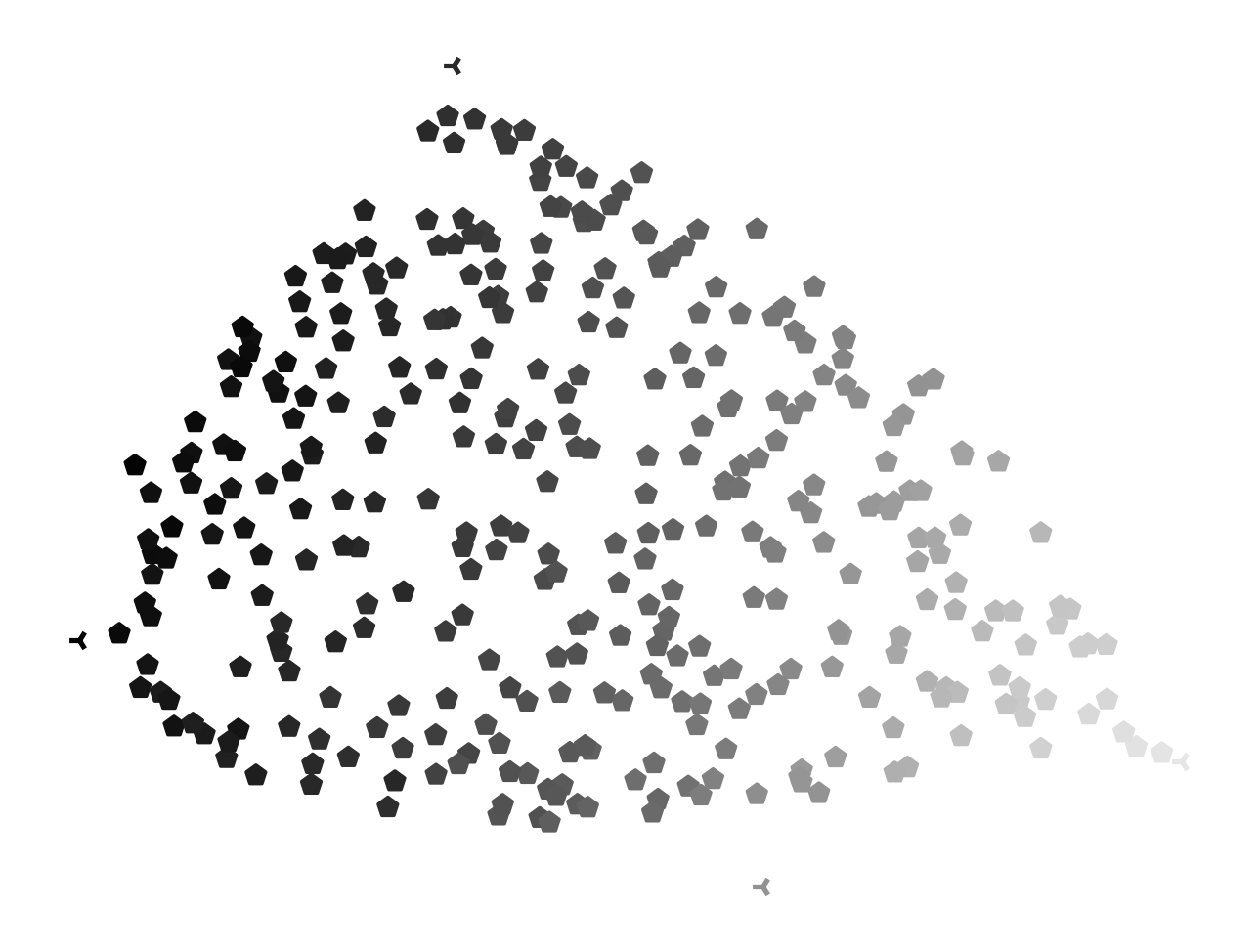}
    \end{subfigure}
    \hfill
     \begin{subfigure}[t]{0.18\textwidth}
     \centering
         \includegraphics[width=\textwidth]{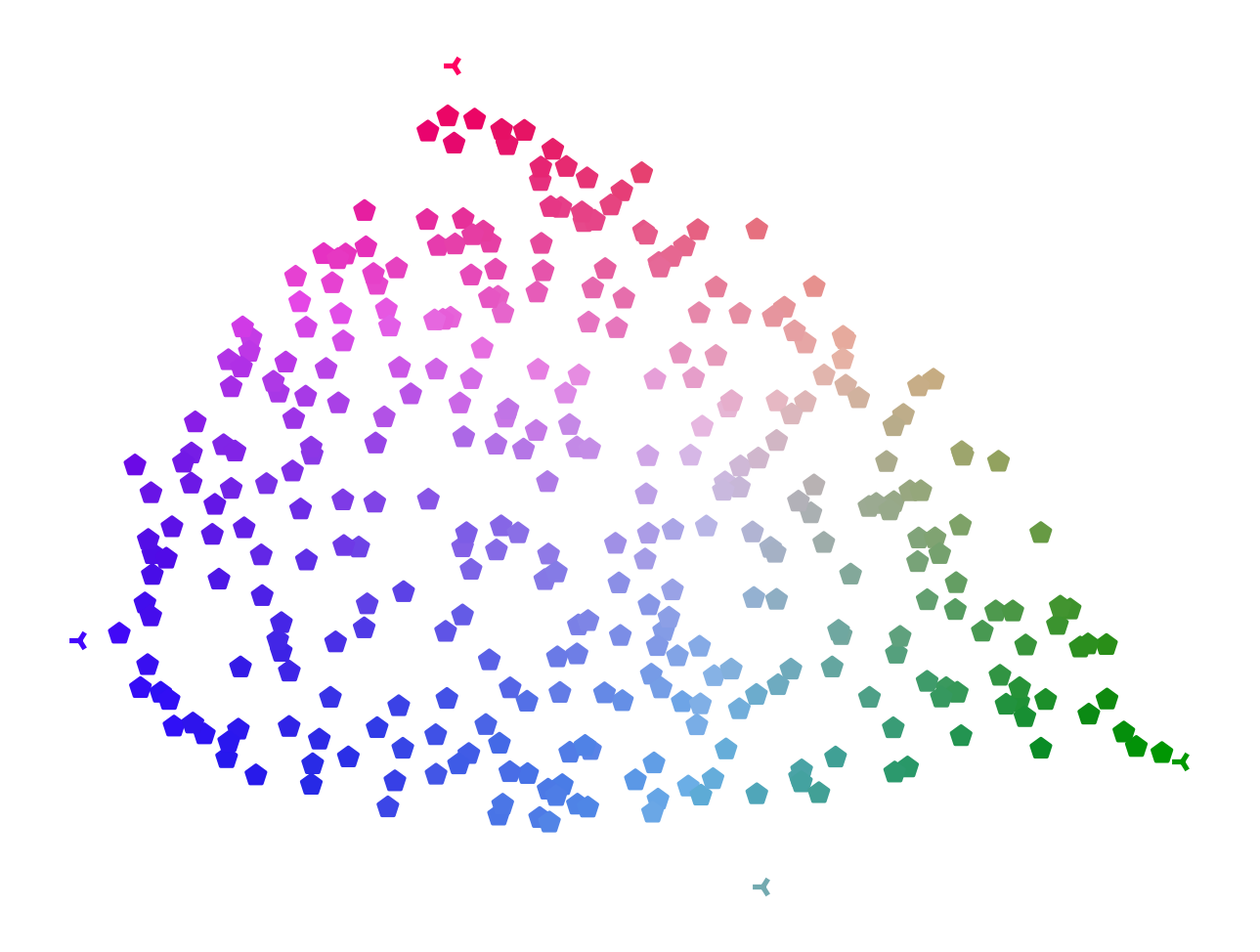}
    \end{subfigure}
\end{figure*}
\begin{figure*}[t]\ContinuedFloat
     \centering
     \begin{subfigure}[t]{0.03\textwidth}
         \centering
         \includegraphics[width=\textwidth]{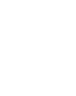}
     \end{subfigure}
     \hfill
     \begin{subfigure}[t]{0.18\textwidth}
         \centering
         \includegraphics[width=\textwidth]{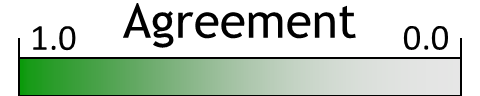}
     \end{subfigure}
     \hfill
     \begin{subfigure}[t]{0.18\textwidth}
         \centering
         \includegraphics[width=\textwidth]{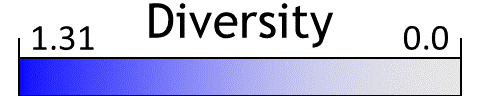}
     \end{subfigure}
     \hfill
     \begin{subfigure}[t]{0.18\textwidth}
         \centering
         \includegraphics[width=\textwidth]{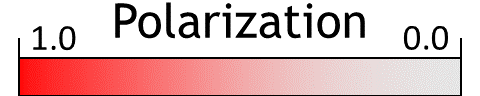}
    \end{subfigure}
    \hfill
    \begin{subfigure}[t]{0.18\textwidth}
     \centering
         \includegraphics[width=\textwidth]{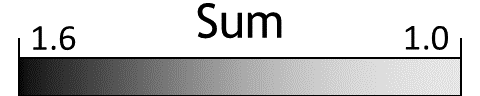}
    \end{subfigure}
    \hfill
     \begin{subfigure}[t]{0.18\textwidth}
     \centering
         \includegraphics[width=\textwidth]{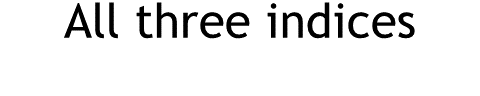}
    \end{subfigure}
    \caption{\label{fig:tenmaps}
    The maps of elections from Figs.~\ref{fig:swap-map:standard},
    \ref{fig:swap-map:extended}, and~\ref{fig:swap-map:mallows},
    where the colors of the dots denote the values of our indices.
    The rows correspond to the datasets, and the
      columns correspond to the indices. The fourth one shows the sum
      of the indices, and the fifth one a superimposition of the
      first three.}
\end{figure*}

\begin{figure*}[b]
     \centering
     \begin{subfigure}[t]{0.05\textwidth}
         \centering
         \includegraphics[width=\textwidth]{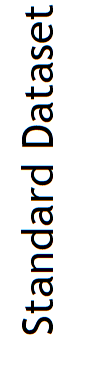}
     \end{subfigure}
     \hspace{0.23cm}
     \begin{subfigure}[t]{0.266\textwidth}
         \centering
         \includegraphics[width=\textwidth]{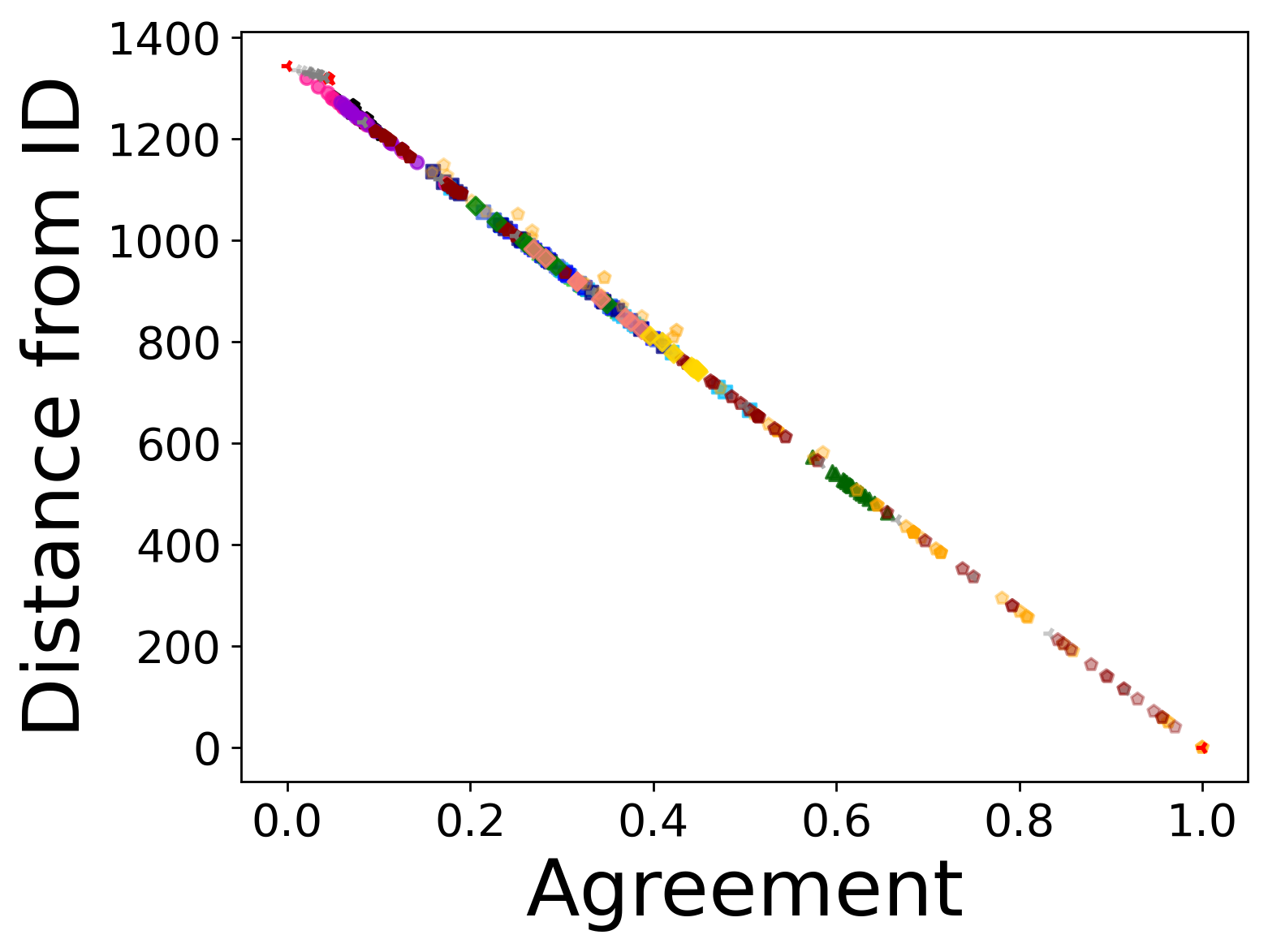}
     \end{subfigure}
     \hspace{0.23cm}
     \begin{subfigure}[t]{0.266\textwidth}
         \centering
         \includegraphics[width=\textwidth]{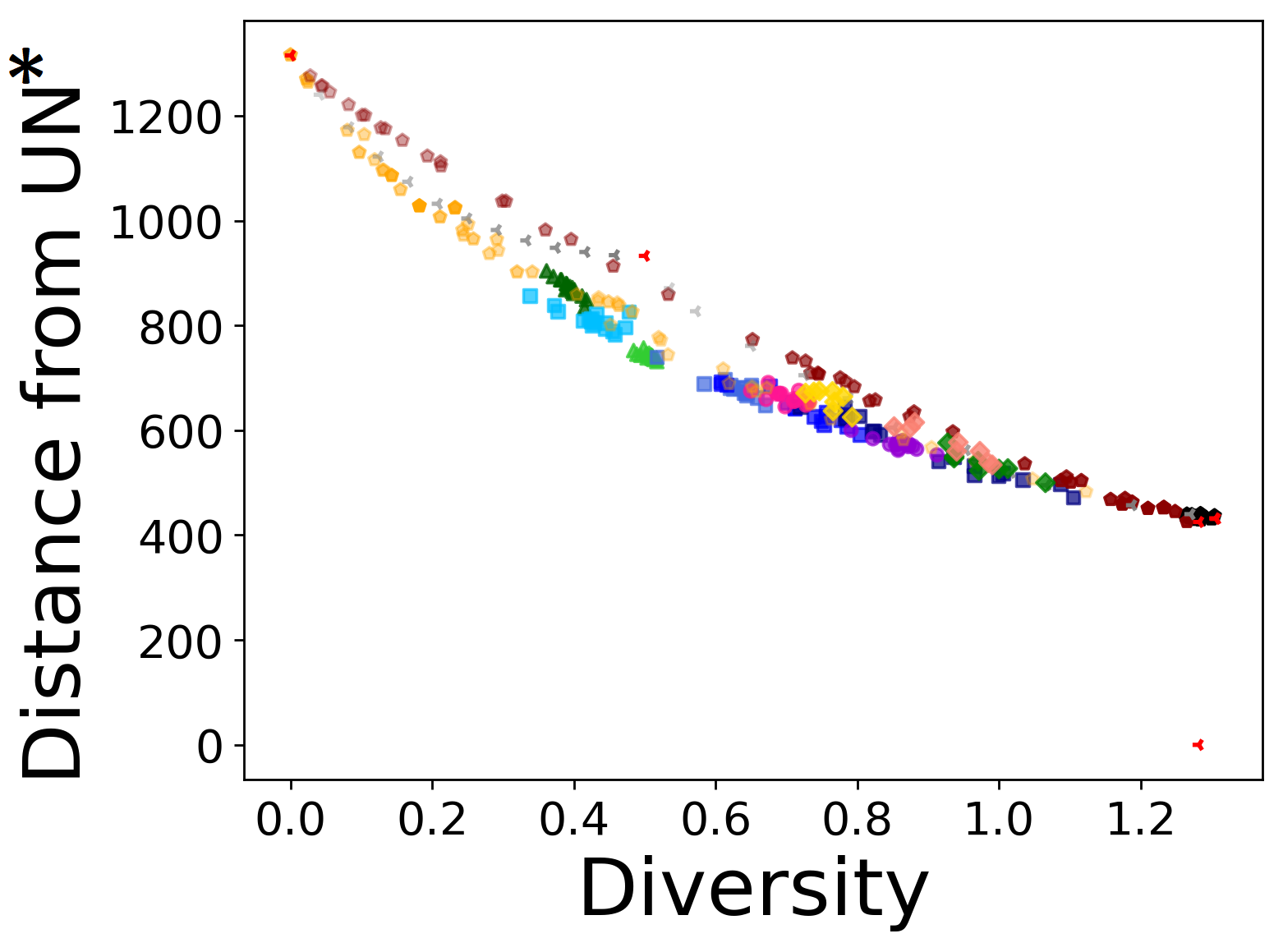}
     \end{subfigure}
     \hspace{0.23cm}
     \begin{subfigure}[t]{0.266\textwidth}
         \centering
         \includegraphics[width=\textwidth]{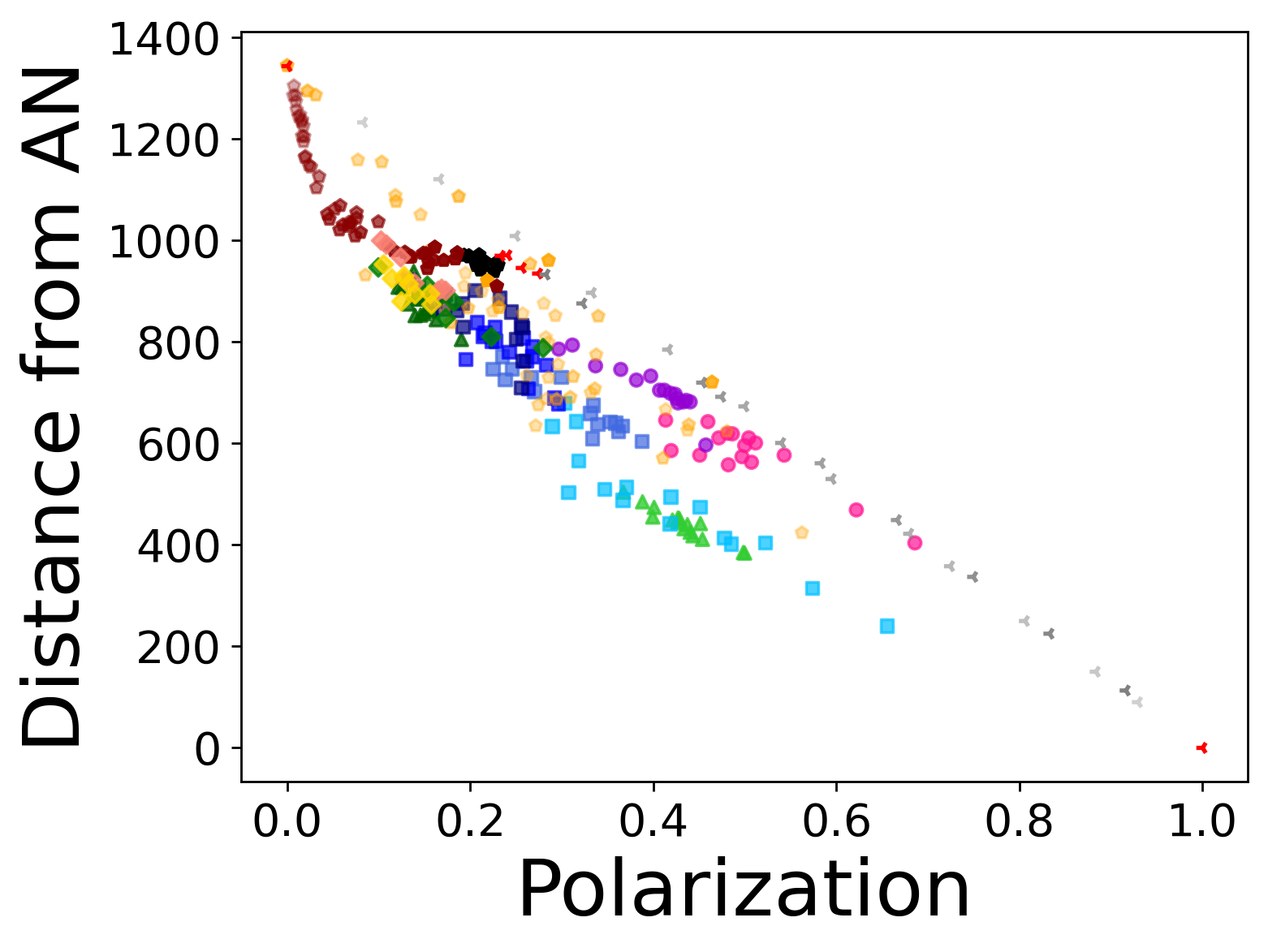}
    \end{subfigure}
\end{figure*}
\begin{figure*}[b]\ContinuedFloat
     \centering
     \begin{subfigure}[t]{0.05\textwidth}
         \centering
         \includegraphics[width=\textwidth]{img/colored_maps/extended_label.png}
     \end{subfigure}
     \hspace{0.23cm}
     \begin{subfigure}[t]{0.266\textwidth}
         \centering
         \includegraphics[width=\textwidth]{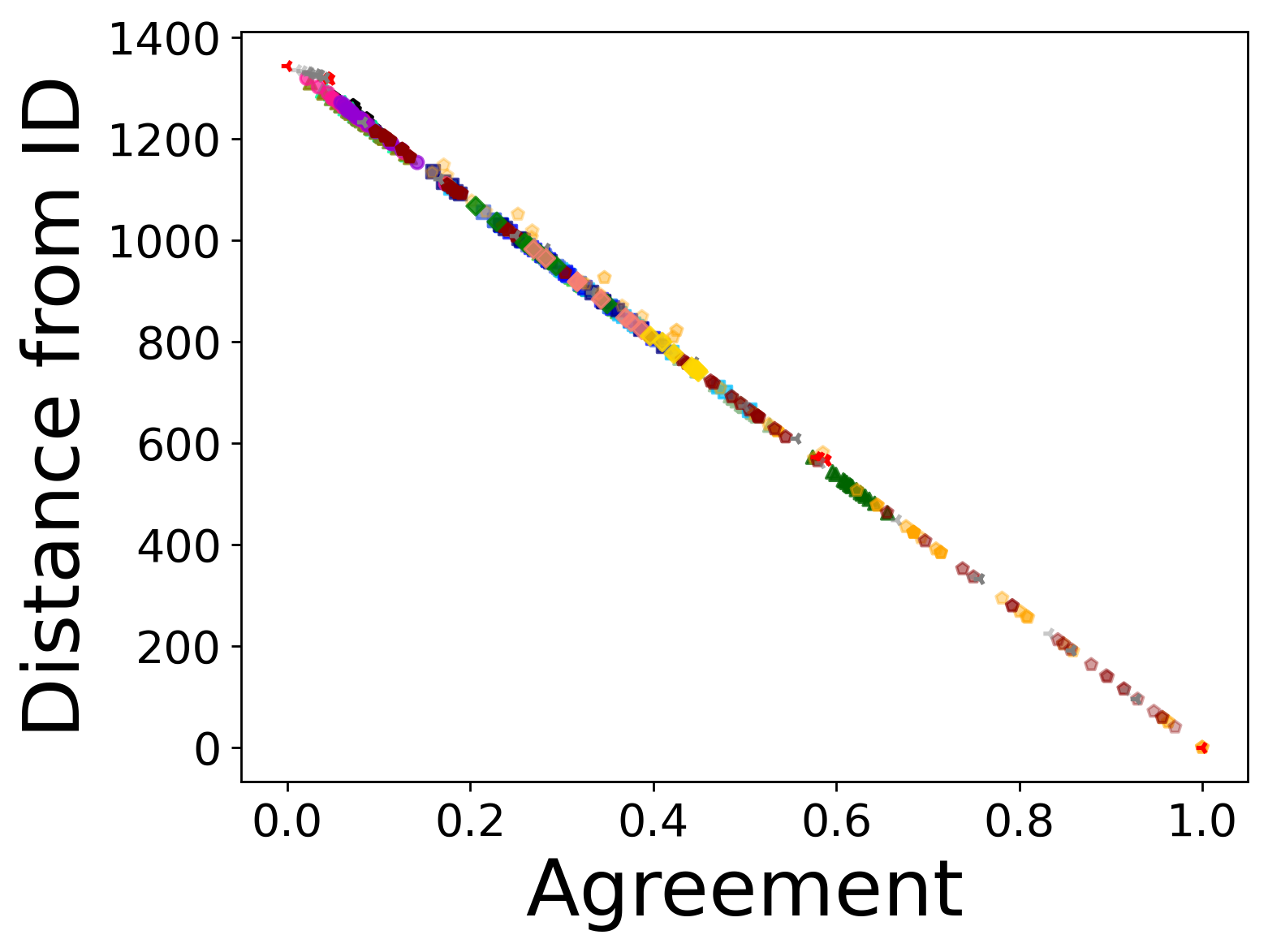}
     \end{subfigure}
     \hspace{0.23cm}
     \begin{subfigure}[t]{0.266\textwidth}
         \centering
         \includegraphics[width=\textwidth]{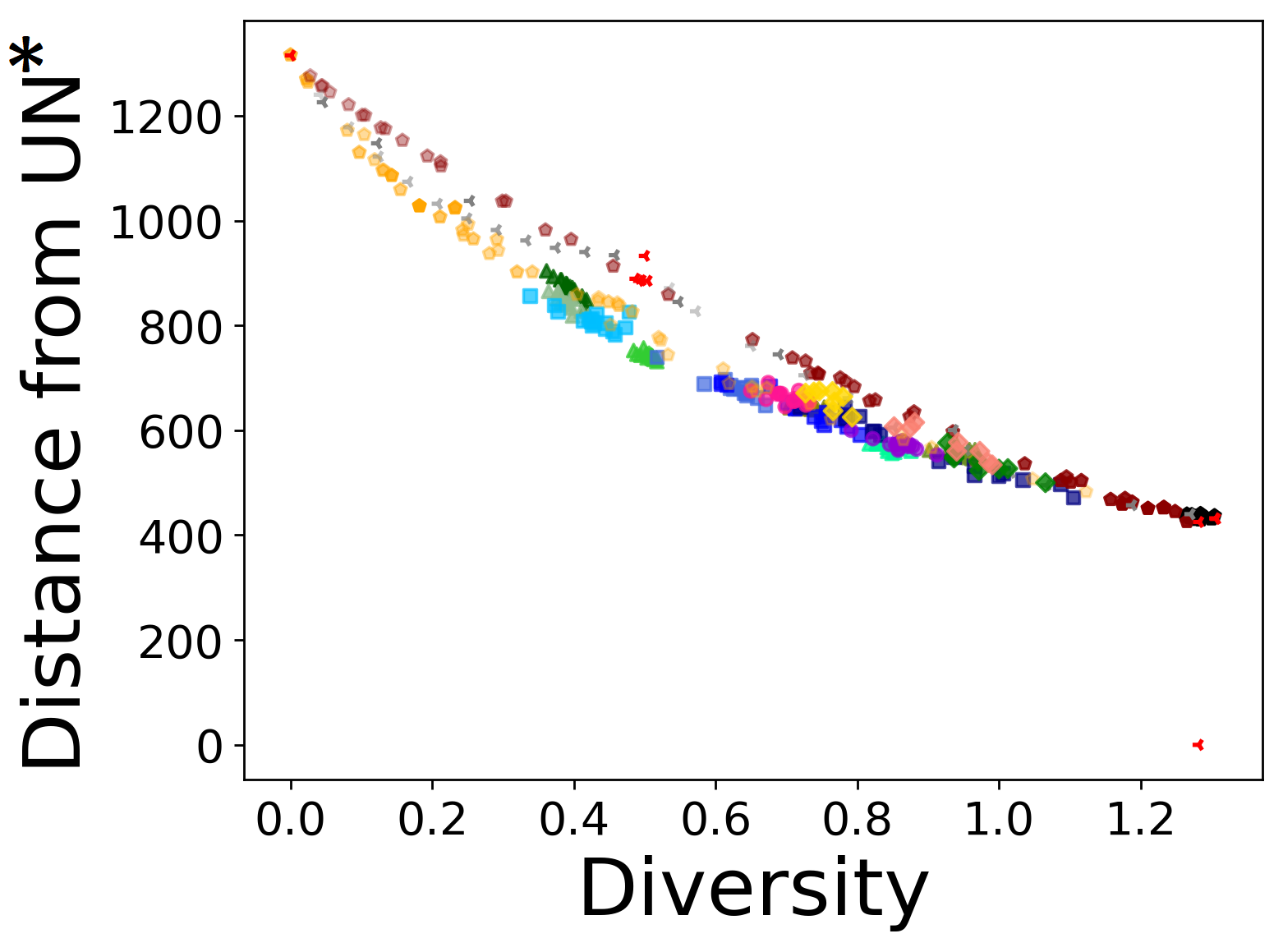}
     \end{subfigure}
     \hspace{0.23cm}
     \begin{subfigure}[t]{0.266\textwidth}
         \centering
         \includegraphics[width=\textwidth]{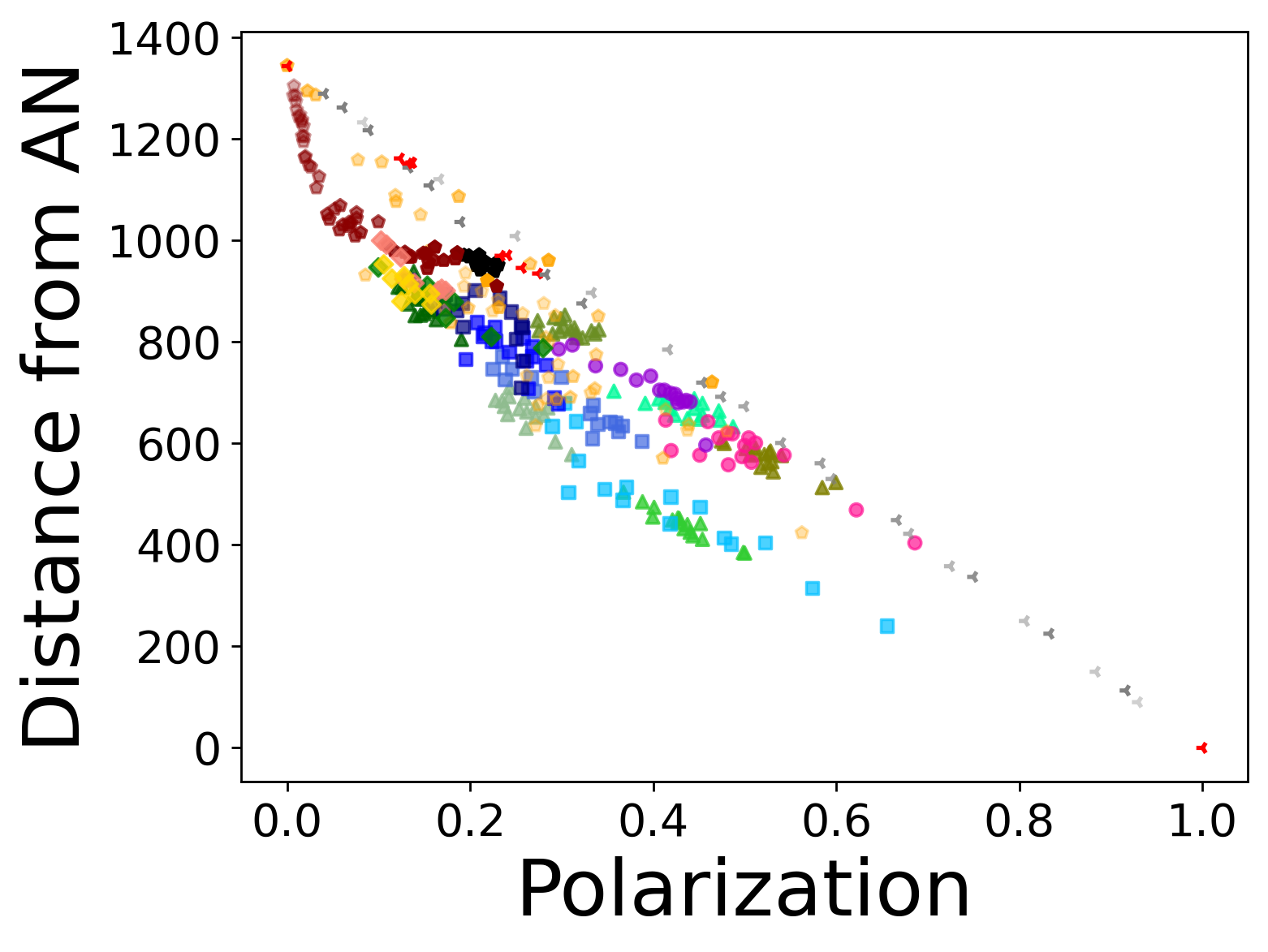}
    \end{subfigure}
\end{figure*}
\begin{figure*}[b]\ContinuedFloat
     \centering
     \begin{subfigure}[t]{0.05\textwidth}
         \centering
         \includegraphics[width=\textwidth]{img/colored_maps/mallows_label.png}
     \end{subfigure}
     \hspace{0.23cm}
     \begin{subfigure}[t]{0.266\textwidth}
         \centering
         \includegraphics[width=\textwidth]{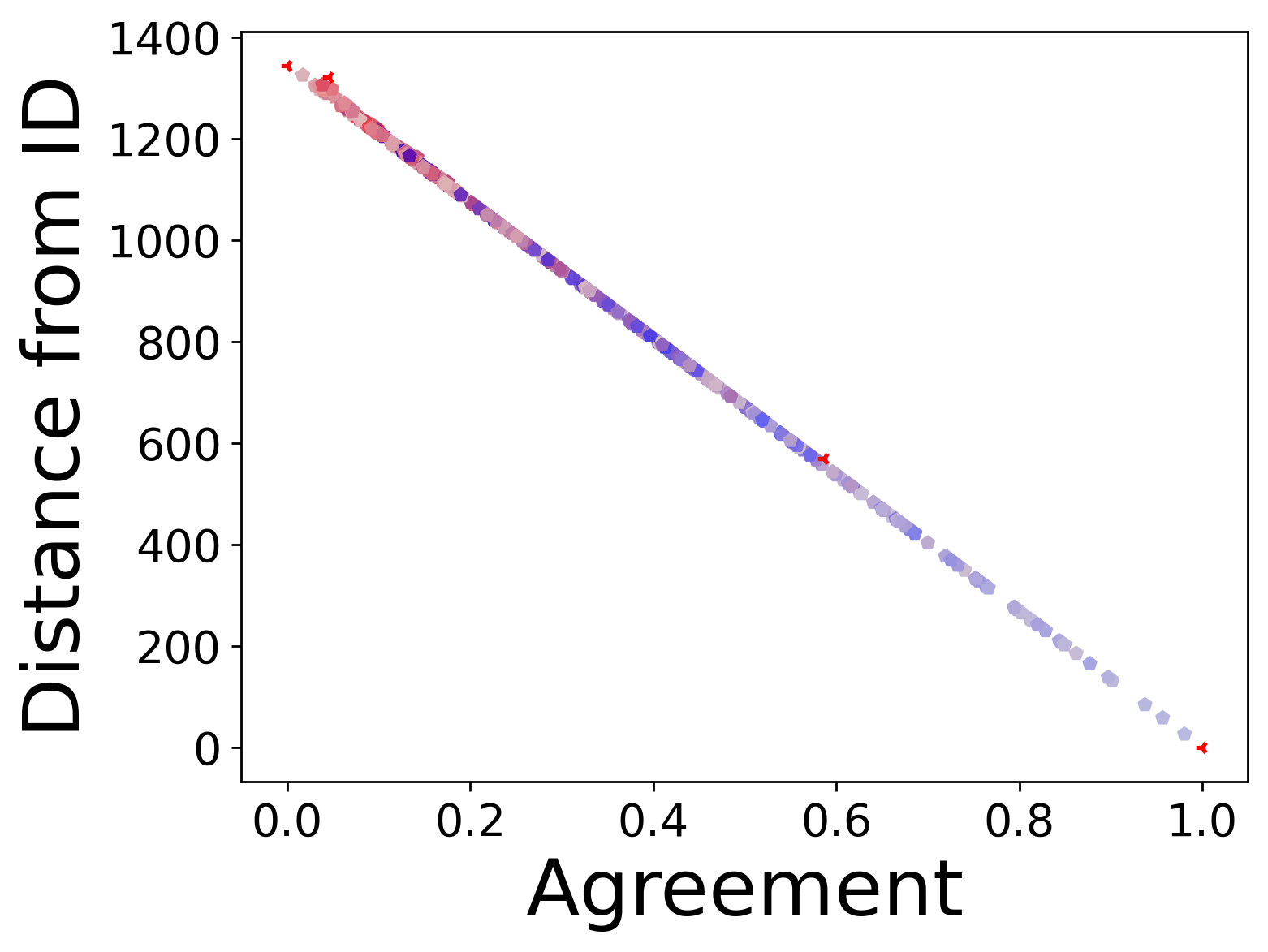}
     \end{subfigure}
     \hspace{0.23cm}
     \begin{subfigure}[t]{0.266\textwidth}
         \centering
         \includegraphics[width=\textwidth]{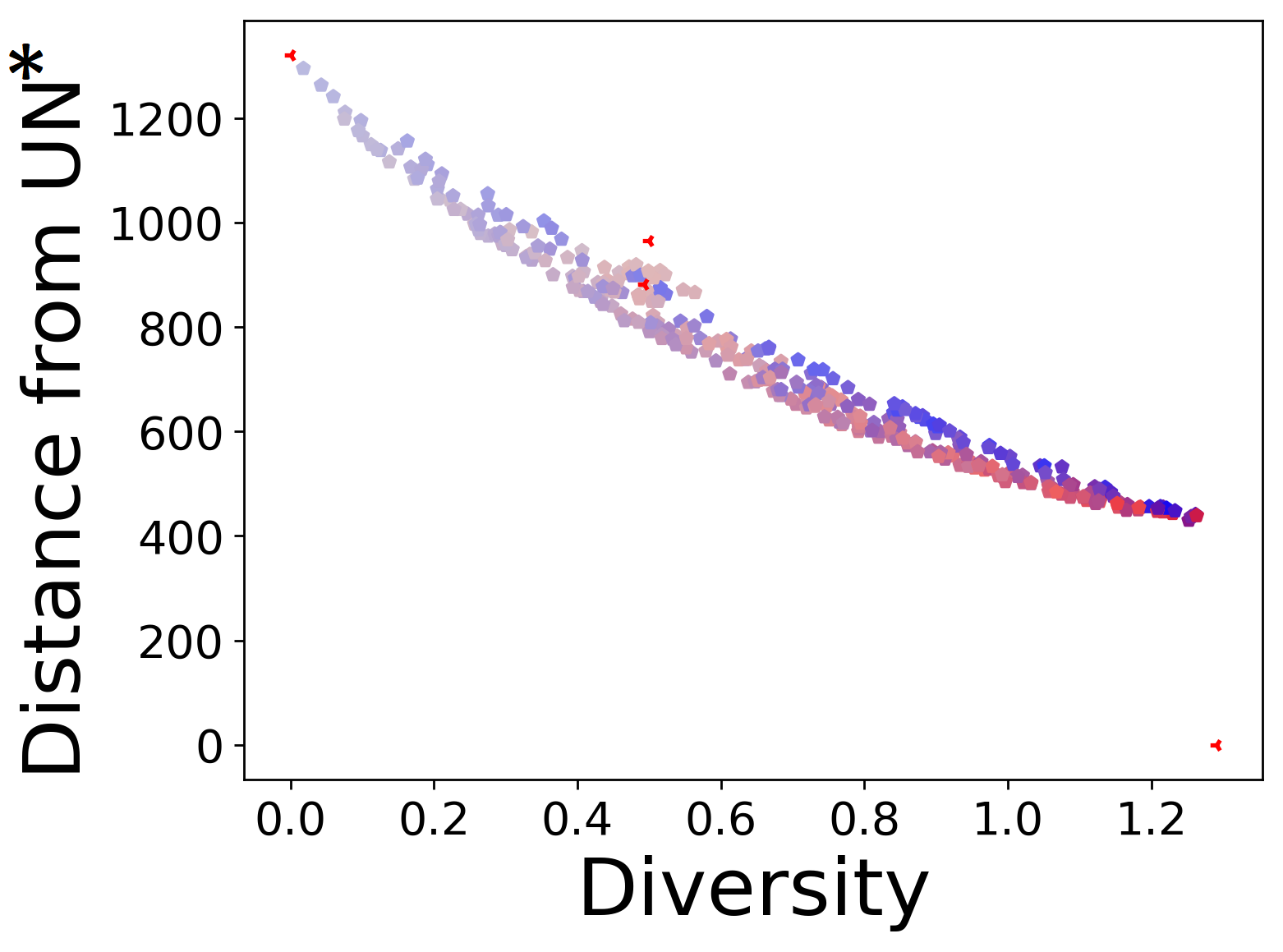}
     \end{subfigure}
     \hspace{0.23cm}
     \begin{subfigure}[t]{0.266\textwidth}
         \centering
         \includegraphics[width=\textwidth]{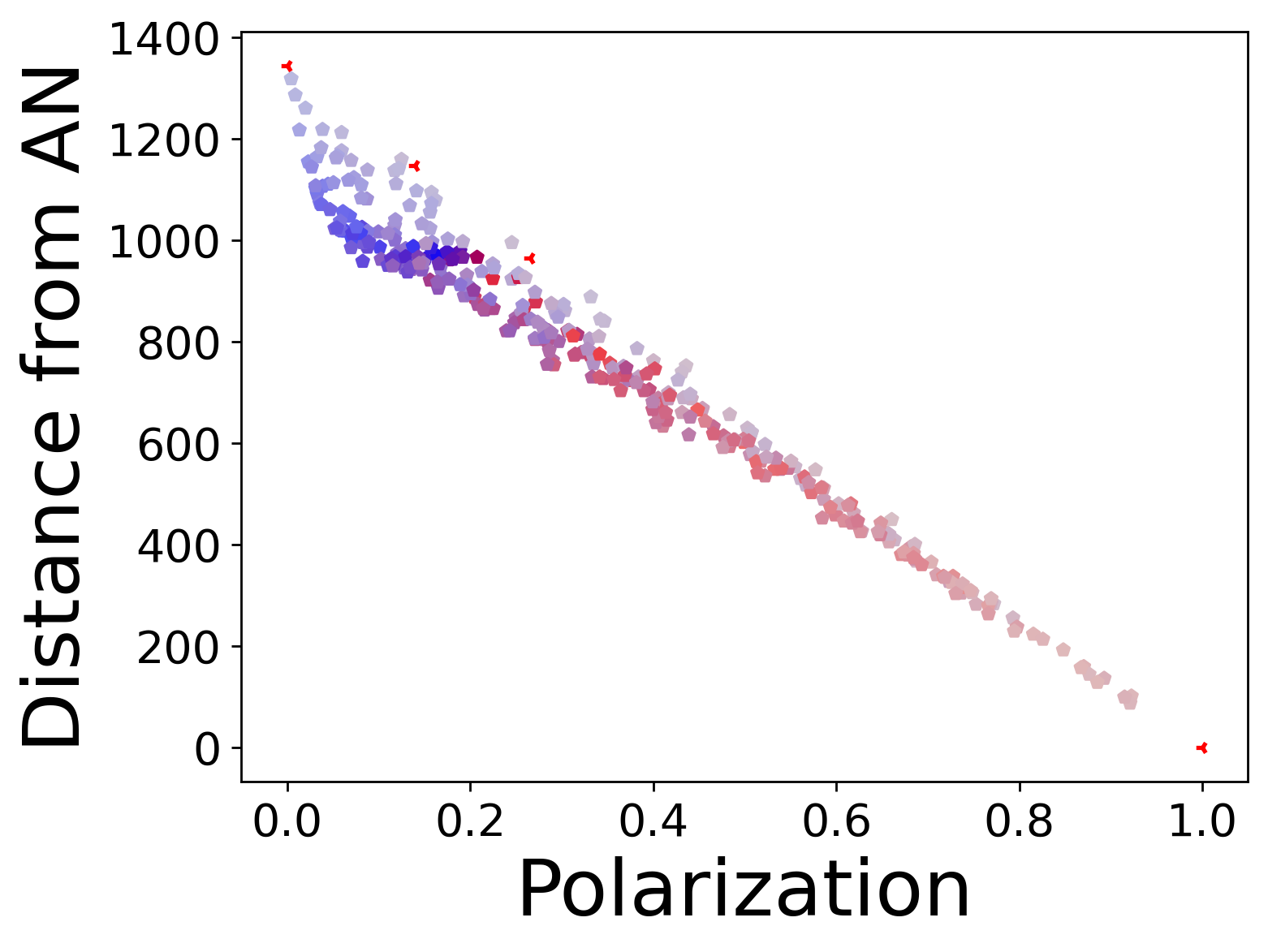}
    \end{subfigure}
    \caption{The plots showing the correlations between agreement, diversity, and polarization indices and distances from \ID, $\appUN$, and \AN, respectively.
    The rows correspond to the datasets and columns to index--distance-from-election pairs.}
    \label{fig:correlations:all}
\end{figure*}

\end{document}